\documentclass[a4paper,english,fleqn,leqno]{llncs}
\makeatletter\let\endnote\@undefined\makeatother
\usepackage[T1]{fontenc}
\usepackage[utf8]{inputenc}
\usepackage{babel}
\useshorthands{"}
\defineshorthand{"=}{\babelhyphen{hard}}
\defineshorthand{"~}{\babelhyphen{nobreak}}
\usepackage{newtxtext}
\usepackage{amsmath,amssymb,stmaryrd,mathrsfs,mathtools}
\usepackage{microtype}
\usepackage{xspace,calc,xifthen}
\usepackage[hidelinks]{hyperref}
\usepackage{csquotes}
\usepackage[capitalise,nameinlink]{cleveref}
\crefname{section}{Sect.}{Sects.}\Crefname{section}{Section}{Sections}
\crefname{figure}{Fig.}{Figs.}\Crefname{figure}{Figure}{Figures}
\crefname{table}{Tab.}{Tabs.}\Crefname{table}{Table}{Tables}
\crefname{appendix}{App.}{Apps.}\Crefname{appendix}{Appendix}{Appendices}
\crefformat{equation}{(#2#1#3)}\crefrangeformat{equation}{(#3#1#4--#5#2#6)}
\crefname{enumi}{}{}\Crefname{enumi}{}{}
\crefformat{enumi}{(#2#1#3)}\crefrangeformat{enumi}{(#3#1#4--#5#2#6)}
\crefname{definition}{Def.}{Defs.}\Crefname{definition}{Definition}{Definitions}
\crefname{lemma}{Lem.}{Lems.}\Crefname{lemma}{Lemma}{Lemmata}
\crefname{proposition}{Prop.}{Props.}\Crefname{proposition}{Proposition}{Propositions}
\crefname{theorem}{Thm.}{Thms.}\Crefname{theorem}{Theorem}{Theorems}
\crefname{example}{Ex.}{Exs.}\Crefname{example}{Example}{Examples}
\usepackage{enumitem}
\setlist[itemize]{leftmargin=*, label={--}}
\setlist[enumerate]{leftmargin=*}
\usepackage{graphics,graphicx}
\usepackage{tikz}
\usetikzlibrary{automata,decorations.pathmorphing,matrix,shapes,arrows,calc}
\usepackage{xstring,xifthen}

\usepackage{listings}
\lstdefinelanguage{CASL}{%
  keywords={spec, then, end,
            from, get,
            sort, sorts,
            op, ops,
            pred, preds,
            free, type,
            axiom
  }
}
\lstdefinelanguage{UMLState}{%
  keywords = {%
    logic, spec, end, var, event, states, init, trans,
  }
}
\newcommand*{\axdot}{\boldsymbol{\cdot}}
\crefname{lstlisting}{Lst.}{Lsts.}\Crefname{lstlisting}{Listing}{Listings}

\usepackage{hetcasl}

\usepackage{algorithm}
\usepackage[noend]{algpseudocode}
\algnewcommand\algorithmicrequirx{\phantom{\textbf{Require:}}}
\algnewcommand\Requirx{\item[\algorithmicrequirx]}
\algnewcommand\algorithmicbreak{\textbf{break}}
\algnewcommand\algorithmicchoose{\textbf{choose}}
\algnewcommand\Break{\algorithmicbreak{}}
\algnewcommand\Choose{\algorithmicchoose{}}
\algrenewcommand\alglinenumber[1]{\tiny#1}
\def\Statex{\item[]\vspace*{-2.4ex}}
\crefname{algorithm}{Alg.}{Algs.}\Crefname{algorithm}{Algorithm}{Algorithms}
\floatname{algorithm}{Alg.}

\usepackage{comment}
\usepackage[framemethod=TikZ]{mdframed}
\mdfsetup{skipabove=\topskip,skipbelow=\topskip}
\global\mdfdefinestyle{mystyle}{
  backgroundcolor=gray!10,
  linewidth=1pt,
  rightline=false,
  leftline=false,
  topline=false,
  bottomline=false,
  leftmargin=0pt,
  rightmargin=0pt,
  skipabove=0pt,
  skipbelow=0pt,
  innertopmargin=0pt,
  innerbottommargin=0pt,
  innerleftmargin=0pt,
  innerrightmargin=0pt,
}
\specialcomment{techreport}{%
  \begin{mdframed}[style=mystyle]%
}{
  \end{mdframed}%
}

\usepackage{pdfcomment}
\newif\ifshowednotes\showednotestrue
\newcommand*{\ednoteauthor}{EdNote}
\newcommand*{\ednotecomment}{No comment.}
\newcommand*{\myenotezwritemark}[1]{\leavevmode\marginpar{\pdftooltip{\footnotesize\ednoteauthor(#1)}{\ednotecomment}}}
\usepackage{enotez}
\setenotez{backref=true}
\setenotez{list-name={EdNotes}}
\setenotez{mark-cs={\myenotezwritemark}}
\newcommand{\ednote}[2][Ednote]{%
  \ifshowednotes%
    \renewcommand*{\ednoteauthor}{#1}%
    \renewcommand*{\ednotecomment}{#2}%
    \bgroup%
      \fontsize{6pt}{6pt}\selectfont%
      \endnote{\ifthenelse{\equal{#1}{Ednote}}{#2}{#1: #2}}%
    \egroup%
  \fi%
}
\ifshowednotes%
  \AtEndDocument{\printendnotes}%
\fi
\usepackage[%
  final, 
]{showkeys}

\allowdisplaybreaks

\newcommand*{\NZ}{\mathbb{N}}

\newcommand*{\compfun}{\mathbin{\circ}}

\newcommand*{\isorel}{\mathrel{\cong}}

\newcommand*{\injto}{\mathrel{\rightarrowtail}}

\newcommand*{\truefrm}{\mathrm{true}}
\newcommand*{\falsefrm}{\mathrm{false}}

\newcommand*{\reductop}{\mathnormal{|}}
\newcommand*{\reduct}[2]{#1\reductop#2}
\newcommand*{\idop}{1}
\newcommand*{\id}[1]{\idop_{#1}}
\newcommand*{\opop}{\mathrm{op}}
\newcommand*{\op}[1]{{#1}^{\opop}}
\newcommand*{\natto}{\mathrel{\dot{\mathnormal{\to}}}}

\newcommand{\just}[2][\Leftrightarrow]{\\ & \llap{$#1$\ }{} \quad\{\;\text{\footnotesize #2}\;\}\\ &}
\newcommand{\stacked}[1]{\begin{array}[t]{@{}l@{}}#1\end{array}}

\mathchardef\cln\mathcode`\:
\mathcode`\:=\string"8000
\begingroup \catcode`\:=\active
  \gdef:{\nobreak\mskip2mu\mathpunct{}\nonscript
    \mkern-\thinmuskip{\cln}\mskip6muplus1mu\relax}
\endgroup

\newcommand{\splitatcommas}[1]{%
  \begingroup
  \ifnum\mathcode`,="8000
  \else
    \begingroup\lccode`~=`, \lowercase{\endgroup
      \edef~{\mathchar\the\mathcode`, \penalty0 \noexpand\hspace{0pt plus 1em}}%
    }\mathcode`,="8000
  \fi
  #1%
  \endgroup
}

\newcommand*{\CASL}{\ensuremath{\text{\normalfont\textsc{Casl}}}\xspace}
\newcommand*{\SPASS}{\ensuremath{\text{\normalfont\textsc{Spass}}}\xspace}
\newcommand*{\HeTS}{\ensuremath{\text{\normalfont\textsc{HeTS}}}\xspace}
\newcommand*{\UMLState}{\ensuremath{\text{\normalfont\textsc{UMLState}}}\xspace}

\newcommand*{\nfatslash}{%
  \mathrel{\mathnormal{\mathchoice%
    {\mkern-8mu\fatslash\mkern0mu}%
    {\mkern-8mu\fatslash\mkern0mu}%
    {\mkern-5mu\fatslash\mkern2mu}%
    {\mkern-4mu\fatslash\mkern1mu}%
  }}%
}
\newcommand*{\lab}[2]{#1\nfatslash#2}
\newcommand*{\labpre}[3]{#1\cln#2\nfatslash#3}

\newcommand*{\hybindop}{\mathnormal{\downarrow}}
\newcommand{\hybind}[2]{\hybindop #1 \,.\, #2}
\newcommand*{\hyatop}[1][]{\mathnormal{@}\ifthenelse{\isempty{#1}}{}{^{#1}}}
\newcommand{\hyat}[3][]{(\hyatop[#1]#2)#3}

\newcommand*{\dldiaop}[1]{\langle#1\rangle}
\newcommand{\dldia}[3]{\dldiaop{\lab{#1}{#2}}{#3}}
\newcommand*{\dlboxop}[1]{[#1]}
\newcommand{\dlbox}[3]{\dlboxop{\lab{#1}{#2}}{#3}}
\newcommand*{\dlwpopen}{\mathopen{{\langle}\mkern-3.6mu{|}}}
\newcommand*{\dlwpclose}{\mathclose{{|}\mkern-3.6mu{\rangle}}}
\newcommand*{\dlwpop}[1]{\dlwpopen#1\dlwpclose}
\newcommand{\dlwp}[4]{\dlwpop{\labpre{#1}{#2}{#3}}#4}
\newcommand*{\mlboxop}[1][]{\Box\ifthenelse{\isempty{#1}}{}{^{#1}}}
\newcommand{\mlbox}[2][]{\mlboxop[#1]{#2}}
\newcommand*{\mldiaop}[1][]{\Diamond\ifthenelse{\isempty{#1}}{}{^{#1}}}
\newcommand{\mldia}[2][]{\mldiaop[#1]{#2}}

\newcommand{\category}[1]{\ensuremath{\mathrm{#1}}}

\newcommand*{\Sig}[1][]{\mathbb{S}\ifthenelse{\isempty{#1}}{}{^{#1}}}
\newcommand*{\Sen}[1][]{\mathrm{Sen}\ifthenelse{\isempty{#1}}{}{^{#1}}}
\newcommand*{\Str}[1][]{\mathit{Str}\ifthenelse{\isempty{#1}}{}{^{#1}}}
\let\texmodels\models
\renewcommand*{\models}[1][]{\texmodels\ifthenelse{\isempty{#1}}{}{^{#1}}}
\newcommand*{\Fm}{\mathscr{F}}
\newcommand*{\Frm}[3][]{\Fm^{#1}_{#2, #3}}
\newcommand*{\Mod}[1][]{\mathrm{Mod}\ifthenelse{\isempty{#1}}{}{^{#1}}}
\newcommand*{\Pres}[1][]{\mathrm{Pres}\ifthenelse{\isempty{#1}}{}{^{#1}}}
\newcommand*{\Sign}{\mathit{Sig}}

\newcommand*{\Data}{\mathit{Dt}}
\newcommand*{\data}{\mathit{dt}}
\newcommand*{\DATA}{\mathcal{D}}

\newcommand*{\Attr}{A}
\newcommand*{\DataSt}{\Omega}

\newcommand*{\DataTr}{\DataSt^2}

\newcommand*{\HDL}{\mathcal{D}^{\downarrow}}
\newcommand*{\EDHDL}{\mathcal{E}^{\downarrow}}
\newcommand*{\EDHML}{\mathcal{M}^{\downarrow}_{\DATA}}
\newcommand*{\EDSign}{\Sigma}
\newcommand*{\Evt}{E}
\newcommand*{\Conf}{\Gamma}
\newcommand*{\CtrlSt}{C}
\newcommand*{\Trans}{T}
\newcommand*{\Rel}{R}

\newcommand*{\ctrlSt}{c}
\newcommand*{\dataSt}{\omega}
\newcommand*{\datapred}{\varphi}

\newcommand*{\evt}[1]{\mathsf{#1}}
\newcommand*{\attr}[1]{\mathsf{#1}}
\newcommand*{\state}[1]{\mathit{#1}}

\includecomment{techreport}
\title{%
  Institution-based Encoding and Verification of\\
  Simple UML State Machines in CASL/SPASS
}

\author{%
  Tobias Rosenberger\inst{1,2}
\and
  Saddek Bensalem\inst{2}
\and\\
  Alexander Knapp\inst{3}
\and
  Markus Roggenbach\inst{1}
}
\institute{
  Swansea University, U.K.\\
  \email{$\{$t.rosenberger.971978$,$ m.roggenbach$\}$@swansea.ac.uk}\\[.5ex]
\and
  Université Grenoble Alpes, France\\
  \email{Saddek.Bensalem@imag.fr}\\[.5ex]
\and 
  Universität Augsburg, Germany\\
  \email{knapp@informatik.uni-augsburg.de}
}

\begin{document}

\maketitle

\begin{abstract}
We present a new approach on how to provide institution-based semantics for UML
state machines. Rather than capturing UML state machines directly as an
institution, we build up a new logical framework $\EDHML$ into which UML state
machines can be embedded. A theoroidal comorphism maps $\EDHML$ into the \CASL
institution.  This allows for symbolic reasoning on UML state machines. By
utilising the heterogeneous toolset \HeTS that supports \CASL, a broad range of
verification tools, including the automatic theorem prover \SPASS, can be
combined in the analysis of a single state machine.
\end{abstract}

\section{Introduction}

As part of a longstanding line of research
\cite{knapp-mossakowski-roggenbach:wirsing-festschrift:2015,knapp-et-al:fase:2015,rosenberger:master:2017,knapp-mossakowski:calco:2017},
we set out on a general programme to bring together multi-view system
specification with UML diagrams and heterogeneous specification and verification
based on institution theory, giving the different system views both a joint
semantics and richer tool support.

Institutions, a formal notion of a logic, are a principled way of
creating such joint semantics. They make moderate assumptions about
the data constituting a logic, give uniform notions of well-behaved
translations between logics and, given a graph of such translations,
automatically give rise to a joint institution.

In this paper, we will focus on UML state machines, which are an
object-based variant of Harel statecharts. Within the UML, state
machines are a central means to specify system behaviour. Here, we
capture simple UML state machines in what we claim to be a true
semantical sense. Focus of this paper are state machines running in
isolation --- interacting state machines and with it the notion of the
event pool are left to future work.

Compared to our previous attempts to institutionalise state machines
\cite{knapp-mossakowski-roggenbach:wirsing-festschrift:2015,knapp-et-al:fase:2015,rosenberger:master:2017,knapp-mossakowski:calco:2017},
this paper takes a different approach. Rather than capturing UML state
machines directly as an institution, we build up a new logical
framework $\EDHML$ in which UML state machines can be embedded. Core
of this framework is a new hybrid modal logic which allows us to
logically encode the \emph{presence} as well as the \emph{absence} of
transitions in the state machines. Data types, guards, and effects of
events are specified in the algebraic specification language \CASL. An
algorithm translates UML state machines into $\EDHML$.

A theoroidal comorphism maps our logical framework $\EDHML$ into the \CASL
institution. This allows to us to utilise the heterogeneous toolset
\HeTS~\cite{mossakowski-maeder-luettich:tacas:2007} and its connected provers
for analysing UML state machines. In this paper we demonstrate how to analyse a
state machine with the automatic first-order prover
\SPASS~\cite{weidenbach-et-al:cade:2009}, which is the default automated prover
of \HeTS.  Such symbolic reasoning can be of advantage as, in principle, it
allows to verify properties of UML state machines with large or infinite state
spaces. Such machines appear routinely in system modelling: though state
machines usually have only finitely many control states, they have a large
number of configurations, or even infinitely many, due to the data variables
involved.

Compared to other symbolic approaches to directly encode UML state machines into
a specific interactive theorem prover
\cite{kyas-et-al:sfedl:2004,groenniger:phd:2010,balser-et-al:icfem:2004}, our logical
framework $\EDHML$ provides first an institutional semantics that is tool
independent. Only in a second step, we translate $\EDHML$ into \CASL. Via \HeTS,
this opens access to a broad range of analysis tools, including SAT solvers,
automatic first-order theorem provers, automated and interactive higher-order
theorem provers, which all can be combined in the analysis of state machines.

This paper is organised as follows: First we provide some background
on institutions, including the \CASL institution in
\cref{sec:background}. Then we discuss simple UML state machines, how
to capture their events, attributes, and transitions, and what their
models are. In \cref{sec:edhml} we define a new hybrid, modal logic
for specifying UML state machine transitions. \Cref{sec:comorphism}
provides the translation into the \CASL institution. In
\cref{sec:proving}, we finally demonstrate the symbolic analysis of a
simple UML state machine as enabled by the previous constructions.  We
conclude in \cref{sec:conclusions} with an outlook to future work. 

\section{Background on Institutions}\label{sec:background}

We briefly recall the basic definitions of institutions and theoroidal
institution comorphisms as well as the algebraic specification language \CASL.
Subsequently we will develop an institutional frame for capturing simple UML
state machines and present a theoroidal institution comorphism from this frame
into \CASL.

\subsection{Institutions and Theoroidal Institution Comorphisms}

Institutions are an abstract formalisation of the notion of logical systems
combining signatures, structures, sentences, and satisfaction under the slogan
``truth is invariant under change of notation''.  Institutions can be related in
different ways by institution (forward) (co-)morphisms, where a so-called
theoroidal institution comorphism covers a particular case of encoding a
``poorer'' logic into a ``richer'' one.

Formally~\cite{goguen-burstall:acm:1992}, an institution $\mathcal{I} =
(\Sig[\mathcal{I}], \Str[\mathcal{I}], \Sen[\mathcal{I}],
{\models[\mathcal{I}]})$ consists of (i)~a category of \emph{signatures}
$\Sig[\mathcal{I}]$; (ii)~a contravariant \emph{structures functor}
$\Str[\mathcal{I}] : \op{(\Sig[\mathcal{I}])} \to \category{Cat}$, where
$\category{Cat}$ is the category of (small) categories; (iii)~a \emph{sentence
  functor} $\Sen[\mathcal{I}] : \Sig[\mathcal{I}] \to \category{Set}$, where
$\category{Set}$ is the category of sets; and (iv)~a family of
\emph{satisfaction relations} ${\models[\mathcal{I}]_{\Sigma}} \subseteq
|\Str[\mathcal{I}](\Sigma)| \times \Sen[\mathcal{I}](\Sigma)$ indexed over
$\Sigma \in |\Sig[\mathcal{I}]|$, such that the following \emph{satisfaction
  condition} holds for all $\sigma : \Sigma \to \Sigma'$ in $\Sig[\mathcal{I}]$,
$\varphi \in \Sen[\mathcal{I}](\Sigma)$, and $M' \in
|\Str[\mathcal{I}](\Sigma')|$:
\begin{equation*}
  \Str[\mathcal{I}](\sigma)(M') \models[\mathcal{I}]_{\Sigma} \varphi
\ \iff\ 
  M' \models[\mathcal{I}]_{\Sigma'} \Sen[\mathcal{I}](\sigma)(\varphi)
\ \text{.}
\end{equation*}
$\Str[\mathcal{I}](\sigma)$ is called the \emph{reduct} functor,
$\Sen[\mathcal{I}](\sigma)$ the \emph{translation} function.

A \emph{theory presentation} $T = (\Sigma, \Phi)$ in the institution
$\mathcal{I}$ consists of a signature $\Sigma \in |\Sig[\mathcal{I}]|$, also
denoted by $\Sign(T)$, and a set of sentences $\Phi \subseteq
\Sen[\mathcal{I}](\Sigma)$.  Its \emph{model class} $\Mod[\mathcal{I}](T)$ is
the class $\{ M \in \Str[\mathcal{I}](\Sigma) \mid M \models[\mathcal{I}]_{\Sigma} \varphi \text{ f.\,a.\ } \varphi \in \Phi \}$ of the $\Sigma$-structures
satisfying the sentences in $\Phi$.  A \emph{theory presentation morphism}
$\sigma : (\Sigma, \Phi) \to (\Sigma', \Phi')$ is given by a signature morphism
$\sigma : \Sigma \to \Sigma'$ such that $M' \models[\mathcal{I}]_{\Sigma'}
\Sen[\mathcal{I}](\sigma)(\varphi)$ for all $\varphi \in \Phi$ and $M' \in
\Mod[\mathcal{I}](\Sigma', \Phi')$.  Theory presentations in $\mathcal{I}$ and
their morphisms form the category $\Pres[\mathcal{I}]$.

A \emph{theoroidal institution comorphism} $\nu = (\nu^{\Sig},\allowbreak
\mu^{\Mod},\allowbreak \nu^{\Sen}) : \mathcal{I} \to \mathcal{I}'$ consists of a
functor $\nu^{\Sig} : \Sig[\mathcal{I}] \to \Pres[\mathcal{I}']$ inducing the
functor $\nu^{\Sign} = \nu^{\Sig}; \Sign : \Sig[\mathcal{I}] \to
\Sig[\mathcal{I}']$ on signatures, a natural transformation $\nu^{\Mod} :
\op{(\nu^{\Sig})}; \Mod[\mathcal{I}'] \natto \Str[\mathcal{I}]$ on structures,
and a natural transformation $\nu^{\Sen} : \Sen[\mathcal{I}] \natto \nu^{\Sign};
\Sen[\mathcal{I}']$ on sentences, such that for all $\Sigma \in
|\Sig[\mathcal{I}]|$, $M' \in |\Mod[\mathcal{I}'](\nu^{\Sig}(\Sigma))|$, and
$\varphi \in \Sen[\mathcal{I}](\Sigma)$ the following \emph{satisfaction
  condition} holds:
\begin{equation*}
  \nu^{\Mod}_{\Sigma}(M') \models[\mathcal{I}]_{\Sigma} \varphi
\iff
  M' \models[\mathcal{I'}]_{\nu^{\Sign}(\Sigma)} \nu^{\Sen}(\Sigma)(\varphi)
\ \text{.}
\end{equation*}

\subsection{\CASL and the Institution CFOL$^=$}

The algebraic specification language \CASL~\cite{mosses:2004} offers
several specification levels: \emph{Basic specifications} essentially
list signature declarations and axioms, thus determining a category of
first-order structures.  \emph{Structured specifications} serve to
combine such basic specifications into larger specifications in a
hierarchical and modular fashion.  Of the many logics available in
\CASL, we will work with the institution CFOL$^=$, of which we briefly
recall the main notions; a detailed account can be found e.g.\ in
\cite{mossakowski:tcs:2002}.

\begin{figure}[!t]
\begin{hetcasl}
\SPEC \=\SIdIndex{Nat} \Ax{=}\\
\> \KW{free} \KW{type} \Id{Nat} \Ax{\cln\cln=} \Ax{0} \AltBar{} \Id{suc}(\Id{Nat})\\
\> \OPS \Ax{\_\_}\Ax{+}\Ax{\_\_} \Ax{\cln} \Id{Nat} \Ax{\times} \Id{Nat} \Ax{\rightarrow} \Id{Nat} \\
\> \PRED \Ax{\_\_}\Ax{<}\Ax{\_\_} \Ax{\cln} \Id{Nat} \Ax{\times} \Id{Nat}\\
\> \Ax{\forall} \Id{n}, \Id{m} \Ax{\cln} \Id{Nat}
\= \Ax{\axdot} \Ax{0} \Ax{+} \Id{n} \Ax{=} \Id{n} \quad
\= \Ax{\axdot} \Id{suc}(\Id{n}) \Ax{+} \Id{m} \Ax{=} \Id{suc}(\Id{n} \Ax{+} \Id{m}) \quad\=\\
\>\> \Ax{\axdot} \Ax{\neg} \Id{n} \Ax{<} \Ax{0}
\> \Ax{\axdot} \Ax{0} \Ax{<} \Id{suc}(\Id{n})
\> \Ax{\axdot} \Id{suc}(\Id{m}) \Ax{<} \Id{suc}(\Id{n}) \Ax{\Leftrightarrow} \Id{m} \Ax{<} \Id{n}\\
\KW{end}
\end{hetcasl}
\vskip-18pt
\caption{A \CASL specification of the natural numbers}\label{lst:nat}
\end{figure}
At the level of basic specifications, cf.\ \cref{lst:nat}, one can
declare \textit{sort}s, \textit{op}erations, and \textit{pred}icates
with given argument and result sorts.  Formally, this defines a
\emph{many-sorted signature} $\Sigma = (S, F, P)$ with a set $S$ of
sorts, a $S^*\times S$-sorted family $F = (F_{w, s})_{w\,s \in S^+}$
of \emph{total function symbols}, and a $S^*$-sorted family $P =
(P_w)_{w \in S^*}$ of \emph{predicate symbols}.  Using these symbols,
one may then write axioms in first-order logic. Moreover, one can
specify data \textit{type}s, given in terms of alternatives consisting
of data constructors and, optionally, selectors, which may be declared
to be \textit{generated} or \textit{free}. Generatedness amounts to an
implicit higher-order induction axiom and intuitively states that all
elements of the data types are reachable by constructor terms (``no
junk''); freeness additionally requires that all these constructor
terms are distinct (``no confusion''). Basic \CASL specifications
denote the class of all algebras which fulfil the declared axioms,
i.e., \CASL has loose semantics.  In structured \CASL specifications,
a \emph{structured free} construct can be used to ensure freeness
(i.e., initial semantics) of a specification. For functions and
predicates, the effect of the structured free construct corresponds to
the effect of free types on sorts. A \emph{many-sorted
  $\Sigma$-structure $M$} consists of a non-empty carrier set $s^M$
for each $s \in S$, a total function $f^M : M_w \to M_s$ for each
function symbol $f \in F_{w, s}$ and a predicate $p^M$ for each
predicate symbol $p \in P_w$.  A \emph{many-sorted $\Sigma$-sentence}
is a closed many-sorted first-order formula over $\Sigma$ or a sort
generation constraint.

\section{Simple UML State Machines}\label{sec:running_example}

UML state machines~\cite{uml-2.5.1} provide means to specify the reactive
behaviour of objects or component instances.  These entities hold an internal
data state, typically given by a set of attributes or properties, and shall
react to event occurrences by firing different transitions in different control
states.  Such transitions may have a guard depending on event arguments and the
internal state and may change, as an effect, the internal control and data state
of the entity as well as raise events on their own.

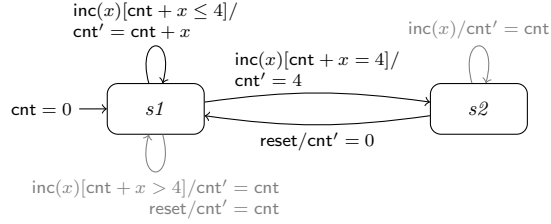
\begin{figure}[!t]\centering
\begin{tikzpicture}[scale=.85, transform shape]
\tikzset{
  every node/.append style={align=left},
  every state/.append style={rectangle, rounded corners, minimum width=1.5cm},
}
\node[state, initial text={\fontsize{8pt}{8pt}\selectfont$\attr{cnt} = 0$}, initial left] (s1) {$\state{s1}$};
\node[state, right of=s1, node distance=5.0cm] (s2) {$\state{s2}$};
\path[->,font={\fontsize{8pt}{8pt}\selectfont}] 
  (s1) edge[bend left=8] node[anchor=base, above] {$\evt{inc}(x) [\attr{cnt} + x = 4] /$\\ $\attr{cnt}' = 4$} (s2)
  (s1) edge[loop above, min distance=8mm] node[anchor=base, above] {$\evt{inc}(x) [\attr{cnt} + x \leq 4] /$\\ $\attr{cnt}' = \attr{cnt} + x$} (s1)
  (s1) edge[color=gray, loop below, min distance=8mm] node[anchor=north, below, align=right] {$\evt{inc}(x) [\attr{cnt} + x > 4] / \attr{cnt}' = \attr{cnt}$\\ $\evt{reset} / \attr{cnt}' = \attr{cnt}$} (s1)
  (s2) edge[bend left=8] node[anchor=top, below] {$\evt{reset} / \attr{cnt}' = 0$} (s1)
  (s2) edge[color=gray, loop above, min distance=8mm] node[anchor=base, above] {$\evt{inc}(x) / \attr{cnt}' = \attr{cnt}$} (s2)
;
\end{tikzpicture}
\vskip-8pt
\caption{Simple UML state machine $\mathit{Counter}$}%
\label{fig:uml-counter}%
\end{figure}
\Cref{fig:uml-counter} shows the example of a bounded, resettable
counter working on an attribute $\attr{cnt}$ (assumed to take values
in the natural numbers) that is initialised with $0$.  The counter can
be $\evt{reset}$ to $0$ or $\evt{inc}$reased by a natural number $x$,
subject to the current control state ($\state{s1}$ or $\state{s2}$)
and the guards (shown in square brackets) and effects (after the
slash) of the outgoing transitions.  An effect describes how the data
state before firing a transition (referred to by unprimed attribute
names) relates to the data state after (primed names) in a single
predicate; this generalises the more usual sequences of assignments
such that $\attr{cnt}' = \attr{cnt} + x$ corresponds to $\attr{cnt}
\gets \attr{cnt} + x$ and $\attr{cnt}' = \attr{cnt}$ to a skip.  The
machine is specified non-deterministically: If event $\evt{inc}(x)$
occurs in state $\state{s1}$ such that the guard $\attr{cnt} + x = 4$
holds, the machine can either stay in $\state{s1}$ or it can proceed
to $\state{s2}$.  Seemingly, the machine does not react to
$\evt{reset}$ in $\state{s1}$ and to $\evt{inc}$ in $\state{s2}$.
However, UML state machines are meant to be input-enabled such that
all event occurrences to which the machine shows no explicit reacting
transition are silently discarded, as indicated by the ``grey''
transitions.  Overall, the machine $\mathit{Counter}$ shall ensure
that $\attr{cnt}$ never exceeds $4$.

It is for such simple UML state machines as the counter in
\cref{fig:uml-counter} that we want to provide proof support in \SPASS
via an institutional encoding in \CASL.  The sub-language covers the
following fundamental state machine features: data, states, and
(non-deterministic) guarded transitions for reacting to events.
However, for the time being, we leave out not only all advanced
modelling constructs, like hierarchical states or compound
transitions, but also defer, most importantly, event-based
communication between state machines to future work.  In the following
we make first precise the syntax of the machines by means of
event/data signatures, data states and transitions, guards and
effects.  Then we introduce semantic structures for the machines and
define their model class. Syntax and semantics of simple UML
state machines form the basis for their institutionalisation.  We thus
also introduce event/data signature morphisms and the corresponding
formulæ translation and structure reducts in order to be able to
change the interface of simple UML state machines.

\subsection{Event/Data Signatures, Data States and Transitions}

We capture the events for a machine in an \emph{event signature} $E$ that
consists of a finite set of events $|E|$ and a map $\upsilon(E)$ assigning to
each $e \in |E|$ a finite set of variables, where we write $e(X)$ for $e \in
|E|$ and $\upsilon(E)(e) = X$, and also $e(X) \in E$ in this case.  For the data
state, we use a \emph{data signature} $A$ consisting of a finite set of
attributes.  An \emph{event/data signature} $\Sigma$ consists of an event
signature $\Evt(\Sigma)$ and a data signature $\Attr(\Sigma)$.

\begin{example}\label{ex:edhml-signatures}
The event/data signature $\Sigma$ of the simple UML state machine in
\cref{fig:uml-counter} is given by the set of events $|\Evt(\Sigma)| = \{
\evt{inc}, \evt{reset} \}$ with argument variables
$\upsilon(\Evt(\Sigma))(\evt{inc}) = \{ x \}$ and
$\upsilon(\Evt(\Sigma))(\evt{reset}) = \emptyset$ such that $\evt{inc}(x) \in
\Evt(\Sigma)$ and $\evt{reset} \in \Evt(\Sigma)$; as well as the data signature
$\Attr(\Sigma) = \{ \attr{cnt} \}$.
\end{example}

For specifying transition guards and effects, we exchange UML's notorious and
intricate expression and action languages both syntactically and semantically by
a straightforward \CASL fragment rendering guards as data state predicates and
effects as data transition predicates: We assume given a fixed universe $\DATA$
of \emph{data values} and a \CASL specification $\Data$ with a dedicated sort
$\data$ in its signature $\Sign(\Data)$ such that the universe $\data^M$ of
every model $M \in \Mod[\CASL](\Data)$ is isomorphic to $\DATA$, i.e., there is
a bijection $\iota_{M, \data} : \data^M \isorel \DATA$. This puts at our
disposal the open formulæ $\Frm[\CASL]{\Sign(\Data)}{X}$ over sorted variables
$X = (X_s)_{s \in S}$ and their satisfaction relation $M, \beta
\models[\CASL]_{\Sign(\Data), X} \varphi$ for models $M \in \Mod[\CASL](\Data)$,
variable valuations $\beta : X \to M$, and formulæ $\varphi \in
\Frm[\CASL]{\Sign(\Data)}{X}$.

\begin{example}
Consider the natural numbers $\NZ$ as data values $\DATA$. The \CASL
specification in \cref{lst:nat} characterises $\NZ$ up to isomorphism
as the carrier set of the dedicated sort $\data = \Id{Nat}$. It
specifies an abstract data type with sort $\Id{Nat},$ operations $+, 0,
\Id{suc},$ and a predicate $<.$

\end{example}

The very simple choice of $\DATA$ capturing data with only a single sort can, in
principal, be replaced by any institutional data modelling language that, for
our purposes of a theoroidal institution comorphism (see \cref{sec:comorphism}),
is faithfully representable in \CASL; one such possibility are UML class
diagrams, see~\cite{james-et-al:wadt:2012}.

\paragraph{Data states and guards.}
A \emph{data state} $\omega$ for a data signature $A$ is given by a function
$\omega : A \to \DATA$; in particular, $\DataSt(A) = \DATA^{A}$ is the set of
$A$"=data states.  The guards of a machine are \emph{state predicates} in
$\Frm[\DATA]{A}{X} = \Frm[\CASL]{\Sign(\Data)}{A \cup X}$, taking $A$ as well as
an additional set $X$ as variables of sort $\data$.  A state predicate $\phi \in
\Frm[\DATA]{A}{X}$ is to be interpreted over an $A$"=data state $\omega$ and
valuation $\beta : X \to \DATA$ and we define the \emph{satisfaction relation}
$\models[\DATA]$ by
\begin{equation*}
  \omega, \beta \models[\DATA]_{A, X} \phi
\iff
  M, \iota_{M, \data}^{-1} \compfun (\omega \cup \beta) \models[\CASL]_{\Sign(\Data), A \cup X} \phi
\end{equation*}
where $M \in \Mod[\CASL](\Data)$ and $\iota_{M, \data} : M(\data) \isorel
\DATA$.  For a state predicate $\varphi \in \Frm[\DATA]{A}{\emptyset}$ not
involving any variables, we write $\omega \models[\DATA]_A \varphi$ for $\omega
\models[\DATA]_{A, \emptyset} \varphi$.

\begin{example}
The guard $\attr{cnt} + x \leq 4$ of the machine in \cref{fig:uml-counter}
features both the attribute $\attr{cnt}$ and the variable $x$. A data state
fulfilling this state predicate for $x = 0$ is $\attr{cnt} \mapsto 3$.
\end{example}

\paragraph{Data transitions and effects.}
A \emph{data transition} $(\omega, \omega')$ for a data signature $A$ is a pair
of $A$-data states; in particular, $\DataTr(A) = (\DATA^{A})^2$ is the set of
$A$"=data transitions.  It holds that $(\DATA^{A})^2 \isorel \DATA^{2A}$, where
$2A = A \uplus A$ and we assume that no attribute in $A$ ends in a prime
$\prime$ and all attributes in the second summand are adorned with an additional
prime.  The effects of a machine are \emph{transition predicates} in
$\Frm[2\DATA]{A}{X} = \Frm[\DATA]{2A}{X}$.  The satisfaction relation
$\models[2\DATA]$ for a transition predicate $\psi \in \Frm[2\DATA]{A}{X}$, data
transition $(\omega, \omega') \in \DataTr(A)$, and valuation $\beta : X \to
\DATA$ is defined as
\begin{equation*}
  (\omega, \omega'), \beta \models[2\DATA]_{A, X} \psi
\iff
  \omega + \omega', \beta \models[\DATA]_{2A, X} \psi
\end{equation*}
where $\omega + \omega' \in \DataSt(2A)$ with $(\omega + \omega')(a) =
\omega(a)$ and $(\omega + \omega')(a') = \omega'(a)$.

\begin{example}
The effect $\attr{cnt}' = \attr{cnt} + x$ of the machine in
\cref{fig:uml-counter} describes the increment of the value of attribute
$\attr{cnt}$ by a variable amount $x$.
\end{example}

\subsection{Syntax of Simple UML State Machines}

A simple UML state machine $U$ uses an event/data signature
$\EDSign(U)$ for its events and attributes and consists of a finite
set of \emph{control states} $\CtrlSt(U)$, a finite set of
\emph{transition specifications} $\Trans(U)$ of the form $(c, \phi,
e(X), \psi, c')$ with $c, c' \in \CtrlSt(U)$, $e(X) \in
\Evt(\EDSign(U))$, a state predicate $\phi \in
\Frm[\DATA]{\Attr(\EDSign(U))}{X}$, a transition predicate $\psi \in
\Frm[2\DATA]{\Attr(\EDSign(U))}{X}$, an \emph{initial control state}
$\ctrlSt_0(U) \in \CtrlSt(U)$, and an \emph{initial state predicate}
$\datapred_0(U) \in \Frm[\DATA]{\Attr(\EDSign(U))}{\emptyset}$, such
that $\CtrlSt(U)$ is \emph{syntactically reachable}, i.e., for every
$c \in \CtrlSt(U) \setminus \{ \ctrlSt_0(U) \}$ there are
$(\ctrlSt_0(U),\allowbreak \phi_1,\allowbreak e_1(X_1),\allowbreak
\psi_1,\allowbreak c_1), \ldots, (c_{n-1},\allowbreak
\phi_n,\allowbreak e_n(X_n),\allowbreak \psi_n,\allowbreak c_n) \in
\Trans(U)$ with $n > 0$ such that $c_n = c$.
Syntactic reachability guarantees initially connected state machine graphs. 
This simplifies graph-based algorithms (see \cref{alg:sen-op-spec}).

\begin{example}
The machine in \cref{fig:uml-counter} has as its control states $\{ \state{s1},
\state{s2} \}$, as its transition specifications $\splitatcommas{\{ (\state{s1},
  \attr{cnt} + x \leq 4, \evt{inc}(x), \attr{cnt}' = \attr{cnt} + x,
  \state{s1}), (\state{s1}, \attr{cnt} + x = 4, \evt{inc}(x), \attr{cnt}' = 4,
  \state{s2}), (\state{s2}, \mathrm{true}, \evt{reset}, \attr{cnt}' = 0,
  \state{s1}) \}}$, as initial control state $\state{s1}$, and as initial state
predicate $\attr{cnt} = 0$.
\end{example}

\subsection{Event/Data Structures and Models of Simple UML State Machines}

For capturing machines semantically, we use event/data structures that are given
over an event/data signature $\Sigma$ and consist of a transition system of
configurations such that all configurations are reachable from its initial
configurations.  Herein, configurations show a control state, corresponding to
machine states, and a data name from which a proper data state over
$\Attr(\Sigma)$ can be retrieved by a labelling function.  Transitions connect
configurations by events from $\Evt(\Sigma)$ with their arguments instantiated
by data from $\DATA$.

Formally, a $\Sigma$"=\emph{event/data structure} $M = (\Gamma, R, \Gamma_0,
\omega)$ over an event/data signature $\Sigma$ consists of a set of
\emph{configurations} $\Gamma \subseteq C \times D$ for some sets of
\emph{control states} $C$ and \emph{data names} $D$, a family of
\emph{transition relations} $R = (R_{e(\beta)} \subseteq \Gamma \times
\Gamma)_{e(X) \in \Evt(\Sigma), \beta : X \to \DATA}$, and a non-empty set of
\emph{initial configurations} $\Gamma_0 = \{ c_0 \} \times D_0 \subseteq \Gamma$
with a unique \emph{initial control state} $c_0 \in C$ such that $\Gamma$ is
\emph{reachable} via $R$, i.e., for all $\gamma \in \Gamma$ there are $\gamma_0
\in \Gamma_0$, $n \geq 0$, $e_1(X_1),\allowbreak \ldots,\allowbreak e_n(X_n) \in
\Evt(\Sigma)$, $\beta_1 : X_1 \to \DATA, \ldots, \beta_n : X_n \to \DATA$, and
$(\gamma_i, \gamma_{i+1}) \in R_{e_{i+1}(\beta_{i+1})}$ for all $0 \leq i < n$
with $\gamma_n = \gamma$; and a \emph{data state labelling} $\omega : D \to
\DataSt(\Attr(\Sigma))$.  We write $\ctrlSt(M)(\gamma) = c$ and
$\dataSt(M)(\gamma) = \omega(d)$ for $\gamma = (c, d) \in \Gamma$, $\Conf(M)$
for $\Gamma$, $\CtrlSt(M)$ for $\{ \ctrlSt(M)(\gamma) \mid \gamma \in \Conf(M)
\}$, $\Rel(M)$ for $R$, $\Conf_0(M)$ for $\Gamma_0$, $\ctrlSt_0(M)$ for $c_0$,
and $\DataSt_0(M)$ for $\{ \omega(M)(\gamma_0) \mid \gamma_0 \in \Gamma_0 \}$.

The restriction to reachable transition systems is not strictly necessary and
could be replaced by constraining all statements on event/data structures to
take into account only their reachable part (see, e.g., \cref{lem:sat-cond}).

\begin{example}\label{ex:edhml-structures}
For an event/data structure for the machine in \cref{fig:uml-counter} over its
signature $\Sigma$ in \cref{ex:edhml-signatures} we may choose the control
states $C$ as $\{ \state{s1}, \state{s2} \}$, and the data names $D$ as the
set $\DataSt(\Attr(\Sigma)) = \DATA^{\{ \attr{cnt} \}}$.  In particular,
the data state labelling $\omega$ is just the identity.  The only initial
configuration is $(\state{s1}, \{ \attr{cnt} \mapsto 0 \})$.  A possible
transition goes from configuration $(\state{s1}, \{ \attr{cnt} \mapsto 2 \})$ to
configuration $(\state{s2}, \{ \attr{cnt} \mapsto 4 \})$ with the instantiated
event $\evt{inc}(2)$.
\end{example}

A $\EDSign(U)$"=event/data structure $M$ is a \emph{model} of a simple UML state
machine $U$ if $\CtrlSt(U) \subseteq \CtrlSt(M)$ up to a bijective renaming,
$\ctrlSt_0(M) = \ctrlSt_0(U)$, $\DataSt_0(M) \subseteq \{ \omega \in
|\DataSt(\Attr(\EDSign(U)))| \mid \omega \models[\DATA]_{\Attr(\EDSign(U))}
\datapred_0(U) \}$, and if the following holds for all $(c, d) \in \Conf(M)$:
\begin{itemize}
  \item for all $(c, \phi, e(X), \psi, c') \in \Trans(U)$ and $\beta : X \to
\DATA$ with $\dataSt(M)(d), \beta \models[\DATA]_{\Attr(\EDSign(U)), X}
\phi$, there is a $((c, d),\allowbreak (c', d')) \in \Rel(M)_{e(\beta)}$ with
$(\dataSt(M)(d),\allowbreak \dataSt(M)(d')), \beta
\models[2\DATA]_{\Attr(\EDSign(U)), X} \mkern-2mu\psi$;

  \item for all $((c, d), (c', d')) \in \Rel(M)_{e(\beta)}$ there is either some
$(c, \phi, e(X), \psi, c') \in \Trans(U)$ with $\dataSt(M)(d), \beta
\models[\DATA]_{\Attr(\EDSign(U)), X} \phi$ and $(\dataSt(M)(d),\allowbreak
\dataSt(M)(d')), \beta \models[2\DATA]_{\Attr(\EDSign(U)), X} \psi$, or
$\dataSt(M)(d), \beta \not\models[\DATA]_{\Attr(\EDSign(U)), X} \mkern-3mu\bigvee_{(c, \phi,
  e(X), \psi, c') \in \Trans(U)} \phi$, $c = c'$, and $\dataSt(M)(d) =
\dataSt(M)(d')$.
\end{itemize}

A model of $U$ thus on the one hand implements each transition prescribed by
$U$, but on the other hand must not show transitions not covered by the
specified transitions.  Moreover, it is \emph{input-enabled}, i.e., every event
can be consumed in every control state: If no precondition of an explicitly
specified transition is satisfied, there is a self-loop which leaves the data
state untouched.  In fact, input-enabledness, as required by the UML
specification~\cite{uml-2.5.1}, can also be rendered as a syntactic
transformation making a simple UML state machine $U$ input-enabled by adding the
following set of transition specifications for idling self-loops:
\begin{equation*}\textstyle
  \{ (c, \neg(\bigvee_{(c, \phi, e(X), \psi, c') \in \Trans(U)} \phi), e(X), \id{\Attr(\EDSign(U))}, c) \mid c \in C,\ e(X) \in \Evt(\EDSign(U)) \}
\ \text{.}
\end{equation*}

\begin{example}
For the simple UML state machine in \cref{fig:uml-counter} the ``grey''
transitions correspond to an input-enabledness completion w.r.t.\ the ``black''
transitions.
\end{example}

The requirement of syntactic reachability for simple UML state machines is
correlated with the requirement of (semantic) reachability of event/data
structures, as a machine violating syntactic reachability cannot have a model.
Equally, a machine with a non-satisfiable initial state predicate fails to have
a model.

\subsection{Event/Data Signature Morphisms, Reducts, and Translations}\label{sec:ed-morphisms}

The external interface of a simple UML state machine is given by
events, its internal interface by attributes. Both interfaces,
represented as an event/data signature, are susceptible to change in
the system development process which is captured by signature
morphisms.  Such changes have also to be reflected in the guards and
effects, i.e., data state and transition predicates, by syntactical
translations as well as in the interpretation domains by semantical
reducts.

A \emph{data signature morphism} from a data signature $A$ to a data signature
$A'$ is a function $\alpha : A \to A'$.  The \emph{$\alpha$-reduct} of an
$A'$-data state $\omega' : A' \to \DATA$ along a data signature morphism $\alpha
: A \to A'$ is given by the $A$-data state $\reduct{\omega'}{\alpha} : A \to
\DATA$ with $(\reduct{\omega'}{\alpha})(a) = \omega'(\alpha(a))$ for every $a
\in A$; the \emph{$\alpha$-reduct} of an $A'$-data transition $(\omega',
\omega'')$ by the $A$-data transition $\reduct{(\omega', \omega'')}{\alpha} =
(\reduct{\omega'}{\alpha}, \reduct{\omega''}{\alpha})$.  The \emph{state
  predicate translation} $\Frm[\DATA]{\alpha}{X} : \Frm[\DATA]{A}{X} \to
\Frm[\DATA]{A'}{X}$ along a data signature morphism $\alpha : A \to A'$ is given
by the \CASL-formula translation $\Frm[\CASL]{\Sign(\Data)}{\alpha \cup \id{X}}$
along the substitution $\alpha \cup \id{X}$; the \emph{transition predicate
  translation} $\Frm[2\DATA]{\alpha}{X}$ by $\Frm[\DATA]{2\alpha}{X}$ with
$2\alpha : 2A \to 2A'$ defined by $2\alpha(a) = \alpha(a)$ and $2\alpha(a') =
\alpha(a)'$.  For each of these two reduct-translation-pairs the
\emph{satisfaction condition} holds due to the general substitution lemma for
\CASL:
\begin{align*}
  \reduct{\omega'}{\alpha}, \beta \models[\DATA]_{A, X} \phi
&\iff
  \omega', \beta \models[\DATA]_{A', X} \Frm[\DATA]{\alpha}{X}(\phi)
\\
  \reduct{(\omega', \omega'')}{\alpha}, \beta \models[2\DATA]_{A, X} \psi
&\iff
  (\omega', \omega''), \beta \models[2\DATA]_{A', X} \Frm[2\DATA]{\alpha}{X}(\psi)
\end{align*}

An \emph{event signature morphism} $\eta : \Evt \to \Evt'$ is a function $\eta :
|\Evt| \to |\Evt'|$ such that $\upsilon(\Evt)(e) = \upsilon(\Evt')(\eta(e))$ for
all $e \in |\Evt|$.  An \emph{event/data signature morphism} $\sigma : \Sigma
\to \Sigma'$ consists of an event signature morphism $\Evt(\sigma) : \Evt(\Sigma)
\to \Evt(\Sigma')$ and a data signature morphism $\Attr(\sigma) : \Attr(\Sigma)
\to \Attr(\Sigma')$.  The \emph{$\sigma$-reduct} of a $\Sigma'$-event/data
structure $M'$ along $\sigma$ is the $\Sigma$-event/data structure
$\reduct{M'}{\sigma}$ such that
\begin{itemize}
  \item $\Conf(\reduct{M'}{\sigma}) \subseteq \Conf(M')$ as well as
$\Rel(\reduct{M'}{\sigma}) = (\Rel(\reduct{M'}{\sigma})_{e(\beta)})_{e(X) \in
  \Evt(\Sigma), \beta : X \to \DATA}$ are inductively defined by
$\Conf(\reduct{M'}{\sigma}) \supseteq \Conf_0(M')$ and, for all $\gamma',
\gamma'' \in \Conf(M')$, $e(X) \in \Evt(\Sigma)$, and $\beta : X \to \DATA$, if
$\gamma' \in \Conf(\reduct{M'}{\sigma})$ and $(\gamma', \gamma'') \in
\Rel(M')_{\Evt(\sigma)(e)(\beta)}$, then $\gamma'' \in
\Conf(\reduct{M'}{\sigma})$ and $(\gamma', \gamma'') \in
\Rel(\reduct{M'}{\sigma})_{e(\beta)}$;

  \item $\Conf_0(\reduct{M'}{\sigma}) = \Conf_0(M')$; and

  \item $\dataSt(\reduct{M'}{\sigma})(\gamma') =
\reduct{(\dataSt(M')(\gamma'))}{\sigma}$ for all $\gamma' \in
\Conf(\reduct{M'}{\sigma})$.
\end{itemize}

Building a reduct of an event/data-structure does not affect the single
configurations, but potentially reduces the set of configurations by restricting
the available events, and the data state observable from the data name of a
configuration.  We denote by $\Conf^F(M, \gamma)$ and $\Conf^F(M)$,
respectively, the set of configurations of a $\Sigma$-event/data structure $M$
that are $F$-reachable from a configuration $\gamma \in \Conf(M)$ and from an
initial configuration $\gamma_0 \in \Conf_0(M)$, respectively, with a set of
events $F \subseteq \Evt(\Sigma)$ where a $\gamma_n \in \Conf(M)$ is
\emph{$F$"=reachable in $M$ from} a $\gamma_1 \in \Conf(M)$ if there are $n \geq
1$, $e_2(X_2), \ldots, e_n(X_n) \in F$, $\beta_2 : X_2 \to \DATA, \ldots,
\beta_n : X_n \to \DATA$, and $(\gamma_i, \gamma_{i+1}) \in
\Rel(M)_{e_{i+1}(\beta_{i+1})}$ for all $1 \leq i < n$.

\begin{techreport}
\begin{lemma}\label{lem:relative}
Let $\sigma : \Sigma \to \Sigma'$ be an event/data signature morphism, $F
\subseteq \Evt(\Sigma)$, and $M'$ a $\Sigma'$-event/data structure.
\begin{enumerate}
  \item\label{it:lem:relative:trans} For all $\gamma_1', \gamma'_2 \in
\Conf(M')$, if $\gamma_1' \in \Conf(\reduct{M'}{\sigma})$, then $(\gamma_1',
\gamma_2') \in \Rel(\reduct{M'}{\sigma})_{e(\beta)}$ if, and only if,
$(\gamma_1', \gamma_2') \in \Rel(M')_{\Evt(\sigma)(e)(\beta)}$.

  \item\label{it:lem:relative:reach-from} For all $\gamma', \gamma'' \in
\Conf(M')$ such that $\gamma' \in \Conf(\reduct{M'}{\sigma})$, $\gamma'' \in
\Conf^F(\reduct{M'}{\sigma}, \gamma')$ if, and only if, $\gamma'' \in
\Conf^{\Evt(\sigma)(F)}(M', \gamma')$.

  \item\label{it:lem:relative:reach} For all $\gamma' \in \Conf(M')$, $\gamma'
\in \Conf^F(\reduct{M'}{\sigma})$ if, and only if, $\gamma' \in
\Conf^{\Evt(\sigma)(F)}(M')$.
\end{enumerate}
\end{lemma}
\begin{proof}
\cref{it:lem:relative:trans}~This follows directly from the inductive definition
of the $\sigma$-reduct of $\Sigma'$-event/data structures.

\smallskip\noindent%
\cref{it:lem:relative:reach-from}~Let $\gamma', \gamma'' \in \Conf(M')$ with
$\gamma' \in \Conf(\reduct{M'}{\sigma})$.  By induction, it holds that $\gamma''
\in \Conf^{\Evt(\sigma)(F)}(M')$ if, and only if, there are $n \geq 0$,
$e_1(X_1), \ldots, e_{n}(X_n) \in F$, and $\beta_1 : X_1 \to \DATA, \ldots,
\beta_n : X_n \to \DATA$, and $(\gamma_i', \gamma_{i+1}') \in
\Rel(M')_{\Evt(\sigma)(e_{i+1})(\beta_{i+1})}$ for all $0 \leq i < n$ with
$\gamma' = \gamma_0'$ and $\gamma'' = \gamma_n'$.  Thus, by
\cref{it:lem:relative:trans}, since $\gamma' \in \Conf(\reduct{M'}{\sigma})$,
$\gamma'' \in \Conf^{\Evt(\sigma)(F)}(M')$ if, and only if, $\gamma'' \in
\Conf^F(\reduct{M'}{\sigma})$.

\smallskip\noindent%
\cref{it:lem:relative:reach}~Let $\gamma' \in \Conf(M')$.  By definition it
holds that $\gamma' \in \Conf^{\Evt(\sigma)(F)}(M')$ if, and only if, $\gamma'
\in \Conf^{\Evt(\sigma)(F)}(M', \gamma_0')$ for some $\gamma_0' \in
\Conf_0(M')$; if, and only if $\gamma' \in \Conf^{F}(\reduct{M'}{\sigma},
\gamma_0')$ for some $\gamma_0' \in \Conf(M')$ by
\cref{it:lem:relative:reach-from} since $\gamma_0' \in
\Conf_0(\reduct{M'}{\sigma})$; if, and only if, $\gamma' \in
\Conf^{F}(\reduct{M'}{\sigma})$ since $\Conf_0(\reduct{M'}{\sigma}) =
\Conf_0(M')$.
\end{proof}

\begin{example}
Let $\Sigma$ be as in \cref{ex:edhml-signatures} and $\Sigma_0$ the event/data
signature with $\Evt(\Sigma_0) = \Evt(\Sigma)$ and $\Attr(\Sigma_0) =
\emptyset$.  Consider the signature morphism $\sigma : \Sigma_0 \to \Sigma$ as
the identity on the events and the trivial embedding on the attributes.  Let $M$
be a model of the simple UML state machine in \cref{fig:uml-counter}.  The
syntactic transition $(\state{s1}, \attr{cnt} + x \leq 4, \evt{inc}(x),
\attr{cnt}' = \attr{cnt} + x, \state{s1})$ induces, among others, the two
semantic transitions $((\state{s1}, d_1), (\state{s1}, d_2)), ((\state{s1},
d_2), (\state{s1}, d_3)) \in R(M)_{\evt{inc}(\{ x \mapsto 1 \})}$ where
$\dataSt(M)(d_i) = \{ \attr{cnt} \mapsto i \}$ for $1 \leq i \leq 3$.  In the
reduct $\reduct{M}{\sigma}$ we find exactly these two semantic transitions,
however, $\dataSt(\reduct{M}{\sigma})(d_i) = \emptyset$ for all $1 \leq i \leq
3$.  This illustrates why we distinguish between data states and data names.
With the distinction, we have a bijection between semantic transitions in the
reduct and semantic transitions in the original structure.  Without the
distinction, the two different transitions in $M$ would collapse into one
transition only as there is just a single data state $\emptyset$.
\end{example}
\end{techreport}

Although it is straightforward to define a translation of simple UML state
machines along an event/data signature morphism, the rather restrictive notion
of their models prevents the satisfaction condition to hold.  In fact, this is
already true for our previous endeavours to institutionalise UML state
machines~\cite{knapp-et-al:fase:2015,knapp-mossakowski:calco:2017}.  There
machines themselves were taken to be sentences over signatures comprising both
events and states, and the satisfaction relation also required that a model
shows exactly the transitions of such a machine sentence.  For signature
morphisms $\sigma$ that are not surjective on states, building the reduct could
result in less states and transitions, which leads to the following
counterexample to the satisfaction condition~\cite{rosenberger:master:2017}:
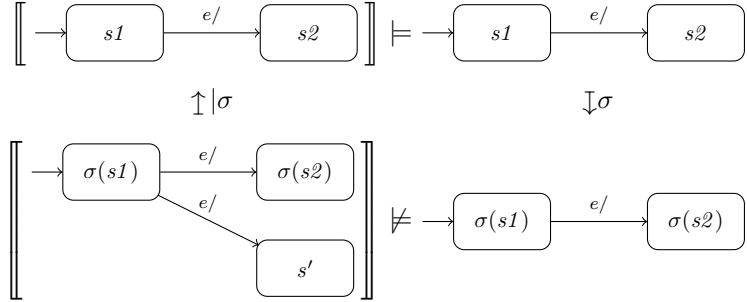
\begin{align*}
\left\llbracket\ %
\begin{tikzpicture}[baseline=(s1.base), auto, scale=.85, transform shape]
\tikzset{
  every node/.append style={align=left},
  every state/.append style={rectangle, rounded corners, minimum width=1.5cm},
  every initial by arrow/.style={outer sep=0pt, inner sep=0pt},
}
\node[state, initial text={}, initial left] (s1) {$\state{s1}$};
\node[state, right of=s1, node distance=3.0cm] (s2) {$\state{s2}$};
\path[->,font={\fontsize{8pt}{8pt}\selectfont}] 
  (s1) edge node[anchor=base, above] {$e /$} (s2)
;
\end{tikzpicture}%
\ \right\rrbracket
&\models
\begin{tikzpicture}[baseline=(s1.base), scale=.85, transform shape]
\tikzset{
  every node/.append style={align=left},
  every state/.append style={rectangle, rounded corners, minimum width=1.5cm},
  every initial by arrow/.style={outer sep=0pt, inner sep=0pt},
}
\node[state, initial text={}, initial left] (s1) {$\state{s1}$};
\node[state, right of=s1, node distance=3.0cm] (s2) {$\state{s2}$};
\path[->,font={\fontsize{8pt}{8pt}\selectfont}] 
  (s1) edge node[anchor=base, above] {$e /$} (s2)
;
\end{tikzpicture}
\\[1ex]
\rotatebox[origin=c]{90}{$\mapsto$}\,\reductop\sigma\hspace*{1.95cm} & \hspace*{2.65cm}\rotatebox[origin=c]{-90}{$\mapsto$}\,\sigma
\\[1ex]
\left\llbracket\ %
\begin{tikzpicture}[baseline={($0.5*(s2.base)+0.5*(s3.base)$)}, auto, scale=.85, transform shape]
\tikzset{
  every node/.append style={align=left},
  every state/.append style={rectangle, rounded corners, minimum width=1.5cm},
  every initial by arrow/.style={outer sep=0pt, inner sep=0pt},
}
\node[state, initial text={}, initial left] (s1) {$\sigma(\state{s1})$};
\node[state, right of=s1, node distance=3.0cm] (s2) {$\sigma(\state{s2})$};
\node[state, below of=s2, node distance=1.5cm] (s3) {$\state{s}'$};
\path[->,font={\fontsize{8pt}{8pt}\selectfont}] 
  (s1) edge node[anchor=base, above] {$e /$} (s2)
  (s1) edge[shorten <=-1pt, shorten >=-1pt] node[anchor=base, above] {$e /$} (s3)
;
\end{tikzpicture}%
\ \right\rrbracket
&\not\models
\begin{tikzpicture}[baseline=(s1.base), scale=.85, transform shape]
\tikzset{
  every node/.append style={align=left},
  every state/.append style={rectangle, rounded corners, minimum width=1.5cm},
  every initial by arrow/.style={outer sep=0pt, inner sep=0pt},
}
\node[state, initial text={}, initial left] (s1) {$\sigma(\state{s1})$};
\node[state, right of=s1, node distance=3.0cm] (s2) {$\sigma(\state{s2})$};
\path[->,font={\fontsize{8pt}{8pt}\selectfont}] 
  (s1) edge node[anchor=base, above] {$e /$} (s2)
;
\end{tikzpicture}
\end{align*}

We therefore propose to make a detour through a more general hybrid modal logic.
This logic is directly based on event/data structures and thus close to the
domain of state machines.  For forming an institution, its hybrid features allow
to avoid control states as part of the signature and its event-based modalities
allow to specify both mandatory and forbidden behaviour in a more fine-grained
manner.  Still, the logic is expressive enough to characterise the model class
of a simple UML state machine syntactically.

\section{A Hybrid Modal Logic for Event/Data Systems}\label{sec:edhml}

The logic $\EDHML$ is a hybrid modal logic for specifying event/data-based
reactive systems and reasoning about them.  The $\EDHML$-signatures are the
event/data signatures, the $\EDHML$-structures the event/data structures.  The
modal part of the logic allows to handle transitions between configurations
where the modalities describe moves between configurations that adhere to a
pre-condition or guard as a state predicate for an event with arguments and a
transition predicate for the data change corresponding to effects.  The hybrid
part of the logic allows to bind control states of system configurations and to
jump to configurations with such control states explicitly, but leaves out
nominals as interfacing names as well as the possibility to quantify over
control states.  The logic builds on the hybrid dynamic logic $\HDL$ for
specifying reactive systems without data~\cite{madeira-et-al:ictac:2016} and its
extension $\EDHDL$ to handle also data~\cite{hennicker-madeira-knapp:fase:2019}.
We restrict ourselves to modal operators consisting only of single instead of
compound actions as done in dynamic logic.  However, we still retain a box
modality for accessing all configurations that are reachable from a given
configuration.  Moreover, we extend $\EDHDL$ by adding parameters to events.

The category of \emph{$\EDHML$-signatures} $\Sig[\EDHML]$ consists of the
event/data signatures and signature morphisms.  The $\Sigma$-event/data
structures form the discrete category $\Str[\EDHML](\Sigma)$ of
\emph{$\EDHML$-structures} over $\Sigma$.  For each signature morphism $\sigma :
\Sigma \to \Sigma'$ in $\Sig[\EDHML]$ the \emph{$\sigma$-reduct functor}
$\Str[\EDHML](\sigma) : \Str[\EDHML](\Sigma') \to \Str[\EDHML](\Sigma)$ is given
by $\Str[\EDHML](\sigma)(M') = \reduct{M'}{\sigma}$.  As the next step we
introduce the formulæ and sentences of $\EDHML$ together with their translation
along $\Sig[\EDHML]$-morphisms and their satisfaction over $\Str[\EDHML]$.  We
then show that for $\EDHML$ the satisfaction condition holds and thus obtain
$\EDHML$ as an institution.  Subsequently, we show that $\EDHML$ is
simultaneously expressive enough to characterise the model class of simple UML
state machines.

\subsection{Formulæ and Sentences of $\EDHML$}

$\EDHML$-formulæ aim at expressing control and data state properties of
configurations as well as accessibility properties of configurations along
transitions for particular events.  The pure data state
part is captured by data state sentences over $\DATA$.  The control state part
can be accessed and manipulated by hybrid operators for binding the control
state in a state variable, $\hybindop s$; checking for a particular control
state, $s$; and accessing all configurations with a particular control state,
$\hyatop^F$, which, however, only pertains to reachable configurations relative
to a set $F$ of events.  Transitions between configurations are covered by
different modalities: a box modality for accessing all configurations that are
reachable from a given configuration, $\mlboxop^F$, again relative to a set $F$ of
events; a diamond modality for checking that an event with arguments is possible
with a particular data state change, $\dldia{e(X)}{\psi}{}$; and a modality for
checking the reaction to an event with arguments according to a pre-condition
and a transition predicate, $\dlwp{e(X)}{\phi}{\psi}{}$.

Formally, the \emph{$\Sigma$"=event/data formulæ} $\Frm[\EDHML]{\Sigma}{S}$ over
an event/data signature $\Sigma$ and a set of \emph{state variables} $S$ are
inductively defined by
\begin{itemize}
  \item $\varphi$ --- data state sentence $\varphi \in \Frm[\DATA]{\Attr(\Sigma)}{\emptyset}$
holds in the current configuration;

  \item $s$ ---
the control state of the current configuration is $s \in S$;

  \item $\hybind{s}{\varrho}$ --- calling the current control state $s$, formula
$\varrho \in \Frm[\EDHML]{\Sigma}{S \uplus \{ s \}}$ holds;

  \item $\hyat[F]{s}{\varrho}$ --- in all configurations with control state $s
\in S$ that are reachable with events from $F \subseteq \Evt(\Sigma)$ formula
$\varrho \in \Frm[\EDHML]{\Sigma}{S}$ holds;

  \item $\mlbox[F]{\varrho}$ --- in all configurations that are reachable
from the current configuration with events from $F \subseteq \Evt(\Sigma)$
formula $\varrho \in \Frm[\EDHML]{\Sigma}{S}$ holds;

  \item $\dldia{e(X)}{\psi}{\varrho}$ --- in the current configuration there is
a valuation of $X$ and a transition for event $e(X) \in \Evt(\Sigma)$ with these
arguments that satisfies transition formula $\psi \in
\Frm[2\DATA]{\Attr(\Sigma)}{X}$ and makes $\varrho \in \Frm[\EDHML]{\Sigma}{S}$
hold afterwards;

  \item $\dlwp{e(X)}{\phi}{\psi}{\varrho}$ --- in the current configuration for
all valuations of $X$ satisfying state formula $\phi\in
\Frm[\DATA]{\Attr(\Sigma)}{X}$ there is a transition for event $e(X) \in
\Evt(\Sigma)$ with these arguments that satisfies transition formula $\psi\in
\Frm[2\DATA]{\Attr(\Sigma)}{X}$ and makes $\varrho \in \Frm[\EDHML]{\Sigma}{S}$
hold afterwards;

  \item $\neg\varrho$ ---
in the current configuration $\varrho \in \Frm[\EDHML]{\Sigma}{S}$ does not hold;

  \item $\varrho_1 \lor \varrho_2$ --- in the current configuration $\varrho_1
\in \Frm[\EDHML]{\Sigma}{S}$ or $\varrho_2 \in \Frm[\EDHML]{\Sigma}{S}$ hold.
\end{itemize}
We write $\hyat{s}{\varrho}$ for $\hyat[\Evt(\Sigma)]{s}{\varrho}$,
$\mlbox{\varrho}$ for $\mlbox[\Evt(\Sigma)]{\varrho}$, $\mldia[F]{\varrho}$ for
$\neg\mlbox[F]{\neg\varrho}$, $\mldia{\varrho}$ for
$\mldia[\Evt(\Sigma)]{\varrho}$, $\dlbox{e(X)}{\psi}{\varrho}$ for
$\neg\dldia{e(X)}{\psi}{\neg\varrho}$, and $\truefrm$ for $\hybind{s}{s}$.

\begin{example}
An event/data formula can make two kinds of requirements on an event/data
structure: On the one hand, it can require the presence of certain mandatory
transitions, on the other hand it can require the absence of certain prohibited
transitions.  Considering the simple UML state machine in
\cref{fig:uml-counter}, the formula
\begin{equation*}
  \hyat{\state{s1}}{\dlwp{\evt{inc}(x)}{\attr{cnt} + x = 4}{\attr{cnt}' = 4}{\state{s2}}}
\end{equation*}
requires for each valuation of $\beta : \{ x \} \to \NZ$ such that $\attr{cnt} +
x = 4$ holds that there is a transition from control state $\state{s1}$ to
control state $\state{s2}$ for the instantiated event $\evt{inc}(\beta)$ where
$\attr{cnt}$ is changed to $4$.  On the other hand, the formula
\begin{equation*}
  \hyat{\state{s2}}{\dlbox{\evt{reset}}{\neg(\attr{cnt}' = 0)}{\falsefrm}}
\end{equation*}
prohibits any transitions out of $\state{s2}$ that are labelled with the event
$\evt{reset}$ but do not satisfy $\attr{cnt}' = 0$.

In the context of \cref{fig:uml-counter}, these formulæ only have
their explained intended meaning when $\state{s1}$ and $\state{s2}$
indeed refer to the eponymous states.  However, $\EDHML$ does not show
nominals for explicitly naming control states as part of the state
machine's interface and the reference to specific states always has to
build these states' context first using the modalities and the bind
operator.  On the other hand, as indicated in \cref{sec:ed-morphisms},
the inclusion of nominals may interfere disadvantageously with the
reduct formation.
\end{example}

Let $\sigma : \Sigma \to \Sigma'$ be an event/data signature morphism.  The
\emph{event/data formulæ translation} $\Frm[\EDHML]{\sigma}{S} :
\Frm[\EDHML]{\Sigma}{S} \to \Frm[\EDHML]{\Sigma'}{S}$ along $\sigma$ is
recursively given by
\begin{itemize}
  \item $\Frm[\EDHML]{\sigma}{S}(\varphi) =
\Frm[\DATA]{\Attr(\sigma)}{\emptyset}(\varphi)$;

  \item $\Frm[\EDHML]{\sigma}{S}(s) = s$;

  \item $\Frm[\EDHML]{\sigma}{S}(\hybind{s}{\varrho}) =
\hybind{s}{\Frm[\EDHML]{\sigma}{S \uplus \{ s \}}(\varrho)}$;

  \item $\Frm[\EDHML]{\sigma}{S}(\hyat[F]{s}{\varrho}) =
\hyat[\Evt(\sigma)(F)]{s}{\Frm[\EDHML]{\sigma}{S}(\varrho)}$;

  \item $\Frm[\EDHML]{\sigma}{S}(\mlbox[F]{\varrho}) =
\mlbox[\Evt(\sigma)(F)]{\Frm[\EDHML]{\sigma}{S}(\varrho)}$;

  \item $\Frm[\EDHML]{\sigma}{S}(\dldia{e(X)}{\psi}{\varrho}) =
\dldia{\Evt(\sigma)(e)(X)}{\Frm[2\DATA]{\Attr(\sigma)}{X}(\psi)}{\Frm[\EDHML]{\sigma}{S}(\varrho)}$;

  \item $\Frm[\EDHML]{\sigma}{S}(\dlwp{e(X)}{\phi}{\psi}{\varrho}) =
\dlwp{\Evt(\sigma)(e)(X)}{\Frm[\DATA]{\Attr(\sigma)}{X}(\phi)}{\Frm[2\DATA]{\Attr(\sigma)}{X}(\psi)}{\Frm[\EDHML]{\sigma}{S}(\varrho)}$;

  \item $\Frm[\EDHML]{\sigma}{S}(\neg \varrho) = \neg\Frm[\EDHML]{\sigma}{S}(\varrho)$;

  \item $\Frm[\EDHML]{\sigma}{S}(\varrho_1 \lor \varrho_2) =
\Frm[\EDHML]{\sigma}{S}(\varrho_1) \lor \Frm[\EDHML]{\sigma}{S}(\varrho_2)$.
\end{itemize}

The set $\Sen[\EDHML](\Sigma)$ of \emph{$\Sigma$-event/data sentences} is given
by $\Frm[\EDHML]{\Sigma}{\emptyset}$, the \emph{event/data sentence translation}
$\Sen[\EDHML](\sigma) : \Sen[\EDHML](\Sigma) \to \Sen[\EDHML](\Sigma')$ by
$\Frm[\EDHML]{\sigma}{\emptyset}$.

\subsection{Satisfaction Relation for $\EDHML$}

The $\EDHML$-satisfaction relation connects $\EDHML$-structures and
$\EDHML$-formulæ, expressing whether in some configuration of the structure a
particular formula holds with respect to an assignment of control states to
state variables.  Let $\Sigma$ be an event/data signature, $M$ a
$\Sigma$"=event/data structure, $S$ a set of state variables, $v : S \to
\CtrlSt(M)$ a state variable assignment, and $\gamma \in \Conf(M)$.  The
\emph{satisfaction relation} for event/data formulæ is inductively given by
\begin{itemize}[leftmargin=*, itemsep=2pt]
  \item $M, v, \gamma \models[\EDHML]_{\Sigma, S} \varphi$ iff
$\dataSt(M)(\gamma) \models[\DATA]_{\Attr(\Sigma)} \varphi$;

  \item $M, v, \gamma \models[\EDHML]_{\Sigma, S} s$ iff $v(s) = \ctrlSt(M)(\gamma)$;

  \item $M, v, \gamma \models[\EDHML]_{\Sigma, S} \hybind{s}{\varrho}$
iff $M, v\{ s \mapsto \ctrlSt(M)(\gamma) \}, \gamma \models[\EDHML]_{\Sigma, S
  \uplus \{ s \}} \varrho$;

  \item $M, v, \gamma \models[\EDHML]_{\Sigma, S} \hyat[F]{s}{\varrho}$ iff $M,
v, \gamma' \models[\EDHML]_{\Sigma, S} \varrho$\\
for all $\gamma' \in \Conf^F(M)$ with $\ctrlSt(M)(\gamma') = v(s)$;

  \item $M, v, \gamma \models[\EDHML]_{\Sigma, S} \mlbox[F]{\varrho}$ iff $M, v,
\gamma' \models[\EDHML]_{\Sigma, S} \varrho$ for all $\gamma' \in \Conf^F(M,
\gamma)$;

  \item $M, v, \gamma \models[\EDHML]_{\Sigma, S}
\dldia{e(X)}{\psi}{\varrho}$ iff there is a $\beta : X \to \DATA$ and a 
$\gamma' \in \Conf(M)$ such that $(\gamma, \gamma') \in
\Rel(M)_{e(\beta)}$, $(\dataSt(M)(\gamma), \dataSt(M)(\gamma')), \beta
\models[2\DATA]_{\Attr(\Sigma), X} \psi$, and $M, v, \gamma'
\models[\EDHML]_{\Sigma, S} \varrho$;

  \item $M, v, \gamma \models[\EDHML]_{\Sigma, S}
\dlwp{e(X)}{\phi}{\psi}{\varrho}$ iff for all $\beta : X \to \DATA$ with
$\dataSt(M)(\gamma), \beta \models[\DATA]_{\Attr(\Sigma), X} \phi$ there is
some $\gamma' \in \Conf(M)$ such that $(\gamma, \gamma') \in
\Rel(M)_{e(\beta)}$,\\ $(\dataSt(M)(\gamma), \dataSt(M)(\gamma')), \beta
\models[2\DATA]_{\Attr(\Sigma), X} \psi$, and $M, v, \gamma'
\models[\EDHML]_{\Sigma, S} \varrho$;

  \item $M, v, \gamma \models[\EDHML]_{\Sigma, S} \neg\varrho$ iff $M, v, \gamma
\not\models[\EDHML]_{\Sigma, S} \varrho$;

  \item $M, v, \gamma \models[\EDHML]_{\Sigma, S} \varrho_1 \lor \varrho_2$ iff
$M, v, \gamma \models[\EDHML]_{\Sigma, S} \varrho_1$ or $M, v, \gamma
\models[\EDHML]_{\Sigma, S} \varrho_2$.
\end{itemize}

This satisfaction relation is well-behaved with respect to reducts of
$\EDHML$-structures.  On the one hand, this is due to the use of abstract data
names rather than data states in the structures, and on the other hand to the
satisfaction condition of $\DATA$ and $2\DATA$.

\begin{lemma}\label{lem:sat-cond}
Let $\sigma : \Sigma \to \Sigma'$ be a event/data signature morphism and $M'$ a
$\Sigma'$-event/data structure.  For all $\varrho \in \Frm[\EDHML]{\Sigma}{S}$, all
$\gamma' \in \Conf(\reduct{M'}{\sigma}) \subseteq \Conf(M')$, and all $v : S \to
\CtrlSt(\reduct{M'}{\sigma}) \subseteq \CtrlSt(M')$ it holds that
\begin{equation*}
  \reduct{M'}{\sigma}, v, \gamma' \models[\EDHML]_{\Sigma, S} \varrho \iff M', v, \gamma' \models[\EDHML]_{\Sigma', S} \Frm[\EDHML]{\sigma}{S}(\varrho)
\ \text{.}
\end{equation*}
\end{lemma}
\begin{techreport}
\begin{proof}
We apply induction on the structure of $\Sigma$-event/data formulæ.
We only consider the cases $\varphi$, $s$, $\hybind{s}{\varrho}$,
$\hyat[F]{s}{\varrho}$, $\dldia{e(X)}{\psi}{\varrho}$, and $\dlwp{e(X)}{\phi}{\psi}{\varrho}$; negation and
disjunction are straightforward.

\abovedisplayskip4pt plus 2pt
\smallskip\noindent\textit{Case $\varphi$}:
\begin{align*} &
  \reduct{M'}{\sigma}, v, \gamma' \models[\EDHML]_{\Sigma, S} \varphi
\just{def.\ $\models[\EDHML]$}
  \dataSt(\reduct{M'}{\sigma})(\gamma') \models[\DATA]_{\Attr(\Sigma)} \varphi
\just{def.\ $\reductop\sigma$}
  \reduct{\dataSt(M')(\gamma')}{\sigma} \models[\DATA]_{\Attr(\Sigma)} \varphi
\just{sat.\ cond.\ $\DATA$}
  \dataSt(M')(\gamma') \models[\DATA]_{\Attr(\Sigma')} \Attr(\sigma)(\varphi)
\just{def.\ $\models[\EDHML]$}
  M', v, \gamma' \models[\EDHML]_{\Sigma', S} \Attr(\sigma)(\varphi)
\just{def.\ $\Frm[\EDHML]{\sigma}{S}$}
  M', v, \gamma' \models[\EDHML]_{\Sigma', S} \Frm[\EDHML]{\sigma}{S}(\varphi)
\end{align*}

\smallskip\noindent\textit{Case $s$}:
\begin{align*} &
  \reduct{M'}{\sigma}, v, \gamma' \models[\EDHML]_{\Sigma, S} s
\just{def.\ $\models[\EDHML]$}
  v(s) = \ctrlSt(\reduct{M'}{\sigma})(\gamma')
\just{def.\ $\reductop\sigma$}
  v(s) = \ctrlSt(M')(\gamma')
\just{def.\ $\models[\EDHML]$}
  M', v, \gamma' \models[\EDHML]_{\Sigma', S} s
\just{def.\ $\Frm[\EDHML]{\sigma}{S}$}
  M', v, \gamma' \models[\EDHML]_{\Sigma', S} \Frm[\EDHML]{\sigma}{S}(s)
\end{align*}

\smallskip\noindent\textit{Case $\hybind{s}{\varrho}$}:
\begin{align*} &
  \reduct{M'}{\sigma}, v, \gamma' \models[\EDHML]_{\Sigma, S} \hybind{s}{\varrho}
\just{def.\ $\models[\EDHML]$}
  \reduct{M'}{\sigma}, v\{ s \mapsto \ctrlSt(\reduct{M'}{\sigma})(\gamma') \}, \gamma' \models[\EDHML]_{\Sigma, S \uplus \{ s \}} \varrho
\just{def.\ $\reductop\sigma$ and I.\,H.}
  M', v\{ s \mapsto \ctrlSt(M')(\gamma') \}, \gamma' \models[\EDHML]_{\Sigma', S \uplus \{ s \}} \Frm[\EDHML]{\sigma}{S \uplus \{ s \}}(\varrho)
\just{def.\ $\models[\EDHML]$}
  M', v, \gamma' \models[\EDHML]_{\Sigma', S} \hybind{s}{\Frm[\EDHML]{\sigma}{S \uplus \{ s \}}(\varrho)}
\just{def.\ $\Frm[\EDHML]{\sigma}{S}$}
  M', v, \gamma' \models[\EDHML]_{\Sigma', S} \Frm[\EDHML]{\sigma}{S}(\hybind{s}{\varrho})
\end{align*}

\smallskip\noindent\textit{Case $\hyat[F]{s}{\varrho}$}:
\begin{align*} &
  \reduct{M'}{\sigma}, v, \gamma' \models[\EDHML]_{\Sigma, S} \hyat[F]{s}{\varrho}
\just{def.\ $\models[\EDHML]$}
  \reduct{M'}{\sigma}, v, \gamma'' \models[\EDHML]_{\Sigma, S} \varrho
\quad\text{for all $\gamma'' \in \Conf^F(\reduct{M'}{\sigma})$ with $\ctrlSt(\reduct{M'}{\sigma})(\gamma'') = v(s)$}
\just{\cref{lem:relative}\cref{it:lem:relative:reach}}
  \reduct{M'}{\sigma}, v, \gamma'' \models[\EDHML]_{\Sigma, S} \varrho\quad\text{for all $\gamma'' \in \Conf^{\Evt(\sigma)(F)}(M')$ with $\ctrlSt(\reduct{M'}{\sigma})(\gamma'') = v(s)$}
\just{def.\ $\reductop\sigma$ and I.\,H.}
  M', v, \gamma'' \models[\EDHML]_{\Sigma', S} \Frm[\EDHML]{\sigma}{S}(\varrho) \quad\text{for all $\gamma'' \in \Conf^{\Evt(\sigma)(F)}(M')$ with $\ctrlSt(M')(\gamma'') = v(s)$}
\just{def.\ $\models[\EDHML]$}
  M', v, \gamma' \models[\EDHML]_{\Sigma', S} \hyat[\Evt(\sigma)(F)]{s}{\Frm[\EDHML]{\sigma}{S}(\varrho)}
\just{def.\ $\Frm[\EDHML]{\sigma}{S}$}
  M', v, \gamma' \models[\EDHML]_{\Sigma', S} \Frm[\EDHML]{\sigma}{S}(\hyat[F]{s}{\varrho})
\end{align*}

\smallskip\noindent\textit{Case $\mlbox[F]{\varrho}$}:
\begin{align*} &
  \reduct{M'}{\sigma}, v, \gamma' \models[\EDHML]_{\Sigma, S} \mlbox[F]{\varrho}
\just{def.\ $\models[\EDHML]$}
  \reduct{M'}{\sigma}, v, \gamma'' \models[\EDHML]_{\Sigma, S} \varrho
\quad\text{for all $\gamma'' \in \Conf^F(\reduct{M'}{\sigma}, \gamma')$}
\just{\cref{lem:relative}\cref{it:lem:relative:reach-from}}
  \reduct{M'}{\sigma}, v, \gamma'' \models[\EDHML]_{\Sigma, S} \varrho\quad\text{for all $\gamma'' \in \Conf^{\Evt(\sigma)(F)}(M', \gamma')$}
\just{I.\,H.}
  M', v, \gamma'' \models[\EDHML]_{\Sigma', S} \Frm[\EDHML]{\sigma}{S}(\varrho) \quad\text{for all $\gamma'' \in \Conf^{\Evt(\sigma)(F)}(M', \gamma')$}
\just{def.\ $\models[\EDHML]$}
  M', v, \gamma' \models[\EDHML]_{\Sigma', S} \mlbox[\Evt(\sigma)(F)]{\Frm[\EDHML]{\sigma}{S}(\varrho)}
\just{def.\ $\Frm[\EDHML]{\sigma}{S}$}
  M', v, \gamma' \models[\EDHML]_{\Sigma', S} \Frm[\EDHML]{\sigma}{S}(\mlbox[F]{\varrho})
\end{align*}

\smallskip\noindent\textit{Case $\dldia{e(X)}{\psi}{\varrho}$}:
\begin{align*} &
  \reduct{M'}{\sigma}, v, \gamma' \models[\EDHML]_{\Sigma, S} \dldia{e(X)}{\psi}{\varrho}
\just{def.\ $\models[\EDHML]$}
  \reduct{M'}{\sigma}, v, \gamma'' \models[\EDHML]_{\Sigma, S} \varrho
\quad\begin{array}[t]{@{}l@{}}
  \text{for some $\beta : X \to \DATA$, $\gamma'' \in \Conf(\reduct{M'}{\sigma})$ with}\\
  \text{$(\gamma', \gamma'') \in \Rel(\reduct{M'}{\sigma})_{e(\beta)}$ and}\\[.5ex]
  \text{$(\dataSt(\reduct{M'}{\sigma})(\gamma'), \dataSt(\reduct{M'}{\sigma})(\gamma'')), \beta \models[2\DATA]_{\Attr(\Sigma), X} \psi$}
\end{array}
\just{\cref{lem:relative}\cref{it:lem:relative:trans}}
  \reduct{M'}{\sigma}, v, \gamma'' \models[\EDHML]_{\Sigma, S} \varrho
\quad\begin{array}[t]{@{}l@{}}
  \text{for some $\beta : X \to \DATA$, $\gamma'' \in \Conf(M')$ with}\\
  \text{$(\gamma', \gamma'') \in \Rel(M')_{\Evt(\sigma)(e)(\beta)}$ and}\\[.5ex]
  \text{$(\dataSt(\reduct{M'}{\sigma})(\gamma'), \dataSt(\reduct{M'}{\sigma})(\gamma'')), \beta \models[2\DATA]_{\Attr(\Sigma), X} \psi$}
\end{array}
\just{def. $\reductop\sigma$}
  \reduct{M'}{\sigma}, v, \gamma'' \models[\EDHML]_{\Sigma, S} \varrho
\quad\begin{array}[t]{@{}l@{}}
  \text{for some $\beta : X \to \DATA$, $\gamma'' \in \Conf(M')$ with}\\
  \text{$(\gamma', \gamma'') \in \Rel(M')_{\Evt(\sigma)(e)(\beta)}$ and}\\[.5ex]
  \text{$(\reduct{\dataSt(M')(\gamma')}{\sigma}, \reduct{\dataSt(M')(\gamma'')}{\sigma}), \beta \models[2\DATA]_{\Attr(\Sigma), X} \psi$}
\end{array}
\just{sat.\ cond.\ $2\DATA$}
  \reduct{M'}{\sigma}, v, \gamma'' \models[\EDHML]_{\Sigma, S} \varrho
\quad\begin{array}[t]{@{}l@{}}
  \text{for some $\beta : X \to \DATA$, $\gamma'' \in \Conf(M')$ with}\\
  \text{$(\gamma', \gamma'') \in \Rel(M')_{\Evt(\sigma)(e)(\beta)}$ and}\\[.5ex]
  \text{$(\dataSt(M')(\gamma'), \dataSt(M')(\gamma'')), \beta \models[2\DATA]_{\Attr(\Sigma'), X} \Frm[2\DATA]{\Attr(\sigma)}{X}(\psi)$}
\end{array}
\just{I.\,H.}
  M', v, \gamma'' \models[\EDHML]_{\Sigma', S} \Frm[\EDHML]{\sigma}{S}(\varrho)
\quad\begin{array}[t]{@{}l@{}}
  \text{for some $\beta : X \to \DATA$, $\gamma'' \in \Conf(M')$ with}\\
  \text{$(\gamma', \gamma'') \in \Rel(M')_{\Evt(\sigma)(e)(\beta)}$ and}\\[.5ex]
  \text{$(\dataSt(M')(\gamma'), \dataSt(M')(\gamma'')), \beta \models[2\DATA]_{\Attr(\Sigma'), X} \Frm[2\DATA]{\Attr(\sigma)}{X}(\psi)$}
\end{array}
\just{def.\ $\models[\EDHML]$}
  M', v, \gamma' \models[\EDHML]_{\Sigma', S} \dldia{\Evt(\sigma)(e)(X)}{\Frm[2\DATA]{\Attr(\sigma)}{X}(\psi)}{\Frm[\EDHML]{\sigma}{S}(\varrho)}
\just{def.\ $\Frm[\EDHML]{\sigma}{S}$}
  M', v, \gamma' \models[\EDHML]_{\Sigma', S} \Frm[\EDHML]{\sigma}{S}(\dldia{e(X)}{\psi}{\varrho})
\end{align*}

\smallskip\noindent\textit{Case $\dlwp{e(X)}{\phi}{\psi}{\varrho}$}:
\begin{align*} &
  \reduct{M'}{\sigma}, v, \gamma' \models[\EDHML]_{\Sigma, S} \dlwp{e(X)}{\phi}{\psi}{\varrho}
\just{def.\ $\models[\EDHML]$}
  \reduct{M'}{\sigma}, v, \gamma'' \models[\EDHML]_{\Sigma, S} \varrho
\\[-.5ex]&\qquad\begin{array}[t]{@{}l@{}}
  \text{for all $\beta : X \to \DATA$ such that $\dataSt(\reduct{M'}{\sigma})(\gamma'), \beta \models[\DATA]_{\Attr(\Sigma), X} \phi$ and}\\
  \text{some $\gamma'' \in \Conf(\reduct{M'}{\sigma})$ with $(\gamma', \gamma'') \in \Rel(\reduct{M'}{\sigma})_{e(\beta)}$ and}\\[.5ex]
  \text{$(\dataSt(\reduct{M'}{\sigma})(\gamma'), \dataSt(\reduct{M'}{\sigma})(\gamma'')), \beta \models[2\DATA]_{\Attr(\Sigma), X} \psi$}
\end{array}
\just{\cref{lem:relative}\cref{it:lem:relative:trans}}
  \reduct{M'}{\sigma}, v, \gamma'' \models[\EDHML]_{\Sigma, S} \varrho
\\[-.5ex]&\qquad\begin{array}[t]{@{}l@{}}
  \text{for all $\beta : X \to \DATA$ such that $\dataSt(\reduct{M'}{\sigma})(\gamma'), \beta \models[\DATA]_{\Attr(\Sigma), X} \phi$ and}\\
  \text{some $\gamma'' \in \Conf(M')$ with $(\gamma', \gamma'') \in \Rel(M')_{\Evt(\sigma)(e)(\beta)}$ and}\\[.5ex]
  \text{$(\dataSt(\reduct{M'}{\sigma})(\gamma'), \dataSt(\reduct{M'}{\sigma})(\gamma'')), \beta \models[2\DATA]_{\Attr(\Sigma), X} \psi$}
\end{array}
\just{def.\ $\reductop\sigma$}
  \reduct{M'}{\sigma}, v, \gamma'' \models[\EDHML]_{\Sigma, S} \varrho
\\[-.5ex]&\qquad\begin{array}[t]{@{}l@{}}
  \text{for all $\beta : X \to \DATA$ such that $\reduct{\dataSt(M')(\gamma')}{\sigma}, \beta \models[\DATA]_{\Attr(\Sigma), X} \phi$ and}\\
  \text{some $\gamma'' \in \Conf(M')$ with $(\gamma', \gamma'') \in \Rel(M')_{\Evt(\sigma)(e)(\beta)}$ and}\\[.5ex]
  \text{$(\reduct{\dataSt(M')(\gamma')}{\sigma}, \reduct{\dataSt(M')(\gamma'')}{\sigma}), \beta \models[2\DATA]_{\Attr(\Sigma), X} \psi$}
\end{array}
\just{sat.\ cond.\ $\DATA$, sat.\ cond.\ $2\DATA$}
  \reduct{M'}{\sigma}, v, \gamma'' \models[\EDHML]_{\Sigma, S} \varrho
\\[-.5ex]&\qquad\begin{array}[t]{@{}l@{}}
  \text{for all $\beta : X \to \DATA$ such that $\dataSt(M')(\gamma'), \beta \models[\DATA]_{\Attr(\Sigma'), X} \Frm[\DATA]{\Attr(\sigma)}{X}(\phi)$ and}\\
  \text{some $\gamma'' \in \Conf(M')$ with $(\gamma', \gamma'') \in \Rel(M')_{\Evt(\sigma)(e)(\beta)}$ and}\\[.5ex]
  \text{$(\dataSt(M')(\gamma'), \dataSt(M')(\gamma'')), \beta \models[2\DATA]_{\Attr(\Sigma'), X} \Frm[2\DATA]{\Attr(\sigma)}{X}(\psi)$}
\end{array}
\just{I.\,H.}
  M', v, \gamma'' \models[\EDHML]_{\Sigma', S} \Frm[\EDHML]{\sigma}{S}(\varrho)
\\[-.5ex]&\qquad\begin{array}[t]{@{}l@{}}
  \text{for all $\beta : X \to \DATA$ such that $\dataSt(M')(\gamma'), \beta \models[\DATA]_{\Attr(\Sigma'), X} \Frm[\DATA]{\Attr(\sigma)}{X}(\phi)$ and}\\
  \text{some $\gamma'' \in \Conf(M')$ with $(\gamma', \gamma'') \in \Rel(M')_{\Evt(\sigma)(e)(\beta)}$ and}\\[.5ex]
  \text{$(\dataSt(M')(\gamma'), \dataSt(M')(\gamma'')), \beta \models[2\DATA]_{\Attr(\Sigma'), X} \Frm[2\DATA]{\Attr(\sigma)}{X}(\psi)$}
\end{array}
\just{def.\ $\models[\EDHML]$}
  M', v, \gamma' \models[\EDHML]_{\Sigma', S} \dlwp{\Evt(\sigma)(e)(X)}{\Frm[\DATA]{\Attr(\sigma)}{X}(\phi)}{\Frm[2\DATA]{\Attr(\sigma)}{X}(\psi)}{\Frm[\EDHML]{\sigma}{S}(\varrho)}
\just{def.\ $\Frm[\EDHML]{\sigma}{S}$}
  M', v, \gamma' \models[\EDHML]_{\Sigma', S} \Frm[\EDHML]{\sigma}{S}(\dlwp{e(X)}{\phi}{\psi}{\varrho})
\end{align*}
\end{proof}
\end{techreport}

For a $\Sigma \in |\Sig[\EDHML]|$, an $M \in |\Str[\EDHML](\Sigma)|$, and a
$\rho \in \Sen[\EDHML](\Sigma)$ the \emph{satisfaction relation} $M
\models[\EDHML]_{\Sigma} \rho$ holds if, and only if, $M, \emptyset, \gamma_0
\models[\EDHML]_{\Sigma, \emptyset} \rho$ for all $\gamma_0 \in \Conf_0(M)$.

\begin{theorem}\label{thm:sat-cond}
$(\Sig[\EDHML], \Str[\EDHML], \Sen[\EDHML], {\models[\EDHML]})$ is an institution.
\end{theorem}
\begin{techreport}
\begin{proof}
The satisfaction condition that for any $\sigma : \Sigma \to \Sigma'$ in
$\Sig[\EDHML]$, $M' \in |\Str[\EDHML](\Sigma')|$, and $\rho \in
\Sen[\EDHML](\Sigma)$, it holds that
\begin{equation*}
  \Str[\EDHML](\sigma)(M') \models[\EDHML]_{\Sigma} \rho
\iff
  M' \models[\EDHML]_{\Sigma'} \Sen[\EDHML](\sigma)(\rho)
\end{equation*}
directly follows from \cref{lem:sat-cond}.
\end{proof}
\end{techreport}

\subsection{Representing Simple UML State Machines in $\EDHML$}\label{sec:simple_uml_state_machines}

\begin{algorithm}[t]
\caption{Constructing an $\EDHML$-sentence from a set of transition specifications}
\label{alg:sen-op-spec}
\begin{algorithmic}[1]
\Require $\raisebox{8pt}{} T \equiv \text{a set of transition specifications}$
\Requirx $\mathit{Im}_T(c) = \{ (\phi, e(X), \psi, c') \mid (c, \phi, e(X), \psi, c') \in T \}$
\Requirx $\mathit{Im}_T(c, e(X)) = \{ (\phi, \psi, c') \mid (c, \varphi, e(X), \psi, c') \in T \}$
\Statex
\Function {$\mathrm{sen}$}{$c, I, V, B$} \Comment{$c$: state, $I$: image to visit, $V$: states to visit, $B$: bound states}
  \If{$I \neq \emptyset$}
    \State $(\phi, e(X), \psi, c') \gets \Choose\ I$
    \If{$c' \in B$}
      \State \Return $\hyat{c}{\dlwp{e(X)}{\phi}{\psi}{(c' \land \mathrm{sen}(c, I \setminus \{ (\phi, e(X), \psi, c') \}, V, B))}}$
    \Else
      \State \Return $\hyat{c}{\dlwp{e(X)}{\phi}{\psi}{(\hybind{c'}{\mathrm{sen}(c, I \setminus \{ (\phi, e(X), \psi, c') \}, V, B \cup \{ c' \}))}}}$
    \EndIf
  \EndIf
  \State $V \gets V \setminus \{ c \}$
  \If{$V \neq \emptyset$}
    \State $c' \gets \Choose\ B \cap V$
    \State \Return $\mathrm{sen}(c', \mathit{Im}_T(c'), V, B)$
  \EndIf
  \State \Return $(\bigwedge_{c \in B} \mathrm{fin}(c)) \land \bigwedge_{c_1 \in B, c_2 \in B \setminus \{ c_1 \}} \neg\hyat{c_1}{c_2}$
\EndFunction
\Statex
\Function {$\mathrm{fin}$}{$c$}
\State \Return $(\hyatop c) \begin{array}[t]{@{}l@{}}
                                \bigwedge_{e(X) \in \Evt(\EDSign(U))} \bigwedge_{P \subseteq \mathit{Im}_T(c, e(X))}\\{}
                                [e(X)\nfatslash\begin{array}[t]{@{}l@{}}
                                              \big(\bigwedge_{(\phi, \psi, c') \in P} (\phi \land \psi)\big) \land{}\\
                                              \neg\big(\bigvee_{(\phi, \psi, c') \in \mathit{Im}_T(c, e(X)) \setminus P} (\phi \land \psi)\big)
                                ] \big(\bigvee_{(\phi, \psi, c') \in P} c'\big)
                                            \end{array}
                              \end{array}$
\EndFunction
\end{algorithmic}
\end{algorithm}
The hybrid modal logic $\EDHML$ is expressive enough to characterise the model
class of a simple UML state machine $U$ by a single sentence $\varrho_U$, i.e.,
an event/data structure $M$ is a model of $U$ if, and only if, $M
\models[\EDHML]_{\EDSign(U)} \varrho_U$.  Such a characterisation is achieved by
means of \cref{alg:sen-op-spec} that is a slight variation of the
char\-ac\-ter\-i\-sa\-tion algorithm for so-called operational specifications
within $\EDHDL$~\cite{hennicker-madeira-knapp:fase:2019} by including also
events with data arguments.  The algorithm constructs a sentence expressing that
semantic transitions according to explicit syntactic transition specifications
are indeed possible and that no other semantic transitions not adhering to any
of the syntactic transition specifications exist. For a set of transition
specifications $T$, a call $\mathrm{sen}(c,\allowbreak I,\allowbreak
V,\allowbreak B)$ performs a recursive breadth-first traversal starting from
$c$, where $I$ holds the unprocessed quadruples $(\phi,\allowbreak
e(X),\allowbreak \psi,\allowbreak c')$ of transitions in $T$ outgoing from $c$,
$V$ the remaining states to visit, and $B$ the set of already bound states.  The
function first requires the existence of each outgoing transition of $I$ in the
resulting formula, binding any newly reached state.  Having visited all states
in $V$, it requires that no other transitions from the states in $B$ exist using
calls to $\mathrm{fin}$, and adds the requirement that all states in $B$ are
pairwise different.  Formula $\mathrm{fin}(c)$ expresses that at $c$, for all
events $e(X)$ and for all subsets $P$ of the transitions in $T$ outgoing from
$c$, whenever an $e(X)$-transition can be done with the combined effect of $P$
but not adhering to any of the effects of the currently not selected
transitions, the $e(X)$-transition must have one of the states as its target
that are target states of $P$.

\begin{example}\label{ex:machine-axioms} 
Applying \cref{alg:sen-op-spec} to the set of explicitly mentioned, ``black''
transition specifications $T$ of the simple UML state machine $\mathit{Counter}$
in \cref{fig:uml-counter}, i.e., calling $\splitatcommas{\mathrm{sen}(\state{s1},
\mathit{Im}_{T}(\state{s1}), \{ \state{s1}, \state{s2} \}, \{ \state{s1} \})}$
yields $\varrho_{\state{s1}, \state{s1}}$ with
\begin{gather*}
  \varrho_{\state{s1},\state{s1}} = \hyat{\state{s1}}{\dlwp{\evt{inc}(x)}{\attr{cnt} + x \leq 4}{\attr{cnt}' = \attr{cnt} + x}{(\state{s1} \land \varrho_{\state{s1},\state{s2}})}}
\\
  \varrho_{\state{s1},\state{s2}} = \hyat{\state{s1}}{\dlwp{\evt{inc}(x)}{\attr{cnt} + x = 4}{\attr{cnt}' = \attr{cnt} + x}{\hybind{\state{s2}}{(\varrho_{\state{s2},\state{s1}})}}}
\\
  \varrho_{\state{s2},\state{s1}} = \hyat{\state{s2}}{\dlwp{\evt{reset}}{\truefrm}{\attr{cnt}' = 0}{(\state{s1} \land \varrho_{\mathrm{fin}})}}
\\
  \varrho_{\mathrm{fin}} = \varrho_{\mathrm{fin}(\state{s1})} \land \varrho_{\mathrm{fin}(\state{s2})} \land \neg\hyat{\state{s1}}{\state{s2}}
\\
  \varrho_{\mathrm{fin}(\state{s1})} = \hyat{\state{s1}}{\big(}%
  [\evt{inc}(x) \nfatslash \neg(\stacked{(\attr{cnt} + x \leq 4 \land \attr{cnt}' = \attr{cnt} + x) \lor{}\\ (\attr{cnt} + x = 4 \land \attr{cnt}' = 4))]\falsefrm \land{}}
\\[-4pt]\phantom{\varrho_{\mathrm{fin}(\state{s1})} = \hyat{\state{s1}}{\big(}}%
  [\evt{inc}(x) \nfatslash \stacked{(\attr{cnt} + x \leq 4 \land \attr{cnt}' = \attr{cnt} + x) \land{}\\ \neg(\attr{cnt} + x = 4 \land \attr{cnt}' = 4)]\state{s1} \land{}}
\\[-4pt]\phantom{\varrho_{\mathrm{fin}(\state{s1})} = \hyat{\state{s1}}{\big(}}%
  [\evt{inc}(x) \nfatslash \stacked{(\attr{cnt} + x = 4 \land \attr{cnt}' = 4) \land{}\\ \neg(\attr{cnt} + x \leq 4 \land \attr{cnt}' = \attr{cnt} + x)]\state{s2} \land{}}
\\[-4pt]\phantom{\varrho_{\mathrm{fin}(\state{s1})} = \hyat{\state{s1}}{\big(}}%
  [\evt{inc}(x) \nfatslash \stacked{(\attr{cnt} + x \leq 4 \land \attr{cnt}' = \attr{cnt} + x) \land{}\\ (\attr{cnt} + x = 4 \land \attr{cnt}' = 4)](\state{s1} \lor \state{s2}) \land{}}
\\[-4pt]\phantom{\varrho_{\mathrm{fin}(\state{s1})} = \hyat{\state{s1}}{\big(}}%
  \dlbox{\evt{reset}}{\truefrm}{\falsefrm}%
\big)
\\
\varrho_{\mathrm{fin}(\state{s2})} = \hyat{\state{s2}}{\big(%
\begin{array}[t]{@{}l@{}}
  \dlbox{\evt{inc}(x)}{\truefrm}{\falsefrm} \land{}\\
  \dlbox{\evt{reset}}{\neg(\attr{cnt}' = 0)}{\falsefrm} \land{}\\
  \dlbox{\evt{reset}}{\attr{cnt}' = 0}{\state{s1}}%
\big)
\end{array}}
\end{gather*}

In fact, there is no outgoing ``black'' transition for $\evt{reset}$ from
$\state{s1}$, thus $P = \emptyset$ is the only choice for this event in
$\mathrm{fin}(\state{s1})$ and the clause
$\dlbox{\evt{reset}}{\truefrm}{\falsefrm}$ is included.  For $\evt{inc}(x)$
there are two outgoing transitions resulting four different clauses checking
whether none, the one or the other, or both transitions are executable.
\end{example}

In order to apply the algorithm to simple UML state machines, the idling
self-loops for achieving input-enabledness first have to be made explicit.  For
a syntactically input-enabled simple UML state machine $U$ a characterising
sentence then reads
\begin{equation*}
  \varrho_U
=
  \hybind{c_0}{\varphi_0 \land \mathrm{sen}(c_0, \mathit{Im}_{\Trans(U)}(c_0), \CtrlSt(U), \{ c_0 \})}
\ \text{,}
\end{equation*}
where $c_0 = \ctrlSt_0(U)$ and $\varphi_0 = \datapred_0(U)$.  Due to syntactic
reachability, the bound states $B$ of \cref{alg:sen-op-spec} become $\CtrlSt(U)$
when $\mathrm{sen}$ is called for $B = \{ \ctrlSt_0(U) \}$ and $V$ reaches
$\emptyset$.

\section{A Theoroidal Comorphism from $\EDHML$ to $\CASL$}\label{sec:comorphism}

We define a theoroidal comorphism from $\EDHML$ to $\CASL$.  The construction
mainly follows the standard translation of modal logics to first-order
logic~\cite{blackburn-de-rijke-venema:2001} which has been considered for hybrid
logics also on an institutional level~\cite{madeira:phd:2013,diaconescu-madeira:mscs:2016}.

\begin{figure}[!t]
\begin{hetcasl}
\KW{from} $\textit{Basic/StructuredDatatypes}$ \KW{get} $\textsc{Set}$ \% $\text{import finite sets}$\\
\KW{spec} $\textsc{Trans}_{\Sigma} = \Data$\\
\KW{then} \=\KW{free} \KW{type} $\mathrm{Evt} \cln\cln= \tau_{\mathrm{e}}(\Evt(\Sigma))$\\
\hspace*{3cm}\% $\tau_{\mathrm{e}}(\{ e(X) \}) = e(\data^{|X|})$, $\tau_{\mathrm{e}}(\{ e(X) \} \cup E) = e(\data^{|X|}) \mid \tau_{\mathrm{e}}(E)$\\
\>\KW{free} \KW{type} $\mathrm{EvtNm} \cln\cln= \tau_{\mathrm{n}}(\Evt(\Sigma))$
\% $\tau_{\mathrm{n}}(\{ e(X) \}) = e$, $\tau_{\mathrm{n}}(\{ e(X) \} \cup E) = e \mid \tau_{\mathrm{n}}(E)$\\
\>\KW{op} $\mathrm{nm} \cln \mathrm{Evt} \to \mathrm{EvtNm}$\\
\>\KW{axiom} $\forall x_1, \ldots, x_n \cln \mathit{dt} \axdot \mathrm{nm}(e(x_1, \ldots, x_n)) = e$
\% for each $e(x_1, \ldots, x_n) \in \Evt(\Sigma)$\\
\KW{then} $\textsc{Set}[$sort $\mathrm{EvtNm}]$\\
\KW{then} \=\KW{sort} $\mathrm{Ctrl}$\\
\>\KW{free type} $\mathrm{Conf} \cln\cln= \mathrm{conf}(\mathrm{c} \cln \mathrm{Ctrl}; \tau_{\mathrm{a}}(\Attr(\Sigma)))$\\
\hspace*{3cm}\% $\tau_{\mathrm{a}}(\{ a \}) = a \cln \data$, $\tau_{\mathrm{a}}(\{ a \} \cup A) = a \cln \data; \tau_{\mathrm{a}}(A)$\\
\>\KW{preds} \=$\mathrm{init} \cln \mathrm{Conf}$;\\
\>\>$\mathrm{trans} \cln \mathrm{Conf} \times \mathrm{Evt} \times \mathrm{Conf}$\\
\>$\axdot\,\exists g \cln \mathrm{Conf} \axdot \mathrm{init}(g)$ \% $\textrm{there is some initial configuration}$\\
\>$\axdot\, \forall g, g' \cln \mathrm{Conf} \axdot \mathrm{init}(g) \land \mathrm{init}(g') \Rightarrow \mathrm{c}(g) = \mathrm{c}(g')$ \% $\textrm{single initial control state}$\\
\>\KW{free} $\{$ \=\KW{pred} $\mathrm{reachable} \cln \mathit{Set}[\mathrm{EvtNm}] \times \mathrm{Conf} \times \mathrm{Conf}$\\
\>\>$\forall g, g', g'' \cln \mathrm{Conf}, E \cln \mathit{Set}[\mathrm{EvtNm}], e \cln \mathrm{Evt}$\\
\>\>$\axdot$ $\mathrm{reachable}(E, g, g)$\\
\>\>$\axdot$ $\mathrm{reachable}(E, g, g') \land \mathrm{nm}(e) \in E \land \mathrm{trans}(g', e, g'') \Rightarrow \mathrm{reachable}(E, g, g'')$ $\}$\\
\>\KW{then} \KW{preds} \=$\mathrm{reachable}(E \cln \mathit{Set}[\mathrm{EvtName}], g \cln \mathrm{Conf}) \Leftrightarrow{}$\\
\hspace*{5cm}$\exists g_0 \cln \mathrm{Conf} \axdot \mathrm{init}(g_0) \land \mathrm{reachable}(E, g_0, g)$;\\
\>\>$\mathrm{reachable}(g \cln \mathrm{Conf}) \Leftrightarrow \mathrm{reachable}(\Evt(\Sigma), g)$\\
\KW{end}
\end{hetcasl}
\vskip-18pt
\caption{Frame for translating $\EDHML$ into $\CASL$}\label{lst:edhml-trans}
\end{figure}
The basis is a representation of $\EDHML$-signatures and the frame given by
$\EDHML$-structures as a $\CASL$-specification as shown in
\cref{lst:edhml-trans}.  The signature translation
\begin{equation*}
  \nu^{\Sig} : \Sig[\EDHML] \to \Pres[\CASL]
\end{equation*}
maps a $\EDHML$-signature $\Sigma$ to the $\CASL$-theory presentation given by
$\textsc{Trans}_{\Sigma}$ and a $\EDHML$-signature morphism to the corresponding
theory presentation morphism.  $\textsc{Trans}_{\Sigma}$ first of all covers the
events and event names according to $\Evt(\Sigma)$ (types $\mathrm{Evt}$ and
$\mathrm{EvtNm}$ with several alternatives separated by ``$|$'') and the
configurations (type $\mathrm{Conf}$ with a single constructor
``$\mathrm{conf}$'') with their control states (sort $\mathrm{Ctrl}$) and data
states given by assignments to the attributes from $\Attr(\Sigma)$ (separated by
``;'').  The remainder of $\textsc{Trans}_{\Sigma}$ sets the frame for
describing reachable transition systems with a set of initial configurations
(predicate $\mathrm{init}$), a transition relation (predicate $\mathrm{trans}$)
and reachability predicates.  The specification of the predicate
$\mathrm{reachable}$ uses \CASL's ``structured free'' construct to ensure
reachability to be inductively defined.  The model translation
\begin{equation*}
  \nu^{\Mod}_{\Sigma} : \Mod[\CASL](\nu^{\Sig}(\Sigma)) \to \Str[\EDHML](\Sigma)
\end{equation*}
then can rely on this encoding.  In particular, for a model $M' \in
\Mod[\CASL](\nu^{\Sig}(\Sigma))$, there are, using the bijection $\iota_{M',
  \data} : \data^{M'} \isorel \DATA$, an injective map $\iota_{M',
  \mathrm{Conf}} : \mathrm{Conf}^{M'} \injto \mathrm{Ctrl}^{M'} \times
\DataSt(\Attr(\Sigma))$ and a bijective map $\iota_{M', \mathrm{Evt}} :
\mathrm{Evt}^{M'} \isorel \{ e(\beta) \mid e(X) \in \Evt(\Sigma),\ \beta : X \to
\DATA \}$.  The $\EDHML$-structure resulting from a $\CASL$-model of
$\textsc{Trans}_{\Sigma}$ can thus be defined by
\begin{itemize}[itemsep=2pt]
  \item $\Conf(\nu^{\Mod}_{\Sigma}(M')) = \iota_{M', \mathrm{Conf}}^{-1}(\{ g' \in M'_{\mathrm{Conf}} \mid \mathrm{reachable}^{M'}(g') \})$

  \item $R(\nu^{\Mod}_{\Sigma}(M'))_{e(\beta)} = \{
\begin{array}[t]{@{}l@{}}
  (\gamma, \gamma') \in \Conf(\nu^{\Mod}_{\Sigma}(M')) \times \Conf(\nu^{\Mod}_{\Sigma}(M')) \mid{}\\
  \quad \mathrm{trans}^{M'}(\iota_{M', \mathrm{Conf}}(\gamma), \iota_{M', \mathrm{Evt}}^{-1}(e(\beta)), \iota_{M', \mathrm{Conf}}(\gamma')) \})
\end{array}$

  \item $\Conf_0(\nu^{\Mod}_{\Sigma}(M')) = \{ \gamma \in \Conf(\nu^{\Mod}_{\Sigma}(M')) \mid \mathrm{init}^{M'}(\iota_{M', \mathrm{Conf}}(\gamma)) \})$

  \item $\dataSt(\nu^{\Mod}_{\Sigma}(M')) = \{ (c, \omega) \in \Conf(\nu^{\Mod}_{\Sigma}(M')) \mapsto \omega \}$
\end{itemize}

For $\EDHML$-sentences, we first define a formula translation
\begin{equation*}
  \nu^{\Fm}_{\Sigma, S, g} : \Frm[\EDHML]{\Sigma}{S} \to \Frm[\CASL]{\nu^{\Sign}(\Sigma)}{S \cup \{ g \}}
\end{equation*}
which, mimicking the standard translation, takes a variable $g \cln
\mathrm{Conf}$ as a parameter that records the ``current configuration'' and
also uses a set $S$ of state names for the control states.  The translation
embeds the data state and 2-data state formulæ using the substitution
$\Attr(\Sigma)(g) = \{ a \mapsto a(g) \mid a \in \Attr(\Sigma) \}$ for replacing
the attributes $a \in \Attr(\Sigma)$ by the accessors $a(g)$.  The translation of $\EDHML$-formulæ then reads
\begin{itemize}[itemsep=2pt]
  \item $\nu^{\Fm}_{\Sigma, S, g}(\varphi) = \Frm[\CASL]{\nu^{\Sign}(\Sigma)}{\Attr(\Sigma)(g)}(\varphi)$

  \item $\nu^{\Fm}_{\Sigma, S, g}(s) = (s = \mathrm{c}(g))$

  \item $\nu^{\Fm}_{\Sigma, S, g}(\hybind{s}{\varrho}) = \exists s \cln \mathrm{Ctrl}
\,.\, s = \mathrm{c}(g) \land \nu^{\Fm}_{\Sigma, S \uplus \{ s \}, g}(\varrho)$

  \item $\nu^{\Fm}_{\Sigma, S, g}(\hyat[F]{s}{\varrho}) = \forall g' \cln
\mathrm{Conf} \,.\,
  (\mathrm{c}(g') = s \land \mathrm{reachable}(F, g')) \Rightarrow
  \nu^{\Fm}_{\Sigma, S, g'}(\varrho)$

  \item $\nu^{\Fm}_{\Sigma, S, g}(\mlbox[F]{\varrho}) = \forall g' \cln
\mathrm{Conf} \,.\,
  \mathrm{reachable}(F, g, g') \Rightarrow
  \nu^{\Fm}_{\Sigma, S, g'}(\varrho)$

  \item $\nu^{\Fm}_{\Sigma, S, g}(\dldia{e(X)}{\psi}{\varrho}) =
\stacked{\exists X \cln \data \,.\, \exists g' \cln \mathrm{Conf} \,.\, \mathrm{trans}(g, e(X), g') \land{}\\
         \quad\Frm[\CASL]{\nu^{\Sign}(\Sigma)}{\Attr(\Sigma)(g) \cup \Attr(\Sigma)(g') \cup \id{X}}(\psi) \land \nu^{\Fm}_{\Sigma, S, g'}(\varrho)}$

  \item $\nu^{\Fm}_{\Sigma, S, g}(\dlwp{e(X)}{\phi}{\psi}{\varrho}) =
\stacked{\forall X \cln \data \,.\, \Frm[\CASL]{\nu^{\Sign}(\Sigma)}{\Attr(\Sigma)(g) \cup \id{X}}(\phi) \Rightarrow{}\\
         \quad\exists g' \cln \mathrm{Conf} \,.\, \mathrm{trans}(g, e(X), g') \land{}\\
         \quad\quad\Frm[\CASL]{\nu^{\Sign}(\Sigma)}{\Attr(\Sigma)(g) \cup \Attr(\Sigma)(g') \cup \id{X}}(\psi) \land \nu^{\Fm}_{\Sigma, S, g'}(\varrho)}$

  \item $\nu^{\Fm}_{\Sigma, S, g}(\neg \varrho) = \neg \nu^{\Fm}_{\Sigma, S, g}(\varrho)$

  \item $\nu^{\Fm}_{\Sigma, S, g}(\varrho_1 \lor \varrho_2) = \nu^{\Fm}_{\Sigma, S, g}(\varrho_1) \lor \nu^{\Fm}_{\Sigma, S, g}(\varrho_2)$
\end{itemize}

\begin{example}
The translation of $\hyat{\state{s1}}{\dlwp{\evt{inc}(x)}{\attr{cnt} + x
    \leq 4}{\attr{cnt}' = \attr{cnt} + x}{\state{s1}}}$ over the state set
$\{ \state{s1} \}$ and the configuration variable $g$ is
\begin{align*}\quad&
  \nu^{\Fm}_{\Sigma, \{ \state{s1} \}, g}(\hyat{\state{s1}}{\dlwp{\evt{inc}(x)}{\attr{cnt} + x \leq 4}{\attr{cnt}' = \attr{cnt} + x}{\state{s1}}})
\\&\llap{$=$\ \ }{}
\begin{array}[t]{@{}l@{}}
  \forall g' \cln \mathrm{Conf} \,.\, (\mathrm{c}(g') = \state{s1} \land \mathrm{reachable}(g')) \Rightarrow{}\\
  \qquad\nu^{\Fm}_{\Sigma, \{ \state{s1} \}, g'}(\dlwp{\evt{inc}(x)}{\attr{cnt} + x \leq 4}{\attr{cnt}' = \attr{cnt} + x}{\state{s1}})
\end{array}
\\&\llap{$=$\ \ }{}
\begin{array}[t]{@{}l@{}}
  \forall g' \cln \mathrm{Conf} \,.\, (\mathrm{c}(g') = \state{s1} \land \mathrm{reachable}(g')) \Rightarrow{}\\
  \qquad\forall x \cln \data \,.\, \mathrm{cnt}(g') + x \leq 4 \Rightarrow{}\\
  \qquad\qquad\exists g'' \cln \mathrm{Conf} \,.\,
\begin{array}[t]{@{}l@{}}
  \mathrm{trans}(g', \mathrm{inc}(x), g'') \land{}\\
  \mathrm{cnt}(g'') = \mathrm{cnt}(g') + x \land
  \nu^{\Fm}_{\Sigma, \{ \state{s1} \}, g''}(\state{s1})
\end{array}
\end{array}
\\&\llap{$=$\ \ }{}
\begin{array}[t]{@{}l@{}}
  \forall g' \cln \mathrm{Conf} \,.\, (\mathrm{c}(g') = \state{s1} \land \mathrm{reachable}(g')) \Rightarrow{}\\
  \qquad\forall x \cln \data \,.\, \mathrm{cnt}(g') + x \leq 4 \Rightarrow{}\\
  \qquad\qquad\exists g'' \cln \mathrm{Conf} \,.\,
\begin{array}[t]{@{}l@{}}
  \mathrm{trans}(g', \mathrm{inc}(x), g'') \land{}\\
  \mathrm{cnt}(g'') = \mathrm{cnt}(g') + x \land
  \state{s1} = \mathrm{c}(g'')
\end{array}
\end{array}
\end{align*}
\end{example}



Building on the translation of formulæ, the sentence translation
\begin{equation*}
  \nu^{\Sen}_{\Sigma} : \Sen[\EDHML](\Sigma) \to \Sen[\CASL](\nu^{\Sign}(\Sigma))
\end{equation*}
only has to require additionally that evaluation starts in an initial state:
\begin{itemize}
  \item $\nu^{\Sen}_{\Sigma}(\rho) = \forall g \cln \mathrm{Conf} \,.\, \mathrm{init}(g) \Rightarrow \nu^{\Fm}_{\Sigma, \emptyset, g}(\rho)$
\end{itemize}

The translation of $\CASL$-models of $\textsc{Trans}_{\Sigma}$ into
$\EDHML$-structures and the translation of $\EDHML$-formulæ into $\CASL$-formulæ
over $\textsc{Trans}_{\Sigma}$ fulfil the requirements of the ``open''
satisfaction condition of theoroidal comorphisms:

\begin{lemma}\label{lem:co-morph}
For a $\varrho \in \Frm[\EDHML]{\Sigma}{S}$, an $M' \in
\Mod[\CASL](\nu^{\Sig}(\Sigma))$, a $v : S \to
\CtrlSt(\nu^{\Mod}_{\Sigma}(M'))$, and a $\gamma \in
\Conf(\nu^{\Mod}_{\Sigma}(M'))$ it holds with $\beta'_{M', g}(v, \gamma) =
\iota_{M', \mathrm{Ctrl}}^{-1} \compfun v \cup \{ g \mapsto \iota_{M',
  \mathrm{Conf}}(\gamma) \}$ that
\begin{equation*}
  \nu^{\Mod}_{\Sigma}(M'), v, \gamma \models[\EDHML]_{\Sigma, S} \varrho
\iff
  M', \beta'_{M', g}(v, \gamma) \models[\CASL]_{\nu^{\Sign}(\Sigma), S \cup \{ g \}} \nu^{\Fm}_{\Sigma, S, g}(\varrho)
\ \text{.}
\end{equation*}
\end{lemma}
\begin{techreport}
\begin{proof}
We apply induction on the structure of $\Sigma$-event/data formulæ.  We only
consider the cases $\varphi$, $s$, $\hybind{s}{\varrho}$,
$\hyat[F]{s}{\varrho}$, $\dldia{e(X)}{\psi}{\varrho}$, and
$\dlwp{e(X)}{\phi}{\psi}{\varrho}$; negation and disjunction are
straightforward.

\abovedisplayskip4pt
\smallskip\noindent\textit{Case $\varphi$}:
\begin{align*} &
  \nu^{\Mod}_{\Sigma}(M'), v, \gamma \models[\EDHML]_{\Sigma, S} \varphi
\just{def.\ $\models[\EDHML]$}
  \dataSt(\nu^{\Mod}_{\Sigma}(M'))(\gamma) \models[\DATA]_{\Attr(\Sigma)} \varphi
\just{def.\ $\models[\DATA]$, def.\ $\models[\CASL]$}
  M', \stacked{\iota^{-1}_{M', \data} \compfun \dataSt(\nu^{\Mod}_{\Sigma}(M'))(\gamma) \cup{}\\[.5ex]
               \{ g \mapsto \iota_{M', \mathrm{Conf}}(\gamma) \} \models[\CASL]_{\nu^{\Sign}(\Sigma), \{ g \}} \Frm[\CASL]{\nu^{\Sign}(\Sigma)}{\Attr(\Sigma)(g)}(\varphi)}
\just{def.\ $\models[\CASL]$}
  M', \beta'_{M', g}(v, \gamma) \models[\CASL]_{\nu^{\Sign}(\Sigma), S \cup \{ g \}}  \Frm[\CASL]{\nu^{\Sign}(\Sigma)}{\Attr(\Sigma)(g)}(\varphi)
\just{def.\ $\nu^{\Fm}$}
  M', \beta'_{M', g}(v, \gamma) \models[\CASL]_{\nu^{\Sign}(\Sigma), S \cup \{ g \}} \nu^{\Fm}_{\Sigma, S, g}(\varphi)
\end{align*}

\smallskip\noindent\textit{Case $s$}:
\begin{align*} &
  \nu^{\Mod}_{\Sigma}(M'), v, \gamma \models[\EDHML]_{\Sigma, S} s
\just{def.\ $\models[\EDHML]$}
  v(s) = \ctrlSt(\nu^{\Mod}_{\Sigma}(M'))(\gamma)
\just{def.\ $\nu^{\Mod}$, $\models[\CASL]$}
  M', \beta'_{M', g}(v, \gamma) \models[\CASL]_{\nu^{\Sign}(\Sigma), S \cup \{ g \}} s = \mathrm{c}(g)
\just{def.\ $\nu^{\Fm}$}
  M', \beta'_{M', g}(v, \gamma) \models[\CASL]_{\nu^{\Sign}(\Sigma), S \cup \{ g \}} \nu^{\Fm}_{\Sigma, S, g}(s)
\end{align*}

\smallskip\noindent\textit{Case $\hybind{s}{\varrho}$}:
\begin{align*} &
  \nu^{\Mod}_{\Sigma}(M'), v, \gamma \models[\EDHML]_{\Sigma, S} \hybind{s}{\varrho}
\just{def.\ $\models[\EDHML]$}
  \nu^{\Mod}_{\Sigma}(M'), v\{ s \mapsto \ctrlSt(\nu^{\Mod}_{\Sigma}(M'))(\gamma) \}, \gamma \models[\EDHML]_{\Sigma, S \uplus \{ s \}} \varrho
\just{I.\,H.}
  M', \beta'_{M', g}(v\{ s \mapsto \ctrlSt(\nu^{\Mod}_{\Sigma}(M'))(\gamma) \}, \gamma) \models[\CASL]_{\nu^{\Sign}(\Sigma), (S \uplus \{ s \}) \cup \{ g \}} \nu^{\Fm}_{\Sigma, S \uplus \{ s \}, g}(\varrho)
\just{def.\ $\models[\CASL]$}
  M', \beta'_{M', g}(v, \gamma) \models[\CASL]_{\nu^{\Sign}(\Sigma), S \cup \{ g \}} \exists s \cln \mathrm{Ctrl} \,.\, s = \mathrm{c}(g) \land \nu^{\Fm}_{\Sigma, S \uplus \{ s \}, g}(\varrho)
\just{def.\ $\nu^{\Fm}$}
  M', \beta'_{M', g}(v, \gamma) \models[\CASL]_{\nu^{\Sign}(\Sigma), S \cup \{ g \}} \nu^{\Fm}_{\Sigma, S, g}(\hybind{s}{\varrho})
\end{align*}

\smallskip\noindent\textit{Case $\hyat[F]{s}{\varrho}$}:
\begin{align*} &
  \nu^{\Mod}_{\Sigma}(M'), v, \gamma \models[\EDHML]_{\Sigma, S} \hyat[F]{s}{\varrho}
\just{def.\ $\models[\EDHML]$}
  \nu^{\Mod}_{\Sigma}(M'), v, \gamma' \models[\EDHML]_{\Sigma, S} \varrho
\\[-.0ex]&\qquad\text{for all $\gamma' \in \Conf^F(\nu^{\Mod}_{\Sigma}(M'))$ with $\ctrlSt(\nu^{\Mod}_{\Sigma}(M'))(\gamma') = v(s)$}
\just{I.\,H.}
  M', \beta'_{M', g'}(v, \gamma') \models[\CASL]_{\nu^{\Sign}(\Sigma), S \cup \{ g' \}} \nu^{\Fm}_{\Sigma, S, g'}(\varrho)
\\[-.0ex]&\qquad\text{for all $\gamma' \in \Conf^F(\nu^{\Mod}_{\Sigma}(M'))$ with $\ctrlSt(\nu^{\Mod}_{\Sigma}(M'))(\gamma') = v(s)$}
\just{def.\ $\models[\CASL]$}
  M', \beta'_{M', g}(v, \gamma) \models[\CASL]_{\nu^{\Sign}(\Sigma), S \cup \{ g \}}{}\\&\qquad\qquad \forall g' \cln \mathrm{Conf} \,.\, (\mathrm{c}(g') = s \land \mathrm{reachable}(F, g')) \Rightarrow \nu^{\Fm}_{\Sigma, S, g'}(\varrho)
\just{def.\ $\nu^{\Fm}$}
  M', \beta'_{M', g}(v, \gamma) \models[\CASL]_{\nu^{\Sign}(\Sigma), S \cup \{ g \}} \nu^{\Fm}_{\Sigma, S, g}(\hyat[F]{s}{\varrho})
\end{align*}

\smallskip\noindent\textit{Case $\mlbox[F]{\varrho}$}:
\begin{align*} &
  \nu^{\Mod}_{\Sigma}(M'), v, \gamma \models[\EDHML]_{\Sigma, S} \mlbox[F]{\varrho}
\just{def.\ $\models[\EDHML]$}
   \nu^{\Mod}_{\Sigma}(M'), v, \gamma' \models[\EDHML]_{\Sigma, S} \varrho
\quad\text{for all $\gamma' \in \Conf^F(\nu^{\Mod}_{\Sigma}(M'), \gamma)$}
\just{I.\,H.}
  M', \beta'_{M', g'}(v, \gamma') \models[\CASL]_{\nu^{\Sign}(\Sigma), S \cup \{ g' \}} \nu^{\Fm}_{\Sigma, S, g'}(\varrho)
\quad\text{for all $\gamma' \in \Conf^F(\nu^{\Mod}_{\Sigma}(M'), \gamma)$}
\just{def.\ $\models[\CASL]$}
  M', \beta'_{M', g}(v, \gamma) \models[\CASL]_{\nu^{\Sign}(\Sigma), S \cup \{ g \}} \forall g' \cln \mathrm{Conf} \,.\, \mathrm{reachable}(F, g, g') \Rightarrow \nu^{\Fm}_{\Sigma, S, g'}(\varrho)
\just{def.\ $\nu^{\Fm}$}
  M', \beta'_{M', g}(v, \gamma) \models[\CASL]_{\nu^{\Sign}(\Sigma), S \cup \{ g \}} \nu^{\Fm}_{\Sigma, S, g}(\mlbox[F]{\varrho})
\end{align*}

\smallskip\noindent\textit{Case $\dldia{e(X)}{\psi}{\varrho}$}:
\begin{align*} &
  \nu^{\Mod}_{\Sigma}(M'), v, \gamma \models[\EDHML]_{\Sigma, S} \dldia{e(X)}{\psi}{\varrho}
\just{def.\ $\models[\EDHML]$}
  \nu^{\Mod}_{\Sigma}(M'), v, \gamma \models[\EDHML]_{\Sigma, S} \varrho
\\[-.5ex]&\qquad\stacked{%
  \text{for some $\beta : X \to \DATA$, $\gamma' \in \Conf(\nu^{\Mod}_{\Sigma}(M'))$ with}\\[.5ex]
  \text{$(\gamma, \gamma') \in \Rel(\nu^{\Mod}_{\Sigma}(M'))_{e(\beta)}$ and}\\[.5ex]
  \text{$(\dataSt(\nu^{\Mod}_{\Sigma}(M'))(\gamma), \dataSt(\nu^{\Mod}_{\Sigma}(M'))(\gamma')), \beta \models[2\DATA]_{\Attr(\Sigma), X} \psi$}
}
\just{I.\,H., def.\ $\models[2\DATA]$, def.\ $\models[\CASL]$}
  M', \beta'_{M', g}(v, \gamma) \models[\CASL]_{\nu^{\Sign}(\Sigma), S \cup \{ g \}} \nu^{\Fm}_{\Sigma, S, g}(\varrho)
\\[-.5ex]&\qquad\stacked{
  \text{for some $\beta : X \to \DATA$, $\gamma' \in \Conf(\nu^{\Mod}_{\Sigma}(M'))$ with}\\[.5ex]
  \text{$(\gamma, \gamma') \in \Rel(\nu^{\Mod}_{\Sigma}(M'))_{e(\beta)}$ and}\\[.5ex]
  M', \stacked{
     \iota^{-1}_{M', \data} \compfun ((\dataSt(\nu^{\Mod}_{\Sigma}(M'))(\gamma) + \dataSt(\nu^{\Mod}_{\Sigma}(M'))(\gamma')) \cup \beta) \cup{}\\[.5ex]
     \{ g \mapsto \iota_{M', \mathrm{Conf}}(\gamma), g' \mapsto \iota_{M', \mathrm{Conf}}(\gamma') \} \models[\CASL]_{\nu^{\Sign}(\Sigma), X \cup \{ g, g' \}}{}\\[.5ex]
     \qquad\Frm[\CASL]{\nu^{\Sign}(\Sigma)}{\Attr(\Sigma)(g) \cup \Attr(\Sigma)(g') \cup \id{X}}(\psi)
  }
}
\just{def.\ $\models[\CASL]$}
  M', \beta'_{M', g}(v, \gamma) \models[\CASL]_{\nu^{\Sign}(\Sigma), S \cup \{ g \}}{}
\\[-.5ex]&\qquad
\stacked{\exists X \cln \data \,.\, \exists g' \cln \mathrm{Conf} \,.\, \mathrm{trans}(g, e(X), g') \land{}\\
    \quad\Frm[\CASL]{\nu^{\Sign}(\Sigma)}{\Attr(\Sigma)(g) \cup \Attr(\Sigma)(g') \cup \id{X}}(\psi) \land \nu^{\Fm}_{\Sigma, S, g}(\varrho)}
\just{def.\ $\nu^{\Fm}$}
  M', \beta'_{M', g}(v, \gamma) \models[\CASL]_{\nu^{\Sign}(\Sigma), S \cup \{ g \}} \nu^{\Fm}_{\Sigma, S, g}(\dldia{e(X)}{\psi}{\varrho})
\end{align*}

\smallskip\noindent\textit{Case $\dlwp{e(X)}{\phi}{\psi}{\varrho}$}:
\begin{align*} &
  \nu^{\Mod}_{\Sigma}(M'), v, \gamma \models[\EDHML]_{\Sigma, S} \dlwp{e(X)}{\phi}{\psi}{\varrho}
\just{def.\ $\models[\EDHML]$}
  \nu^{\Mod}_{\Sigma}(M'), v, \gamma \models[\EDHML]_{\Sigma, S} \varrho
\\[-.5ex]&\qquad\stacked{%
  \text{for all $\beta : X \to \DATA$ such that $\dataSt(\nu^{\Mod}_{\Sigma}(M'))(\gamma), \beta \models[\DATA]_{\Attr(\Sigma), X} \phi$}\\[.5ex]
  \text{and some $\gamma' \in \Conf(\nu^{\Mod}_{\Sigma}(M'))$ with $(\gamma, \gamma') \in \Rel(\nu^{\Mod}_{\Sigma}(M'))_{e(\beta)}$}\\[.5ex]
  \text{and $(\dataSt(\nu^{\Mod}_{\Sigma}(M'))(\gamma), \dataSt(\nu^{\Mod}_{\Sigma}(M'))(\gamma')), \beta \models[2\DATA]_{\Attr(\Sigma), X} \psi$}
}
\just{I.\,H., def.\ $\models[\DATA]$, def.\ $\models[2\DATA]$, def.\ $\models[\CASL]$}
  M', \beta'_{M', g}(v, \gamma) \models[\CASL]_{\nu^{\Sign}(\Sigma), S \cup \{ g \}} \nu^{\Fm}_{\Sigma, S, g}(\varrho)
\\[-.5ex]&\qquad\stacked{%
  \text{for all $\beta : X \to \DATA$ such that}\\[.5ex]
  \qquad M', \stacked{\iota^{-1}_{M', \data} \compfun (\dataSt(\nu^{\Mod}_{\Sigma}(M'))(\gamma) \cup \beta) \cup{}\\[.5ex]
                      \{ g \mapsto \iota_{M', \mathrm{Conf}}(\gamma) \} \models[\CASL]_{\nu^{\Sign}(\Sigma), X \cup \{ g \}} \Frm[\CASL]{\nu^{\Sign}(\Sigma)}{\Attr(\Sigma)(g) \cup \id{X}}(\phi)}\\[4.5ex]
  \text{and some $\gamma' \in \Conf(\nu^{\Mod}_{\Sigma}(M'))$ with $(\gamma, \gamma') \in \Rel(\nu^{\Mod}_{\Sigma}(M'))_{e(\beta)}$}\\[.5ex]
  \text{and }M', \stacked{
     \iota^{-1}_{M', \data} \compfun ((\dataSt(\nu^{\Mod}_{\Sigma}(M'))(\gamma) + \dataSt(\nu^{\Mod}_{\Sigma}(M'))(\gamma')) \cup \beta) \cup{}\\[.5ex]
     \{ g \mapsto \iota_{M', \mathrm{Conf}}(\gamma), g' \mapsto \iota_{M', \mathrm{Conf}}(\gamma') \} \models[\CASL]_{\nu^{\Sign}(\Sigma), X \cup \{ g, g' \}}{}\\[.5ex]
     \qquad\Frm[\CASL]{\nu^{\Sign}(\Sigma)}{\Attr(\Sigma)(g) \cup \Attr(\Sigma)(g') \cup \id{X}}(\psi)
  }
}
\just{def.\ $\models[\CASL]$}
  M', \beta'_{M', g}(v, \gamma) \models[\CASL]_{\nu^{\Sign}(\Sigma), S \cup \{ g \}}{}
\\[-.5ex]&\qquad
\stacked{\forall X \cln \data \,.\,\Frm[\CASL]{\nu^{\Sign}(\Sigma)}{\Attr(\Sigma)(g) \cup \id{X}}(\phi) \Rightarrow{}\\
    \quad\exists g' \cln \mathrm{Conf} \,.\, \mathrm{trans}(g, e(X), g') \land{}\\
   \qquad\Frm[\CASL]{\nu^{\Sign}(\Sigma)}{\Attr(\Sigma)(g) \cup \Attr(\Sigma)(g') \cup \id{X}}(\psi) \land \nu^{\Fm}_{\Sigma, S, g}(\varrho)}
\just{def.\ $\nu^{\Fm}$}
  M', \beta'_{M', g}(v, \gamma) \models[\CASL]_{\nu^{\Sign}(\Sigma), S \cup \{ g \}} \nu^{\Fm}_{\Sigma, S, g}(\dlwp{e(X)}{\phi}{\psi}{\varrho})
\end{align*}
\end{proof}
\end{techreport}

\begin{theorem}\label{thm:co-morph}
$(\nu^{\Sig}, \nu^{\Mod}, \nu^{\Sen})$ is a theoroidal comorphism from $\EDHML$ to $\CASL$.
\end{theorem}
\begin{techreport}
\begin{proof}
Let $\Sigma \in
\Sig[\EDHML]$, $M' \in |\Mod[\CASL](\nu^{\Sig}(\Sigma))|$, and $\rho \in
\Sen[\EDHML](\Sigma)$.  The satisfaction condition follows from
\begin{align*}&
  \nu^{\Mod}_{\Sigma}(M') \models[\EDHML]_{\Sigma} \rho
\just{def.\ $\models[\EDHML]$}
  \nu^{\Mod}_{\Sigma}(M'), \emptyset, \gamma_0 \models[\EDHML]_{\Sigma, \emptyset} \rho
\quad\text{for all $\gamma_0 \in \Conf_0(\nu^{\Mod}_{\Sigma}(M'))$}
\just{\cref{lem:co-morph}}
  M', \beta'_{M', g}(\emptyset, \gamma_0) \models[\CASL]_{\nu^{\Sign}(\Sigma), \{ g \}} \nu^{\Fm}_{\Sigma, \emptyset, g}(\rho)
\quad\text{for all $\gamma_0 \in \Conf_0(\nu^{\Mod}_{\Sigma}(M'))$}
\just{def.\ $\models[\CASL]$}
  M' \models[\CASL]_{\nu^{\Sign}(\Sigma)} \forall g \cln \mathrm{Conf} \,.\, \mathrm{init}(g) \Rightarrow \nu^{\Fm}_{\Sigma, \emptyset, g}(\rho)
\just{def.\ $\nu^{\Sen}$}
  M' \models[\CASL]_{\nu^{\Sign}(\Sigma)} \nu^{\Sen}_{\Sigma}(\rho)
\end{align*}
\end{proof}
\end{techreport}

\section{Proving Properties of UML State Machines with \HeTS and \SPASS}\label{sec:proving}

We implemented the translation of simple UML state machines into \CASL
specifications within the heterogeneous toolset
\HeTS~\cite{mossakowski-maeder-luettich:tacas:2007}.  Based on this translation
we explain how to prove properties symbolically in the automated theorem prover
\SPASS~\cite{weidenbach-et-al:cade:2009} for our running example of a counter.

\subsection{Implementation in \HeTS}

\begin{lstlisting}[language=UMLState, float=t, caption={Representation of the simple UML state machine $\mathit{Counter}$ in \UMLState}, label={lst:uml-counter-in-UMLState}]
logic $\normalfont\textsc{UMLState}$
spec $\mathit{Counter}$ =
  var $\attr{cnt}$;
  event $\evt{inc}(x)$;
  event $\evt{reset}$;
  states $\state{s1}$, $\state{s2}$;
  init $\state{s1}$ : $\attr{cnt} = 0$;
  trans $\state{s1}$ --> $\state{s1}$ : $\evt{inc}(x)$ [$\attr{cnt} + x < 4$] / { $\attr{cnt}$ := $\attr{cnt} + x$ };
  trans $\state{s1}$ --> $\state{s2}$ : $\evt{inc}(x)$ [$\attr{cnt} + x = 4$] / { $\attr{cnt}$ := $\attr{cnt} + x$ };
  trans $\state{s2}$ --> $\state{s1}$ : $\evt{reset}$ [$\attr{cnt} = 4$] / { $\attr{cnt}$ := $0$ };
end   
\end{lstlisting}
For a \HeTS chain from simple UML state machine to \CASL and \SPASS, we first
defined the input language \UMLState and extended \HeTS with a parser for this
language.  The syntax of \UMLState closely follows the ideas of
PlantUML~\cite{plantuml}, such that, in particular, its textual specifications
can potentially be rendered graphically as UML state machines.
\Cref{lst:uml-counter-in-UMLState} gives a representation of our running example
$\mathit{Counter}$, cf.\ \cref{sec:running_example}, in \UMLState.  Note that
\UMLState uses more conventional UML syntax for effects on transitions, e.g.,
``\texttt{$\attr{cnt}$ := $\attr{cnt} + x$}''.  Next, we extended \HeTS with a
syntax representation of our logic $\EDHML$, cf.~\cref{sec:edhml}, and
implemented \cref{alg:sen-op-spec} in \HeTS to automatically translate \UMLState
specifications into $\EDHML$ specifications, where we arrive at the
institutional level.  In this step, effects on transitions are turned into
logical formulæ, like ``$\attr{cnt}' = \attr{cnt} + x$''.  Finally, we extended
\HeTS with an implementation of the comorphism from $\EDHML$ into \CASL, cf.\
\cref{sec:comorphism}.  The implementation has been bundled in a fork of \HeTS
(\url{https://github.com/spechub/hets}) and provides a translation chain from
\UMLState via \CASL to the input languages of various proof tools, such as the
automated theorem prover \SPASS.



\subsection{Proving in \SPASS}

\Cref{fig:proof_conditions} shows the \CASL specification representing
the state machine from \cref{fig:uml-counter}, extended by a proof
obligation {\small{}\KW{\%}(Safe)\KW{\%}} and proof infrastructure for
it.  We want to prove the safety property that $\attr{cnt}$ never
exceeds $4$ using the automated theorem prover \SPASS.

The \CASL specification \SIdIndex{Counter} imports a specification
\SId{Trans} which instantiates the generic frame translating $\EDHML$
into $\CASL$, cf.\ \cref{lst:edhml-trans}.  However, the first-order
theorem prover \SPASS does not support \CASL's structured free that we
use for expressing reachability.  For invariance properties this
deficiency can be circumvented by loosely specifying \Id{reachable}
(i.e., omitting the keyword \texttt{free}), introducing a predicate
\Id{invar}, and adding a first-order induction axiom. This means that
we have to establish the safety property for a larger model class than
we would have with freeness. When carrying out symbolic reasoning for
invariant referring to a single configuration, the presented induction
axiom suffices. Other properties would require more involved induction
axioms, e.g., referring to several configurations.
\begin{figure}[!t]
\begin{hetcasl}
\SPEC \=\SIdIndex{Counter} \Ax{=} \SId{Trans}\\
\THEN 
\=\PRED \Id{invar}(\=\Id{g} \Ax{\cln} \Id{Conf})
\Ax{\Leftrightarrow} \=(\=\Id{c}(\Id{g}) \Ax{=} \Id{s1} \Ax{\wedge} \=\Id{cnt}(\Id{g}) \Ax{\leq} \Ax{4}) \Ax{\vee} (\=\Id{c}(\Id{g}) \Ax{=} \Id{s2} \Ax{\wedge} \=\Id{cnt}(\Id{g}) \Ax{\leq} \Ax{4})\\
\> {\small{}\KW{\%\%} induction scheme for ``reachable'' predicate\Ax{,} instantiated for ``invar'':}\\
\> \Ax{\axdot} \=(\=(\=\Ax{\forall} \Id{g} \Ax{\cln} \Id{Conf} \Ax{\axdot} \=\Id{init}(\Id{g}) \Ax{\Rightarrow} \Id{invar}(\Id{g})) \\
\>\>\> \Ax{\wedge} \=\Ax{\forall} \Id{g}, \Id{g'} \Ax{\cln} \Id{Conf}; \Id{e} \Ax{\cln} \Id{Evt} \\
\>\>\>\> \Ax{\axdot} \=(\=\Id{reachable}(\Id{g}) \Ax{\Rightarrow} \Id{invar}(\Id{g})) \Ax{\wedge} \Id{reachable}(\Id{g}) \Ax{\wedge} \Id{trans}(\=\Id{g}, \Id{e}, \Id{g'})  \Ax{\Rightarrow} \Id{invar}(\Id{g'})) \\
\>\> \Ax{\Rightarrow} \=\Ax{\forall} \Id{g} \Ax{\cln} \Id{Conf} \Ax{\axdot} \=\Id{reachable}(\Id{g}) \Ax{\Rightarrow} \Id{invar}(\Id{g}) \\
\THEN \ldots\ \Id{machine} \Id{axioms} \ldots\\
\THEN \={\small{}\KW{\%}\KW{implies}}\\
\> {\small{}\KW{\%\%} the safety assertion for our counter:}\\
\> \Ax{\forall} \=\Id{g} \Ax{\cln} \Id{Conf} \=\Ax{\axdot} \=\Id{reachable}(\Id{g}) \Ax{\Rightarrow} \=\Id{cnt}(\Id{g}) \Ax{\leq} \Ax{4} \`{\small{}\KW{\%}(Safe)\KW{\%}}\\
\> {\small{}\KW{\%\%} steering \SPASS with case distinction lemmas\Ax{,} could be generated algorithmically:}\\
\> \Ax{\forall} \Id{g}, \Id{g'} \Ax{\cln} \Id{Conf}; \Id{e} \Ax{\cln} \Id{Evt}; \=\Id{k} \Ax{\cln} \Id{Nat} \\
\> \Ax{\axdot} \=\Id{init}(\Id{g}) \Ax{\Rightarrow} \Id{invar}(\Id{g}) \`{\small{}\KW{\%}(InvarInit)\KW{\%}}\\
\> \Ax{\axdot} \=(\=\Id{reachable}(\Id{g}) \Ax{\Rightarrow} \Id{invar}(\Id{g})) \Ax{\wedge} \Id{reachable}(\Id{g}) \Ax{\wedge} \Id{trans}(\=\Id{g}, \Id{e}, \Id{g'}) \Ax{\wedge} \=\Id{e} \Ax{=} \Id{reset} \\
\>\> \Ax{\Rightarrow} \Id{invar}(\Id{g'}) \`{\small{}\KW{\%}(InvarReset)\KW{\%}}\\
\> \Ax{\axdot} \=(\=\Id{reachable}(\Id{g}) \Ax{\Rightarrow} \Id{invar}(\Id{g})) \Ax{\wedge} \Id{reachable}(\Id{g}) \Ax{\wedge} \Id{trans}(\=\Id{g}, \Id{e}, \Id{g'}) \Ax{\wedge} \=\Id{e} \Ax{=} \Id{inc}(\Id{k})  \\
\>\> \Ax{\Rightarrow} \Id{invar}(\Id{g'}) \`{\small{}\KW{\%}(InvarInc)\KW{\%}}\\
\> \Ax{\axdot} \=(\=\Id{reachable}(\Id{g}) \Ax{\Rightarrow} \Id{invar}(\Id{g})) \Ax{\wedge} \Id{reachable}(\Id{g}) \Ax{\wedge} \Id{trans}(\=\Id{g}, \Id{e}, \Id{g'}) \\
\>\> \Ax{\Rightarrow} \Id{invar}(\Id{g'}) \`{\small{}\KW{\%}(InvarStep)\KW{\%}}\\
\> \Ax{\axdot} \=\Id{invar}(\Id{g}) \Ax{\Rightarrow} \=\Id{cnt}(\Id{g}) \Ax{\leq} \Ax{4} \`{\small{}\KW{\%}(InvarImpliesSafe)\KW{\%}}\\
\end{hetcasl}
\vskip-20pt
\caption{\CASL specification of our running example}\label{fig:proof_conditions}
\end{figure}

Then the specification provides the machine axioms as stated
(partially) in \cref{ex:machine-axioms}.  The axioms following the
\textit{\%implies} directive are treated as proof obligations.  We
first state the safety property that we wish to establish: in all
reachable configurations, the counter value is less or equal $4$ --
\textit{\%}(Safe)\textit{\%}.  The remainder steers the proving
process in \SPASS by providing suitable case distinctions. For
invariants referring to a single configuration, these could also be
generated automatically based on the transition structure of the state
machine.

As proof of concept, we automatically verified this safety property in
\SPASS. In this experiment, we performed some optimising,
semantics-preserving logical transformations on the result of applying
the comorphism, to make the specification more digestible to the
theorem prover. These transformations include the removal of double
negations, splitting a conjunction into separate axioms, and turning
existentially quantified control states into constants by
Skolemisation.

\section{Conclusions and Future Work}\label{sec:conclusions}

We have described a new, institution-based logical framework $\EDHML$
that captures simple UML state machines. This is in contrast to
previous approaches that modelled UML parts directly as an institution
and ran into difficulties in establishing the satisfaction
condition~\cite{rosenberger:master:2017}. By (1) defining an
institution-based translation from $\EDHML$ into the \CASL institution
and (2) implementing and thus automatising our translation within
\HeTS, we made it possible to analyse UML state machines with the
broad range of provers accessible via \HeTS.

The resulting tool chain allows us to apply an automatic prover (as
demonstrated here using the theorem prover \SPASS), or several
automatic provers, where they work and switch to interactive tools
like Isabelle where necessary (not needed in the analysis of our
example $\mathit{Counter}$). Not only does this switch require no
manual reformulation into the interactive tool's input language,
rather, it can be done even within one development: We could possibly
show some lemmas via automatic first-order provers, some lemmas via
domain-specific tools, then use those to prove a difficult lemma in an
interactive prover, then apply all those lemmas to automatically prove
the final theorem. \HeTS allows us to use the best language and the
best tool for each job, and takes care of linking the results together
under the hood.


It is future work to extend $\EDHML$ to cover more elements of UML state
machines, such as hierarchical states and communication networks. The main
challenge here will be to enrich $\EDHML$ in such a way that it offers suitable
logical representations for the additional structural elements (hierarchical
states or communication networks) rather than to flatten these: We anticipate
symbolic reasoning on UML state machines to be ``easier'' if their structural
elements are still ``visible'' in their \CASL representations.

In the long term, we work towards heterogeneous verification
of different UML diagrams. One possible setting would be to utilise
interactions as a specification mechanism, where communicating state
machines model implementations.
         

\bibliographystyle{splncs04}
\bibliography{bibliography}

\begin{techreport}
\clearpage
\begin{appendix}

\section{Full \CASL specifications}
\begin{hetcasl}
\KW{library} \Id{counter\Ax{\_}generated}\\
\\
{\small{}\KW{\%}\KW{display}(}\={\small{}\Ax{\_\_}}{\small{}\Ax{<=}}{\small{}\Ax{\_\_}}{\small{} }{\small{}\KW{\%}LATEX }\={\small{}\Ax{\_\_}}\Ax{\leq}{\small{}\Ax{\_\_}}{\small{})\KW{\%}}\\
{\small{}\KW{\%}\KW{display}(}\={\small{}\Ax{\_\_}}{\small{}\Ax{>=}}{\small{}\Ax{\_\_}}{\small{} }{\small{}\KW{\%}LATEX }\={\small{}\Ax{\_\_}}\Ax{\geq}{\small{}\Ax{\_\_}}{\small{})\KW{\%}}\\
{\small{}\KW{\%}\KW{prec}(}\={\small{}\{}\={\small{}\Ax{\_\_}}{\small{}\Ax{+}}{\small{}\Ax{\_\_}}{\small{}\}}{\small{} }{\small{}\Ax{<} }{\small{}\{}\={\small{}\Ax{\_\_}}{\small{}\Ax{*}}{\small{}\Ax{\_\_}}{\small{},}{\small{} }{\small{}\Ax{\_\_}}{\small{}\Id{div}}{\small{}\Ax{\_\_}}{\small{},}{\small{} }\={\small{}\Ax{\_\_}}{\small{}\Id{mod}}{\small{}\Ax{\_\_}}{\small{}\}}{\small{})\KW{\%}}\\
{\small{}\KW{\%}\KW{prec}(}\={\small{}\{}\={\small{}\Ax{\_\_}}{\small{}\Ax{*}}{\small{}\Ax{\_\_}}{\small{},}{\small{} }{\small{}\Ax{\_\_}}{\small{}\Id{div}}{\small{}\Ax{\_\_}}{\small{},}{\small{} }\={\small{}\Ax{\_\_}}{\small{}\Id{mod}}{\small{}\Ax{\_\_}}{\small{}\}}{\small{} }{\small{}\Ax{<} }{\small{}\{}\={\small{}\Ax{\_\_}}{\small{}\Ax{\Ax{\hat{\ }}}}{\small{}\Ax{\_\_}}{\small{}\}}{\small{})\KW{\%}}\\
{\small{}\KW{\%}\KW{right\_assoc}(}\={\small{}\Ax{\_\_}}{\small{}\Ax{+}}{\small{}\Ax{\_\_}}{\small{},}{\small{} }{\small{}\Ax{\_\_}}{\small{}\Ax{*}}{\small{}\Ax{\_\_}}{\small{},}{\small{} }\={\small{}\Ax{\_\_}}{\small{}\Ax{-}}{\small{}\Ax{\_\_}}{\small{})\KW{\%}}\\
{\small{}\KW{\%}\KW{number} }\={\small{}\Ax{\_\_}}{\small{}\Ax{@@}}{\small{}\Ax{\_\_}}\\
\\
\SPEC \=\SIdIndex{Nat} \Ax{=}\\
\> \TYPE \=\Id{Nat} \Ax{:}\Ax{:}\=\Ax{=} \Ax{0} \AltBar{} \Id{suc}(\Id{Nat})\\
\> \OPS \=\Ax{\_\_}\Ax{+}\Ax{\_\_} \Ax{:} \=\Id{Nat} \Ax{\times} \Id{Nat} \Ax{\rightarrow} \Id{Nat};\\
\>\> \Ax{\_\_}\Ax{*}\Ax{\_\_} \Ax{:} \=\Id{Nat} \Ax{\times} \Id{Nat} \Ax{\rightarrow} \Id{Nat}\\
\> \PREDS \=\Ax{\_\_}\Ax{\leq}\Ax{\_\_}, \Ax{\_\_}\Ax{<}\Ax{\_\_}, \Ax{\_\_}\Ax{>}\Ax{\_\_}, \Ax{\_\_}\Ax{\geq}\Ax{\_\_}, \Ax{\_\_}\Ax{==}\Ax{\_\_} \\
\>\> \Ax{:} \=\Id{Nat} \Ax{\times} \Id{Nat}\\
\> \Ax{\forall} \=\Id{n}, \Id{m} \Ax{:} \Id{Nat} \\
\> \Ax{\bullet} \=\Ax{0} \Ax{+} \Id{n} \Ax{=} \Id{n}\\
\> \Ax{\bullet} \=\Id{suc}(\Id{n}) \Ax{+} \Id{m} \Ax{=} \Id{suc}(\=\Id{n} \Ax{+} \Id{m})\\
\> \Ax{\bullet} \=\Ax{0} \Ax{*} \Id{n} \Ax{=} \Ax{0}\\
\> \Ax{\bullet} \=\Id{suc}(\Id{n}) \Ax{*} \Id{m} \Ax{=} \=\Id{m} \Ax{+} \=\Id{n} \Ax{*} \Id{m}\\
\> \Ax{\bullet} \=\Ax{0} \Ax{\leq} \Id{n}\\
\> \Ax{\bullet} \Ax{\neg} \=\Id{suc}(\Id{n}) \Ax{\leq} \Ax{0}\\
\> \Ax{\bullet} \=\Id{suc}(\Id{m}) \Ax{\leq} \Id{suc}(\Id{n}) \Ax{\Leftrightarrow} \=\Id{m} \Ax{\leq} \Id{n}\\
\> \Ax{\bullet} \=\Id{m} \Ax{\geq} \Id{n} \Ax{\Leftrightarrow} \=\Id{n} \Ax{\leq} \Id{m}\\
\> \Ax{\bullet} \=\Id{m} \Ax{<} \Id{n} \Ax{\Leftrightarrow} \=\Id{m} \Ax{\leq} \Id{n} \Ax{\wedge} \Ax{\neg} \=\Id{m} \Ax{=} \Id{n}\\
\> \Ax{\bullet} \=\Id{m} \Ax{>} \Id{n} \Ax{\Leftrightarrow} \=\Id{n} \Ax{<} \Id{m}\\
\> \Ax{\bullet} \=\Id{m} \Ax{==} \Id{n} \Ax{\Leftrightarrow} \=\Id{m} \Ax{=} \Id{n}\\
\THEN \={\small{}\KW{\%}\KW{def}}\\
\> {\small{}\KW{\%\%} Operations to represent natural numbers with digits\Ax{:}}\\
\> \OPS \=\Ax{1} \Ax{:} \Id{Nat} \Ax{=} \Id{suc}(\Ax{0});\\
\>\> \Ax{2} \Ax{:} \Id{Nat} \Ax{=} \Id{suc}(\Ax{1});\\
\>\> \Ax{3} \Ax{:} \Id{Nat} \Ax{=} \Id{suc}(\Ax{2});\\
\>\> \Ax{4} \Ax{:} \Id{Nat} \Ax{=} \Id{suc}(\Ax{3});\\
\>\> \Ax{5} \Ax{:} \Id{Nat} \Ax{=} \Id{suc}(\Ax{4});\\
\>\> \Ax{6} \Ax{:} \Id{Nat} \Ax{=} \Id{suc}(\Ax{5});\\
\>\> \Ax{7} \Ax{:} \Id{Nat} \Ax{=} \Id{suc}(\Ax{6});\\
\>\> \Ax{8} \Ax{:} \Id{Nat} \Ax{=} \Id{suc}(\Ax{7});\\
\>\> \Ax{9} \Ax{:} \Id{Nat} \Ax{=} \Id{suc}(\Ax{8});\\
\>\> \Ax{\_\_}\Ax{@@}\Ax{\_\_}\=(\=\Id{m} \Ax{:} \Id{Nat}; \=\Id{n} \Ax{:} \Id{Nat}) \Ax{:} \Id{Nat} \Ax{=} \=\Id{m} \Ax{*} \Id{suc}(\Ax{9}) \Ax{+} \Id{n} \\
\>\> \`{\small{}\KW{\%}(decimal\Ax{\_}def)\KW{\%}}\\
\THEN \=\OPS \=\Ax{\_\_}\Ax{+}\Ax{\_\_} \Ax{:} \=\Id{Nat} \Ax{\times} \Id{Nat} \Ax{\rightarrow} \Id{Nat}, \Id{assoc}, \Id{comm}, \Id{unit} \Ax{0};\\
\>\> \Ax{\_\_}\Ax{*}\Ax{\_\_} \Ax{:} \=\Id{Nat} \Ax{\times} \Id{Nat} \Ax{\rightarrow} \Id{Nat}, \Id{assoc}, \Id{comm}, \Id{unit} \Ax{1}\\
\> \Ax{\forall} \=\Id{r}, \Id{s}, \Id{t} \Ax{:} \Id{Nat} \\
\> \Ax{\bullet} \=(\=\Id{r} \Ax{+} \Id{s}) \Ax{*} \Id{t} \Ax{=} \=\Id{r} \Ax{*} \Id{t} \Ax{+} \=\Id{s} \Ax{*} \Id{t}\\
\> \Ax{\bullet} \=\Id{t} \Ax{*} (\=\Id{r} \Ax{+} \Id{s}) \Ax{=} \=\Id{t} \Ax{*} \Id{r} \Ax{+} \=\Id{t} \Ax{*} \Id{s}\\
\KW{end}\\
\\
\SPEC \=\SIdIndex{GenerateSet}[\KW{sort} \Id{Elem}] \Ax{=} {\small{}\KW{\%}\KW{mono}}\\
\> \KW{generated} \KW{type} \=\Id{Set}[\Id{Elem}] \Ax{:}\Ax{:}\=\Ax{=} \=\Ax{\{}\Ax{\}} \AltBar{} \Ax{\_\_}\Ax{+}\Ax{\_\_}(\=\Id{Set}[\Id{Elem}]; \Id{Elem})\\
\> \PRED \=\Ax{\_\_}\Id{eps}\Ax{\_\_} \Ax{:} \=\Id{Elem} \Ax{\times} \=\Id{Set}[\Id{Elem}]\\
\> \Ax{\forall} \Id{x}, \Id{y} \Ax{:} \Id{Elem}; \=\Id{M}, \Id{N} \Ax{:} \=\Id{Set}[\Id{Elem}] \\
\> \Ax{\bullet} \Ax{\neg} \=\Id{x} \Id{eps} \=\Ax{\{}\Ax{\}} \`{\small{}\KW{\%}(elemOf\Ax{\_}empty\Ax{\_}Set)\KW{\%}}\\
\> \Ax{\bullet} \=\Id{x} \Id{eps} \=\Id{M} \Ax{+} \Id{y} \Ax{\Leftrightarrow} \=\Id{x} \Ax{=} \Id{y} \Ax{\vee} \=\Id{x} \Id{eps} \Id{M} \\
\> \`{\small{}\KW{\%}(elemOf\Ax{\_}NonEmpty\Ax{\_}Set)\KW{\%}}\\
\> \Ax{\bullet} \=\Id{M} \Ax{=} \Id{N} \Ax{\Leftrightarrow} \=\Ax{\forall} \Id{x} \Ax{:} \Id{Elem} \Ax{\bullet} \=\Id{x} \Id{eps} \Id{M} \Ax{\Leftrightarrow} \=\Id{x} \Id{eps} \Id{N} \\
\> \`{\small{}\KW{\%}(equality\Ax{\_}Set)\KW{\%}}\\
\KW{end}\\
\\
\SPEC \=\SIdIndex{Set}[\KW{sort} \Id{Elem}] \KW{given} \SId{Nat} \Ax{=} {\small{}\KW{\%}\KW{mono}}\\
\> \SId{GenerateSet}[\KW{sort} \Id{Elem}]\\
\THEN \={\small{}\KW{\%}\KW{def}}\\
\> \PREDS \=\Id{isNonEmpty} \Ax{:} \=\Id{Set}[\Id{Elem}];\\
\>\> \Ax{\_\_}\Id{isSubsetOf}\Ax{\_\_} \Ax{:} \=\Id{Set}[\Id{Elem}] \Ax{\times} \=\Id{Set}[\Id{Elem}]\\
\> \OPS \=\Ax{\{}\Ax{\_\_}\Ax{\}} \Ax{:} \=\Id{Elem} \Ax{\rightarrow} \=\Id{Set}[\Id{Elem}];\\
\>\> \Ax{\#}\Ax{\_\_} \Ax{:} \=\Id{Set}[\Id{Elem}] \Ax{\rightarrow} \Id{Nat};\\
\>\> \Ax{\_\_}\Ax{+}\Ax{\_\_} \Ax{:} \=\Id{Elem} \Ax{\times} \Id{Set}[\Id{Elem}] \Ax{\rightarrow} \=\Id{Set}[\Id{Elem}];\\
\>\> \Ax{\_\_}\Ax{-}\Ax{\_\_} \Ax{:} \=\Id{Set}[\Id{Elem}] \Ax{\times} \Id{Elem} \Ax{\rightarrow} \=\Id{Set}[\Id{Elem}];\\
\>\> \Ax{\_\_}\Id{intersection}\Ax{\_\_}, \Ax{\_\_}\Id{union}\Ax{\_\_}, \Ax{\_\_}\Ax{-}\Ax{\_\_}, \\
\>\> \Ax{\_\_}\Id{symDiff}\Ax{\_\_} \Ax{:} \=\Id{Set}[\Id{Elem}] \Ax{\times} \Id{Set}[\Id{Elem}] \Ax{\rightarrow} \=\Id{Set}[\Id{Elem}]\\
\KW{end}\\
\\
\SPEC \=\SIdIndex{EvtNameSet} \Ax{=}\\
\> \SId{Set}[\KW{sort} \Id{EvtName}] \KW{with} \=\Id{Set}[\Id{EvtName}] \Ax{\mapsto} \Id{EvtNameSet}\\
\KW{end}\\
\\
\SPEC \=\SIdIndex{Trans} \Ax{=}\\
\> \SId{Nat}\\
\THEN \SId{EvtNameSet}\\
\THEN \=\SORTS \=\Id{Conf}, \Id{Ctrl}, \Id{Evt}, \Id{EvtName}, \Id{EvtNameSet}, \Id{Nat}\\
\> \OP \=\Id{cnt} \Ax{:} \=\Id{Conf} \Ax{\rightarrow} \Id{Nat}\\
\> \OP \=\Id{conf} \Ax{:} \=\Id{Ctrl} \Ax{\times} \Id{Nat} \Ax{\rightarrow} \Id{Conf}\\
\> \OP \=\Id{ctrl} \Ax{:} \=\Id{Conf} \Ax{\rightarrow} \Id{Ctrl}\\
\> \OP \=\Id{evtName} \Ax{:} \=\Id{Evt} \Ax{\rightarrow} \Id{EvtName}\\
\> \OP \=\Id{evtName\Ax{\_}inc} \Ax{:} \Id{EvtName}\\
\> \OP \=\Id{evtName\Ax{\_}reset} \Ax{:} \Id{EvtName}\\
\> \OP \=\Id{evt\Ax{\_}inc} \Ax{:} \=\Id{Nat} \Ax{\rightarrow} \Id{Evt}\\
\> \OP \=\Id{evt\Ax{\_}reset} \Ax{:} \Id{Evt}\\
\> \OP \=\Id{s1} \Ax{:} \Id{Ctrl}\\
\> \OP \=\Id{s2} \Ax{:} \Id{Ctrl}\\
\> \PRED \=\Id{init} \Ax{:} \Id{Conf}\\
\> \PRED \=\Id{reachable2} \Ax{:} \=\Id{EvtNameSet} \Ax{\times} \Id{Conf}\\
\> \PRED \=\Id{trans} \Ax{:} \=\Id{Conf} \Ax{\times} \Id{Evt} \Ax{\times} \Id{Conf}\\
\> \Ax{\forall} \=\Id{g} \Ax{:} \Id{Conf} \\
\> \Ax{\bullet} \=\Id{init}(\Id{g}) \Ax{\Leftrightarrow} \=\Id{ctrl}(\Id{g}) \Ax{=} \Id{s1} \Ax{\wedge} \=\Id{cnt}(\Id{g}) \Ax{==} (\=\Id{op} \Ax{0} \Ax{:} \Id{Nat}) \`{\small{}\KW{\%}(init)\KW{\%}}\\
\> \\
\> {\small{}\KW{\%\%} free}\\
\> \KW{generated} \KW{types} \\
\> \Id{Evt} \Ax{:}\Ax{:}\=\Ax{=} \Id{evt\Ax{\_}inc}(\Id{Nat}) \AltBar{} \Id{evt\Ax{\_}reset};\\
\> \Id{EvtName} \Ax{:}\Ax{:}\=\Ax{=} \Id{evtName\Ax{\_}inc} \AltBar{} \Id{evtName\Ax{\_}reset};\\
\> \Id{Conf} \Ax{:}\Ax{:}\Ax{=} \=\Id{conf}(\=\Id{Ctrl}; \Id{Nat}) \`{\small{}\KW{\%}(free\Ax{\_}types)\KW{\%}}\\
\> \Ax{\forall} \=\Id{g'} \Ax{:} \Id{Conf} \\
\> \Ax{\bullet} \=\Id{ctrl}(\Id{g'}) \Ax{=} \Id{s1} \\
\>\> \Ax{\wedge} \Id{reachable2}(\=(\=\Id{evtName\Ax{\_}inc} \\
\>\>\>\> \Ax{+} \=(\=\Id{evtName\Ax{\_}reset} \Ax{+} (\=\Id{op} \Ax{\{}\Ax{\}} \Ax{:} \Id{EvtNameSet})) \Ax{:} \\
\>\>\>\>\> \Id{EvtNameSet}) \\
\>\>\> \Ax{:} \\
\>\>\> \Id{EvtNameSet}, \\
\>\>\> \Id{g'} \Ax{:} \Id{Conf}) \\
\>\> \Ax{\Rightarrow} \=\Ax{\forall} \Id{k} \Ax{:} \Id{Nat} \\
\>\>\> \Ax{\bullet} \=(\=\Id{cnt}(\Id{g'}) \Ax{+} \=\Id{k} \Ax{:} \Id{Nat}) \Ax{:} \Id{Nat} \Ax{<} (\=\Id{op} \Ax{4} \Ax{:} \Id{Nat}) \\
\>\>\>\> \Ax{\Rightarrow} \=\Ax{\exists} \Id{g''} \Ax{:} \Id{Conf} \\
\>\>\>\>\> \Ax{\bullet} \=(\=\Id{trans}(\=\Id{g'}, \Id{evt\Ax{\_}inc}(\Id{k}), \Id{g''}) \\
\>\>\>\>\>\>\> \Ax{\wedge} \=\Id{cnt}(\Id{g''}) \Ax{==} \=(\=\Id{cnt}(\Id{g'}) \Ax{+} \=\Id{k} \Ax{:} \Id{Nat}) \Ax{:} \Id{Nat}) \\
\>\>\>\>\>\> \Ax{\wedge} \=\Id{s1} \Ax{=} \Id{ctrl}(\Id{g''}) \\
\> \`{\small{}\KW{\%}(machine)\KW{\%}}\\
\> \Ax{\forall} \=\Id{g'} \Ax{:} \Id{Conf} \\
\> \Ax{\bullet} \=\Id{ctrl}(\Id{g'}) \Ax{=} \Id{s1} \\
\>\> \Ax{\wedge} \Id{reachable2}(\=(\=\Id{evtName\Ax{\_}inc} \\
\>\>\>\> \Ax{+} \=(\=\Id{evtName\Ax{\_}reset} \Ax{+} (\=\Id{op} \Ax{\{}\Ax{\}} \Ax{:} \Id{EvtNameSet})) \Ax{:} \\
\>\>\>\>\> \Id{EvtNameSet}) \\
\>\>\> \Ax{:} \\
\>\>\> \Id{EvtNameSet}, \\
\>\>\> \Id{g'} \Ax{:} \Id{Conf}) \\
\>\> \Ax{\Rightarrow} \=\Ax{\forall} \Id{k} \Ax{:} \Id{Nat} \\
\>\>\> \Ax{\bullet} \=(\=\Id{cnt}(\Id{g'}) \Ax{+} \=\Id{k} \Ax{:} \Id{Nat}) \Ax{:} \Id{Nat} \Ax{==} (\=\Id{op} \Ax{4} \Ax{:} \Id{Nat}) \\
\>\>\>\> \Ax{\Rightarrow} \=\Ax{\exists} \Id{g''} \Ax{:} \Id{Conf} \\
\>\>\>\>\> \Ax{\bullet} \=(\=\Id{trans}(\=\Id{g'}, \Id{evt\Ax{\_}inc}(\Id{k}), \Id{g''}) \\
\>\>\>\>\>\>\> \Ax{\wedge} \=\Id{cnt}(\Id{g''}) \Ax{==} \=(\=\Id{cnt}(\Id{g'}) \Ax{+} \=\Id{k} \Ax{:} \Id{Nat}) \Ax{:} \Id{Nat}) \\
\>\>\>\>\>\> \Ax{\wedge} \=\Id{s2} \Ax{=} \Id{ctrl}(\Id{g''}) \\
\> \`{\small{}\KW{\%}(machine\Ax{\_}1)\KW{\%}}\\
\> \Ax{\forall} \=\Id{g'} \Ax{:} \Id{Conf} \\
\> \Ax{\bullet} \=\Id{ctrl}(\Id{g'}) \Ax{=} \Id{s2} \\
\>\> \Ax{\wedge} \Id{reachable2}(\=(\=\Id{evtName\Ax{\_}inc} \\
\>\>\>\> \Ax{+} \=(\=\Id{evtName\Ax{\_}reset} \Ax{+} (\=\Id{op} \Ax{\{}\Ax{\}} \Ax{:} \Id{EvtNameSet})) \Ax{:} \\
\>\>\>\>\> \Id{EvtNameSet}) \\
\>\>\> \Ax{:} \\
\>\>\> \Id{EvtNameSet}, \\
\>\>\> \Id{g'} \Ax{:} \Id{Conf}) \\
\>\> \Ax{\Rightarrow} \=\Id{cnt}(\Id{g'}) \Ax{==} (\=\Id{op} \Ax{4} \Ax{:} \Id{Nat}) \\
\>\>\> \Ax{\Rightarrow} \=\Ax{\exists} \Id{g''} \Ax{:} \Id{Conf} \\
\>\>\>\> \Ax{\bullet} \=(\=\Id{trans}(\=\Id{g'}, \Id{evt\Ax{\_}reset}, \Id{g''}) \Ax{\wedge} \=\Id{cnt}(\Id{g''}) \Ax{==} (\=\Id{op} \Ax{0} \Ax{:} \Id{Nat})) \\
\>\>\>\>\> \Ax{\wedge} \=\Id{s1} \Ax{=} \Id{ctrl}(\Id{g''}) \\
\> \`{\small{}\KW{\%}(machine\Ax{\_}2)\KW{\%}}\\
\> \Ax{\forall} \=\Id{g'} \Ax{:} \Id{Conf} \\
\> \Ax{\bullet} \=\Id{ctrl}(\Id{g'}) \Ax{=} \Id{s1} \\
\>\> \Ax{\wedge} \Id{reachable2}(\=(\=\Id{evtName\Ax{\_}inc} \\
\>\>\>\> \Ax{+} \=(\=\Id{evtName\Ax{\_}reset} \Ax{+} (\=\Id{op} \Ax{\{}\Ax{\}} \Ax{:} \Id{EvtNameSet})) \Ax{:} \\
\>\>\>\>\> \Id{EvtNameSet}) \\
\>\>\> \Ax{:} \\
\>\>\> \Id{EvtNameSet}, \\
\>\>\> \Id{g'} \Ax{:} \Id{Conf}) \\
\>\> \Ax{\Rightarrow} \=(\=(\=(\=(\Ax{\neg} \=\Ax{\exists} \Id{k} \Ax{:} \Id{Nat} \\
\>\>\>\>\>\>\> \Ax{\bullet} \=\Ax{\exists} \Id{g''} \Ax{:} \Id{Conf} \\
\>\>\>\>\>\>\>\> \Ax{\bullet} \=(\=\Id{trans}(\=\Id{g'}, \Id{evt\Ax{\_}inc}(\Id{k}), \Id{g''}) \\
\>\>\>\>\>\>\>\>\>\> \Ax{\wedge} (\=\Id{true} \\
\>\>\>\>\>\>\>\>\>\>\> \Ax{\wedge} \Ax{\neg} (\=(\=(\Id{cnt}(\Id{g'}) \Ax{+} \Id{k} \Ax{:} \Id{Nat}) \Ax{:} \Id{Nat} \Ax{<} (\Id{op} \Ax{4} \Ax{:} \Id{Nat}) \\
\>\>\>\>\>\>\>\>\>\>\>\>\> \Ax{\wedge} \Id{cnt}(\Id{g''}) \Ax{==} (\Id{cnt}(\Id{g'}) \Ax{+} \Id{k} \Ax{:} \Id{Nat}) \Ax{:} \Id{Nat}) \\
\>\>\>\>\>\>\>\>\>\>\>\> \Ax{\vee} (\=(\Id{cnt}(\Id{g'}) \Ax{+} \Id{k} \Ax{:} \Id{Nat}) \Ax{:} \Id{Nat} \\
\>\>\>\>\>\>\>\>\>\>\>\>\> \Ax{==} (\Id{op} \Ax{4} \Ax{:} \Id{Nat}) \\
\>\>\>\>\>\>\>\>\>\>\>\>\> \Ax{\wedge} \Id{cnt}(\Id{g''}) \\
\>\>\>\>\>\>\>\>\>\>\>\>\> \Ax{==} (\Id{cnt}(\Id{g'}) \Ax{+} \Id{k} \Ax{:} \Id{Nat}) \Ax{:} \Id{Nat})))) \\
\>\>\>\>\>\>\>\>\> \Ax{\wedge} \Id{true}) \\
\>\>\>\>\>\> \Ax{\wedge} \Ax{\neg} \=\Ax{\exists} \Id{k} \Ax{:} \Id{Nat} \\
\>\>\>\>\>\>\> \Ax{\bullet} \=\Ax{\exists} \Id{g''} \Ax{:} \Id{Conf} \\
\>\>\>\>\>\>\>\> \Ax{\bullet} \=(\=\Id{trans}(\=\Id{g'}, \Id{evt\Ax{\_}inc}(\Id{k}), \Id{g''}) \\
\>\>\>\>\>\>\>\>\>\> \Ax{\wedge} (\=(\=(\=\Id{cnt}(\Id{g'}) \Ax{+} \Id{k} \Ax{:} \Id{Nat}) \Ax{:} \Id{Nat} \Ax{<} (\Id{op} \Ax{4} \Ax{:} \Id{Nat}) \\
\>\>\>\>\>\>\>\>\>\>\>\> \Ax{\wedge} \=\Id{cnt}(\Id{g''}) \Ax{==} (\Id{cnt}(\Id{g'}) \Ax{+} \Id{k} \Ax{:} \Id{Nat}) \Ax{:} \Id{Nat}) \\
\>\>\>\>\>\>\>\>\>\>\> \Ax{\wedge} \Ax{\neg} (\=(\=\Id{cnt}(\Id{g'}) \Ax{+} \Id{k} \Ax{:} \Id{Nat}) \Ax{:} \Id{Nat} \Ax{==} (\Id{op} \Ax{4} \Ax{:} \Id{Nat}) \\
\>\>\>\>\>\>\>\>\>\>\>\> \Ax{\wedge} \=\Id{cnt}(\Id{g''}) \Ax{==} (\Id{cnt}(\Id{g'}) \Ax{+} \Id{k} \Ax{:} \Id{Nat}) \Ax{:} \Id{Nat}))) \\
\>\>\>\>\>\>\>\>\> \Ax{\wedge} \Ax{\neg} \=\Id{s1} \Ax{=} \Id{ctrl}(\Id{g''})) \\
\>\>\>\>\> \Ax{\wedge} \Ax{\neg} \=\Ax{\exists} \Id{k} \Ax{:} \Id{Nat} \\
\>\>\>\>\>\> \Ax{\bullet} \=\Ax{\exists} \Id{g''} \Ax{:} \Id{Conf} \\
\>\>\>\>\>\>\> \Ax{\bullet} \=(\=\Id{trans}(\=\Id{g'}, \Id{evt\Ax{\_}inc}(\Id{k}), \Id{g''}) \\
\>\>\>\>\>\>\>\>\> \Ax{\wedge} (\=(\=(\=\Id{cnt}(\Id{g'}) \Ax{+} \=\Id{k} \Ax{:} \Id{Nat}) \Ax{:} \Id{Nat} \Ax{==} (\Id{op} \Ax{4} \Ax{:} \Id{Nat}) \\
\>\>\>\>\>\>\>\>\>\>\> \Ax{\wedge} \=\Id{cnt}(\Id{g''}) \Ax{==} \=(\Id{cnt}(\Id{g'}) \Ax{+} \Id{k} \Ax{:} \Id{Nat}) \Ax{:} \Id{Nat}) \\
\>\>\>\>\>\>\>\>\>\> \Ax{\wedge} \Ax{\neg} (\=(\=\Id{cnt}(\Id{g'}) \Ax{+} \=\Id{k} \Ax{:} \Id{Nat}) \Ax{:} \Id{Nat} \Ax{<} (\Id{op} \Ax{4} \Ax{:} \Id{Nat}) \\
\>\>\>\>\>\>\>\>\>\>\> \Ax{\wedge} \=\Id{cnt}(\Id{g''}) \Ax{==} \=(\Id{cnt}(\Id{g'}) \Ax{+} \Id{k} \Ax{:} \Id{Nat}) \Ax{:} \Id{Nat}))) \\
\>\>\>\>\>\>\>\> \Ax{\wedge} \Ax{\neg} \=\Id{s2} \Ax{=} \Id{ctrl}(\Id{g''})) \\
\>\>\>\> \Ax{\wedge} \Ax{\neg} \=\Ax{\exists} \Id{k} \Ax{:} \Id{Nat} \\
\>\>\>\>\> \Ax{\bullet} \=\Ax{\exists} \Id{g''} \Ax{:} \Id{Conf} \\
\>\>\>\>\>\> \Ax{\bullet} \=(\=\Id{trans}(\=\Id{g'}, \Id{evt\Ax{\_}inc}(\Id{k}), \Id{g''}) \\
\>\>\>\>\>\>\>\> \Ax{\wedge} (\=(\=(\=(\=\Id{cnt}(\Id{g'}) \Ax{+} \=\Id{k} \Ax{:} \Id{Nat}) \Ax{:} \Id{Nat} \Ax{<} (\Id{op} \Ax{4} \Ax{:} \Id{Nat}) \\
\>\>\>\>\>\>\>\>\>\>\> \Ax{\wedge} \=\Id{cnt}(\Id{g''}) \Ax{==} \=(\Id{cnt}(\Id{g'}) \Ax{+} \Id{k} \Ax{:} \Id{Nat}) \Ax{:} \Id{Nat}) \\
\>\>\>\>\>\>\>\>\>\> \Ax{\wedge} (\=(\=\Id{cnt}(\Id{g'}) \Ax{+} \=\Id{k} \Ax{:} \Id{Nat}) \Ax{:} \Id{Nat} \Ax{==} (\Id{op} \Ax{4} \Ax{:} \Id{Nat}) \\
\>\>\>\>\>\>\>\>\>\>\> \Ax{\wedge} \=\Id{cnt}(\Id{g''}) \Ax{==} \=(\Id{cnt}(\Id{g'}) \Ax{+} \Id{k} \Ax{:} \Id{Nat}) \Ax{:} \Id{Nat})) \\
\>\>\>\>\>\>\>\>\> \Ax{\wedge} \Ax{\neg} \Id{false})) \\
\>\>\>\>\>\>\> \Ax{\wedge} (\=\Ax{\neg} \=\Id{s1} \Ax{=} \Id{ctrl}(\Id{g''}) \Ax{\wedge} \Ax{\neg} \=\Id{s2} \Ax{=} \Id{ctrl}(\Id{g''}))) \\
\>\>\> \Ax{\wedge} \Ax{\neg} \=\Ax{\exists} \Id{g''} \Ax{:} \Id{Conf} \\
\>\>\>\> \Ax{\bullet} \=(\=\Id{trans}(\=\Id{g'}, \Id{evt\Ax{\_}reset}, \Id{g''}) \Ax{\wedge} (\=\Id{true} \Ax{\wedge} \Ax{\neg} \Id{false})) \Ax{\wedge} \Id{true} \\
\> \`{\small{}\KW{\%}(machine\Ax{\_}3)\KW{\%}}\\
\> \Ax{\forall} \=\Id{g'} \Ax{:} \Id{Conf} \\
\> \Ax{\bullet} \=\Id{ctrl}(\Id{g'}) \Ax{=} \Id{s2} \\
\>\> \Ax{\wedge} \Id{reachable2}(\=(\=\Id{evtName\Ax{\_}inc} \\
\>\>\>\> \Ax{+} \=(\=\Id{evtName\Ax{\_}reset} \Ax{+} (\=\Id{op} \Ax{\{}\Ax{\}} \Ax{:} \Id{EvtNameSet})) \Ax{:} \\
\>\>\>\>\> \Id{EvtNameSet}) \\
\>\>\> \Ax{:} \\
\>\>\> \Id{EvtNameSet}, \\
\>\>\> \Id{g'} \Ax{:} \Id{Conf}) \\
\>\> \Ax{\Rightarrow} \=(\=(\Ax{\neg} \=\Ax{\exists} \Id{k} \Ax{:} \Id{Nat} \\
\>\>\>\>\> \Ax{\bullet} \=\Ax{\exists} \Id{g''} \Ax{:} \Id{Conf} \\
\>\>\>\>\>\> \Ax{\bullet} \=(\=\Id{trans}(\=\Id{g'}, \Id{evt\Ax{\_}inc}(\Id{k}), \Id{g''}) \Ax{\wedge} (\=\Id{true} \Ax{\wedge} \Ax{\neg} \Id{false})) \Ax{\wedge} \Id{true}) \\
\>\>\>\> \Ax{\wedge} \Ax{\neg} \=\Ax{\exists} \Id{g''} \Ax{:} \Id{Conf} \\
\>\>\>\>\> \Ax{\bullet} \=(\=\Id{trans}(\=\Id{g'}, \Id{evt\Ax{\_}reset}, \Id{g''}) \\
\>\>\>\>\>\>\> \Ax{\wedge} (\=\Id{true} \\
\>\>\>\>\>\>\>\> \Ax{\wedge} \Ax{\neg} (\=\Id{cnt}(\Id{g'}) \Ax{==} (\=\Id{op} \Ax{4} \Ax{:} \Id{Nat}) \\
\>\>\>\>\>\>\>\>\> \Ax{\wedge} \=\Id{cnt}(\Id{g''}) \Ax{==} (\=\Id{op} \Ax{0} \Ax{:} \Id{Nat})))) \\
\>\>\>\>\>\> \Ax{\wedge} \Id{true}) \\
\>\>\> \Ax{\wedge} \Ax{\neg} \=\Ax{\exists} \Id{g''} \Ax{:} \Id{Conf} \\
\>\>\>\> \Ax{\bullet} \=(\=\Id{trans}(\=\Id{g'}, \Id{evt\Ax{\_}reset}, \Id{g''}) \\
\>\>\>\>\>\> \Ax{\wedge} (\=(\=\Id{cnt}(\Id{g'}) \Ax{==} (\=\Id{op} \Ax{4} \Ax{:} \Id{Nat}) \\
\>\>\>\>\>\>\>\> \Ax{\wedge} \=\Id{cnt}(\Id{g''}) \Ax{==} (\=\Id{op} \Ax{0} \Ax{:} \Id{Nat})) \\
\>\>\>\>\>\>\> \Ax{\wedge} \Ax{\neg} \Id{false})) \\
\>\>\>\>\> \Ax{\wedge} \Ax{\neg} \=\Id{s1} \Ax{=} \Id{ctrl}(\Id{g''}) \\
\> \`{\small{}\KW{\%}(machine\Ax{\_}4)\KW{\%}}\\
\> \Ax{\bullet} \Ax{\neg} \=\Ax{\forall} \Id{g'} \Ax{:} \Id{Conf} \\
\>\> \Ax{\bullet} \=\Id{ctrl}(\Id{g'}) \Ax{=} \Id{s1} \\
\>\>\> \Ax{\wedge} \Id{reachable2}(\=(\=\Id{evtName\Ax{\_}inc} \\
\>\>\>\>\> \Ax{+} \=(\=\Id{evtName\Ax{\_}reset} \\
\>\>\>\>\>\>\> \Ax{+} (\=\Id{op} \Ax{\{}\Ax{\}} \Ax{:} \Id{EvtNameSet})) \\
\>\>\>\>\>\> \Ax{:} \\
\>\>\>\>\>\> \Id{EvtNameSet}) \\
\>\>\>\> \Ax{:} \\
\>\>\>\> \Id{EvtNameSet}, \\
\>\>\>\> \Id{g'} \Ax{:} \Id{Conf}) \\
\>\>\> \Ax{\Rightarrow} \=\Id{s2} \Ax{=} \Id{ctrl}(\Id{g'}) \\
\> \`{\small{}\KW{\%}(machine\Ax{\_}5)\KW{\%}}\\
\> \Ax{\bullet} \Ax{\neg} \=\Ax{\forall} \Id{g'} \Ax{:} \Id{Conf} \\
\>\> \Ax{\bullet} \=\Id{ctrl}(\Id{g'}) \Ax{=} \Id{s2} \\
\>\>\> \Ax{\wedge} \Id{reachable2}(\=(\=\Id{evtName\Ax{\_}inc} \\
\>\>\>\>\> \Ax{+} \=(\=\Id{evtName\Ax{\_}reset} \\
\>\>\>\>\>\>\> \Ax{+} (\=\Id{op} \Ax{\{}\Ax{\}} \Ax{:} \Id{EvtNameSet})) \\
\>\>\>\>\>\> \Ax{:} \\
\>\>\>\>\>\> \Id{EvtNameSet}) \\
\>\>\>\> \Ax{:} \\
\>\>\>\> \Id{EvtNameSet}, \\
\>\>\>\> \Id{g'} \Ax{:} \Id{Conf}) \\
\>\>\> \Ax{\Rightarrow} \=\Id{s1} \Ax{=} \Id{ctrl}(\Id{g'}) \\
\> \`{\small{}\KW{\%}(machine\Ax{\_}6)\KW{\%}}\\
\> \Ax{\forall} \=\Id{k} \Ax{:} \Id{Nat} \\
\> \Ax{\bullet} \=\Id{evtName}(\Id{evt\Ax{\_}inc}(\Id{k})) \Ax{=} \Id{evtName\Ax{\_}inc} \`{\small{}\KW{\%}(evtEqs)\KW{\%}}\\
\> \Ax{\bullet} \=\Id{evtName}(\Id{evt\Ax{\_}reset}) \Ax{=} \Id{evtName\Ax{\_}reset} \`{\small{}\KW{\%}(evtEqs\Ax{\_}7)\KW{\%}}\\
\> \Ax{\forall} \Id{e} \Ax{:} \Id{Evt}; \=\Id{k}, \Id{l} \Ax{:} \Id{Nat} \\
\> \Ax{\bullet} \=\Id{e} \Ax{=} \Id{evt\Ax{\_}reset} \Ax{\vee} \=\Ax{\exists} \Id{k} \Ax{:} \Id{Nat} \Ax{\bullet} \=\Id{e} \Ax{=} \Id{evt\Ax{\_}inc}(\Id{k})\\
\> \Ax{\bullet} \Ax{\neg} \=\Id{evt\Ax{\_}reset} \Ax{=} \Id{evt\Ax{\_}inc}(\Id{k})\\
\> \Ax{\bullet} \=\Id{evt\Ax{\_}inc}(\Id{k}) \Ax{=} \Id{evt\Ax{\_}inc}(\Id{l}) \Ax{\Leftrightarrow} \=\Id{k} \Ax{=} \Id{l}\\
\> \Ax{\forall} \=\Id{g} \Ax{:} \Id{Conf} \Ax{\bullet} \=\Ax{\exists} \Id{c} \Ax{:} \Id{Ctrl}; \Id{k} \Ax{:} \Id{Nat} \Ax{\bullet} \=\Id{g} \Ax{=} \Id{conf}(\=\Id{c}, \Id{k})\\
\> \Ax{\forall} \Id{c}, \Id{d} \Ax{:} \Id{Ctrl}; \=\Id{k}, \Id{l} \Ax{:} \Id{Nat} \Ax{\bullet} \=\Id{conf}(\=\Id{c}, \Id{k}) \Ax{=} \Id{conf}(\=\Id{d}, \Id{l}) \Ax{\Leftrightarrow} \=\Id{c} \Ax{=} \Id{d} \Ax{\wedge} \=\Id{k} \Ax{=} \Id{l}\\
\KW{end}\\
\\
\SPEC \=\SIdIndex{Prove} \Ax{=}\\
\> \SId{Trans}\\
\THEN \=\OPS \=\Id{s1}, \Id{s2} \Ax{:} \Id{Ctrl}\\
\> \Ax{\bullet} \Ax{\neg} \=\Id{s1} \Ax{=} \Id{s2}\\
\THEN \=\OP \=\Id{allEvts} \Ax{:} \Id{EvtNameSet} \\
\>\> \Ax{=} \=\Id{evtName\Ax{\_}inc} \Ax{+} \=\Id{evtName\Ax{\_}reset} \Ax{+} \=\Ax{\{}\Ax{\}}\\
\> \Ax{\bullet} \=\Ax{\exists} \Id{g} \Ax{:} \Id{Conf} \Ax{\bullet} \Id{init}(\Id{g})\\
\> \Ax{\forall} \=\Id{g1}, \Id{g2} \Ax{:} \Id{Conf} \Ax{\bullet} \=\Id{init}(\Id{g1}) \Ax{\wedge} \Id{init}(\Id{g2}) \Ax{\Rightarrow} \=\Id{ctrl}(\Id{g1}) \Ax{=} \Id{ctrl}(\Id{g2})\\
\> \PRED \Id{invar}(\=\Id{g} \Ax{:} \Id{Conf}) \\
\> \Ax{\Leftrightarrow} \=(\=\Id{ctrl}(\Id{g}) \Ax{=} \Id{s1} \Ax{\wedge} \=\Id{cnt}(\Id{g}) \Ax{\leq} \Ax{4}) \\
\>\> \Ax{\vee} (\=\Id{ctrl}(\Id{g}) \Ax{=} \Id{s2} \Ax{\wedge} \=\Id{cnt}(\Id{g}) \Ax{\leq} \Ax{4});\\
\> \\
\> {\small{}\KW{\%\%} induction scheme for ''reachable'' predicate\Ax{,} instantiated for invar}\\
\> \Ax{\forall} \=\Id{es} \Ax{:} \Id{EvtNameSet} \\
\> \Ax{\bullet} \=(\=(\=\Ax{\forall} \Id{g} \Ax{:} \Id{Conf} \Ax{\bullet} \=\Id{init}(\Id{g}) \Ax{\Rightarrow} \Id{invar}(\Id{g})) \\
\>\>\> \Ax{\wedge} \=\Ax{\forall} \Id{g}, \Id{g'} \Ax{:} \Id{Conf}; \Id{e} \Ax{:} \Id{Evt} \\
\>\>\>\> \Ax{\bullet} \=(\=\Id{reachable2}(\=\Id{es}, \Id{g}) \Ax{\Rightarrow} \Id{invar}(\Id{g})) \Ax{\wedge} \Id{reachable2}(\=\Id{es}, \Id{g}) \\
\>\>\>\>\> \Ax{\wedge} \Id{trans}(\=\Id{g}, \Id{e}, \Id{g'}) \\
\>\>\>\>\> \Ax{\Rightarrow} \Id{invar}(\Id{g'})) \\
\>\> \Ax{\Rightarrow} \=\Ax{\forall} \Id{g} \Ax{:} \Id{Conf} \Ax{\bullet} \=\Id{reachable2}(\=\Id{es}, \Id{g}) \Ax{\Rightarrow} \Id{invar}(\Id{g}) \\
\> \`{\small{}\KW{\%}(InvarIsReachableInd)\KW{\%}}\\
\THEN \={\small{}\KW{\%}\KW{implies}}\\
\> \Ax{\forall} \Id{m}, \Id{n} \Ax{:} \Id{Nat}; \=\Id{g} \Ax{:} \Id{Conf} \\
\> \Ax{\bullet} \=\Id{m} \Ax{=} \Id{n} \Ax{\Leftrightarrow} \=\Id{m} \Ax{==} \Id{n} \`{\small{}\KW{\%}(thm\Ax{\_}eq)\KW{\%}}\\
\> \Ax{\bullet} \=\Id{init}(\Id{g}) \Ax{\Rightarrow} \=\Id{ctrl}(\Id{g}) \Ax{=} \Id{s1} \`{\small{}\KW{\%}(thm\Ax{\_}init\Ax{\_}ctrl)\KW{\%}}\\
\> \Ax{\bullet} \=\Id{init}(\Id{g}) \Ax{\Rightarrow} \=\Id{cnt}(\Id{g}) \Ax{=} \Ax{0} \`{\small{}\KW{\%}(thm\Ax{\_}init\Ax{\_}cnt)\KW{\%}}\\
\> \Ax{\bullet} \=\Ax{0} \Ax{*} \Id{m} \Ax{+} \Id{n} \Ax{=} \Id{n} \`{\small{}\KW{\%}(thm\Ax{\_}0at0)\KW{\%}}\\
\> \Ax{\bullet} \=\Ax{0} \Ax{*} \Ax{9} \Ax{+} \Id{n} \Ax{=} \Id{n} \`{\small{}\KW{\%}(thm\Ax{\_}0at0)\KW{\%}}\\
\> \Ax{\bullet} \=\Ax{0} \Ax{@@} \Id{n} \Ax{=} \Id{n} \`{\small{}\KW{\%}(thm\Ax{\_}0at0)\KW{\%}}\\
\> \\
\> {\small{}\KW{\%\%} case distinction lemmas\Ax{,} can be generated algorithmically}\\
\> \Ax{\forall} \Id{g}, \Id{g'} \Ax{:} \Id{Conf}; \Id{e} \Ax{:} \Id{Evt}; \=\Id{k} \Ax{:} \Id{Nat} \\
\> \Ax{\bullet} \=\Id{init}(\Id{g}) \Ax{\Rightarrow} \Id{invar}(\Id{g}) \`{\small{}\KW{\%}(InvarInit)\KW{\%}}\\
\> \Ax{\bullet} \=(\=\Id{reachable2}(\=\Id{allEvts}, \Id{g}) \Ax{\Rightarrow} \Id{invar}(\Id{g})) \Ax{\wedge} \Id{reachable2}(\=\Id{allEvts}, \Id{g}) \\
\>\> \Ax{\wedge} \Id{trans}(\=\Id{g}, \Id{e}, \Id{g'}) \Ax{\wedge} \=\Id{e} \Ax{=} \Id{evt\Ax{\_}reset} \\
\>\> \Ax{\Rightarrow} \Id{invar}(\Id{g'}) \\
\> \`{\small{}\KW{\%}(InvarReset)\KW{\%}}\\
\> \Ax{\bullet} \=(\=\Id{reachable2}(\=\Id{allEvts}, \Id{g}) \Ax{\Rightarrow} \Id{invar}(\Id{g})) \Ax{\wedge} \Id{reachable2}(\=\Id{allEvts}, \Id{g}) \\
\>\> \Ax{\wedge} \Id{trans}(\=\Id{g}, \Id{e}, \Id{g'}) \Ax{\wedge} \=\Id{e} \Ax{=} \Id{evt\Ax{\_}inc}(\Id{k}) \\
\>\> \Ax{\Rightarrow} \Id{invar}(\Id{g'}) \\
\> \`{\small{}\KW{\%}(InvarInc)\KW{\%}}\\
\> \Ax{\bullet} \=(\=\Id{reachable2}(\=\Id{allEvts}, \Id{g}) \Ax{\Rightarrow} \Id{invar}(\Id{g})) \Ax{\wedge} \Id{reachable2}(\=\Id{allEvts}, \Id{g}) \\
\>\> \Ax{\wedge} \Id{trans}(\=\Id{g}, \Id{e}, \Id{g'}) \Ax{\wedge} \=\Id{e} \Ax{=} \Id{evt\Ax{\_}inc}(\Id{k}) \Ax{\wedge} \=\Id{ctrl}(\Id{g}) \Ax{=} \Id{s1} \\
\>\> \Ax{\wedge} \=\Id{cnt}(\Id{g}) \Ax{+} \Id{k} \Ax{<} \Ax{4} \\
\>\> \Ax{\Rightarrow} \Id{invar}(\Id{g'}) \\
\> \`{\small{}\KW{\%}(InvarInc)\KW{\%}}\\
\> \Ax{\bullet} \=(\=\Id{reachable2}(\=\Id{allEvts}, \Id{g}) \Ax{\Rightarrow} \Id{invar}(\Id{g})) \Ax{\wedge} \Id{reachable2}(\=\Id{allEvts}, \Id{g}) \\
\>\> \Ax{\wedge} \Id{trans}(\=\Id{g}, \Id{e}, \Id{g'}) \Ax{\wedge} \=\Id{e} \Ax{=} \Id{evt\Ax{\_}inc}(\Id{k}) \Ax{\wedge} \=\Id{ctrl}(\Id{g}) \Ax{=} \Id{s1} \\
\>\> \Ax{\wedge} \=\Id{cnt}(\Id{g}) \Ax{+} \Id{k} \Ax{=} \Ax{4} \\
\>\> \Ax{\Rightarrow} \=\Id{ctrl}(\Id{g'}) \Ax{=} \Id{s2} \\
\> \`{\small{}\KW{\%}(InvarInc)\KW{\%}}\\
\> \Ax{\bullet} \=(\=\Id{reachable2}(\=\Id{allEvts}, \Id{g}) \Ax{\Rightarrow} \Id{invar}(\Id{g})) \Ax{\wedge} \Id{reachable2}(\=\Id{allEvts}, \Id{g}) \\
\>\> \Ax{\wedge} \Id{trans}(\=\Id{g}, \Id{e}, \Id{g'}) \Ax{\wedge} \=\Id{e} \Ax{=} \Id{evt\Ax{\_}inc}(\Id{k}) \Ax{\wedge} \=\Id{ctrl}(\Id{g}) \Ax{=} \Id{s1} \\
\>\> \Ax{\wedge} \=\Id{cnt}(\Id{g}) \Ax{+} \Id{k} \Ax{=} \Ax{4} \\
\>\> \Ax{\Rightarrow} \=\Id{cnt}(\Id{g'}) \Ax{=} \Ax{4} \\
\> \`{\small{}\KW{\%}(InvarInc)\KW{\%}}\\
\> \Ax{\bullet} \=(\=\Id{reachable2}(\=\Id{allEvts}, \Id{g}) \Ax{\Rightarrow} \Id{invar}(\Id{g})) \Ax{\wedge} \Id{reachable2}(\=\Id{allEvts}, \Id{g}) \\
\>\> \Ax{\wedge} \Id{trans}(\=\Id{g}, \Id{e}, \Id{g'}) \Ax{\wedge} \=\Id{e} \Ax{=} \Id{evt\Ax{\_}inc}(\Id{k}) \Ax{\wedge} \=\Id{ctrl}(\Id{g}) \Ax{=} \Id{s2} \\
\>\> \Ax{\wedge} \=\Id{cnt}(\Id{g}) \Ax{+} \Id{k} \Ax{=} \Ax{4} \\
\>\> \Ax{\Rightarrow} \Id{false} \\
\> \`{\small{}\KW{\%}(InvarInc)\KW{\%}}\\
\> \Ax{\bullet} \=(\=\Id{reachable2}(\=\Id{allEvts}, \Id{g}) \Ax{\Rightarrow} \Id{invar}(\Id{g})) \Ax{\wedge} \Id{reachable2}(\=\Id{allEvts}, \Id{g}) \\
\>\> \Ax{\wedge} \Id{trans}(\=\Id{g}, \Id{e}, \Id{g'}) \\
\>\> \Ax{\Rightarrow} \Id{invar}(\Id{g'}) \\
\> \`{\small{}\KW{\%}(InvarStep)\KW{\%}}\\
\> \Ax{\bullet} \=\Id{invar}(\Id{g}) \Ax{\Rightarrow} \=\Id{cnt}(\Id{g}) \Ax{\leq} \Ax{4} \`{\small{}\KW{\%}(InvarImpliesSafe)\KW{\%}}\\
\> \\
\> {\small{}\KW{\%\%} the safety theorem for our counter}\\
\> \Ax{\forall} \=\Id{g} \Ax{:} \Id{Conf} \=\Ax{\bullet} \=\Id{reachable2}(\=\Id{allEvts}, \Id{g}) \Ax{\Rightarrow} \=\Id{cnt}(\Id{g}) \Ax{\leq} \Ax{4} \`{\small{}\KW{\%}(Safe)\KW{\%}}\\
\KW{end}
\end{hetcasl}

\end{appendix}
\end{techreport}

\end{document}